\documentclass[11pt]{article}
\usepackage{microtype}
\usepackage{amsfonts}
\usepackage{amsmath, amssymb, amsthm}
\usepackage{dsfont}
\usepackage{parskip}
\usepackage[title]{appendix}
\usepackage{multirow}
\usepackage{multicol}
\usepackage{caption}
\usepackage{subcaption}
\usepackage{graphicx}
\usepackage{xcolor}
\usepackage{algorithm2e}
\usepackage{hyperref}
\hypersetup{colorlinks=true,
linkcolor=blue,
urlcolor=blue,
citecolor=blue}
\PassOptionsToPackage{unicode}{hyperref}
\PassOptionsToPackage{naturalnames}{hyperref}

\usepackage{natbib}

\DeclareMathOperator*{\argmin}{\arg\!\min}

\newtheorem{theorem}{Theorem}[section]

\newtheorem{condition}{Condition}[section]

\newtheorem{assumption}{Assumption}[section]
\newtheorem{definition}{Definition}[section]
\newtheorem{example}{Example}[section]

\newtheorem{lemma}{Lemma}[section]
\newtheorem{proposition}{Proposition}[section]
\theoremstyle{definition}
\newtheorem{remark}{Remark}[section]

\newtheorem{thm}{Theorem}[section]

\numberwithin{equation}{section}
\begin{document}

\title{Minimum Sliced Distance Estimation in a Class of Nonregular Econometric Models 
\thanks{This paper is a substantially revised version of part of Chapter 1 of Park’s dissertation (\cite{Park2022_dissertation}). We thank Debopam Bhattacharya, Tetsuya Kaji, Tong Li, Ruixuan Liu, and participants of several conferences/seminars for helpful discussions.}}
\author{Yanqin Fan\thanks{%
Department of Economics, University of Washington, Seattle, WA 98195, USA;
email: \texttt{fany88@uw.edu}.} \ and Hyeonseok Park \thanks{\textit{Corresponding author}, 
Center for Industrial and Business Organization and Institute for Advanced Economic Research, Dongbei University of Finance and Economics, China;
email: \texttt{hynskpark21@dufe.edu.cn}}}
\date{\today }
\maketitle

\begin{abstract}
This paper proposes minimum sliced distance estimation in structural econometric models with possibly parameter-dependent supports. In contrast to likelihood-based estimation, we show that under mild regularity
conditions, the minimum sliced distance estimator is asymptotically normally distributed leading to simple inference regardless of the presence/absence of
parameter-dependent supports. We illustrate the performance of our estimator on an auction model. 

\textbf{Keywords}: Auction Model; Asymptotic Normality; Parameter-Dependent Support; Sliced Cr\'amer Distance; Sliced Wasserstein Distance.

\textbf{JEL Codes:} C1; C31; C34
\end{abstract}

\newpage

\section{Introduction}

Classical likelihood-based estimation and inference may face challenges in structural econometrics models such as auction models (see e.g., \cite{Paarsch1992} and \cite%
{Donald2002}) and equilibrium job-search models (see e.g.,\cite{Bowlus2001}). In such models, the density or conditional density function of the \textit{observable} variable of interest may have a jump at a parameter-dependent boundary such as the
one-sided and two-sided models studied in \cite{Chernozhukov2004}. In the one-sided model in \cite{Chernozhukov2004} and \cite{Hirano2003},
the conditional density function of the dependent variable is assumed 
to be strictly bounded away from zero at the parameter-dependent boundary; 
In the two-sided models in \cite{Chernozhukov2004}, 
the conditional density function is assumed to have a jump at the parameter-dependent boundary, 
see Example \ref{example:parameter-dependent-support} in the next section. \cite{Chernozhukov2004} develop non-normal asymptotic theory for MLE and
Bayes estimation (BE) for both one-sided and two-sided models. \cite{Hirano2003} study efficiency considerations
according to the local asymptotic minimax criterion for conventional loss
functions in one-sided models.\footnote{\cite{Chernozhukov2004} shows that Bayesian credible intervals based on posterior quantiles, which are computationally attractive, are valid in large samples and perform well in small samples. } 

Besides the complexity of the non-normal asymptotic distribution of MLE established in \cite{Chernozhukov2004}, inference is further complicated by the dichotomy of the asymptotic theory: MLE is asymptotically normal when the conditional density function of the dependent variable has no jump and non-normal otherwise. To simplify inference in such models, \cite{Li2010} proposes a two-step estimator based on indirect inference. In the first step, an auxiliary linear regression model is estimated from both real data and simulated data drawn from the structural model; in the second step, a minimum distance between estimators of the parameters in the auxiliary model based on real data and those based on synthetic data is adopted to estimate the structural parameter of interest. \textit{Assuming that the auxiliary model ensures identification of the structural parameter},  \cite{Li2010} shows that under standard regularity conditions, the indirect inference estimator is asymptotically normally distributed and inference is simple and standard.

The motivation for this paper is the same as that of \cite{Li2010}. Unlike \cite{Li2010}, we develop a class of (one-step) \textit{sliced} minimum distance estimators of the structural parameter of interest without relying on an auxiliary regression model. Each estimator in this class makes use of a \textit{sliced} $L_{2}$-distance between some \textquotedblleft empirical measure\textquotedblright\ of the observed data and a
parametric or semiparametric measure induced by the
structural model. We refer to the resulting estimator as a \textit{minimum sliced distance} (MSD) estimator. Two prominent examples of sliced $L_{2}$-distances are sliced $2$-Wasserstein distance and sliced Cr\'amer distance\footnote{The sliced Cr\'amer distance is used in \cite{Zhu1997} for goodness-of-fit testing.} based on which we construct \textit{minimum sliced Wasserstein distance} (MSWD) estimator and  \textit{minimum sliced Cr\'amer distance} (MSCD) estimator respectively. Wasserstein
distances have recently been used in the statistics and machine learning
literatures for testing and estimation of generative models. For example, 
\cite{Bernton2019} proposes minimum WD estimator based on the $1$%
-Wasserstein distance and establishes non-normal asymptotic distribution for \textit{univariate} unconditional parametric distributions; \cite{Nadjahi2020AsymSWD} proposes a sliced
version of the estimator in \cite{Bernton2019} and extends the non-normal
asymptotic distribution in \cite{Bernton2019}.\footnote{The technical analysis of 
\cite{Bernton2019} (and \cite{Nadjahi2020AsymSWD}) relies critically on the
equivalence between the $1$-Wasserstein distance and the $1$-Cr\'amer distance
for univariate distributions. Such an equivalence no longer holds for the $2$%
-Wasserstein distance adopted in the current paper and the analysis in \cite{Bernton2019} breaks down.}

This paper makes several contributions. First, we establish consistency and
asymptotic normality of the MSD estimator under a set of high-level assumptions applicable to a wide range of models and sliced distances. In contrast to likelihood-based inference, inference using our
MSD estimator is standard. Second, we verify the high-level assumptions for the MSCD estimator under primitive conditions for conditional models including the one-sided and two-sided models in \cite%
{Chernozhukov2004} and \cite{Hirano2003}. For the latter, our primitive conditions imply
that the MSCD estimator is asymptotically normally distributed regardless of the presence/absence of 
 a jump in the conditional density function. Third, we confirm the accuracy of the asymptotic normal distribution of MSCD estimator in a simulation experiment based on 
the auction model in \cite{Paarsch1992} and \cite{Donald2002} and compare it with the indirect inference estimator of \cite{Li2010}.

The rest of the paper is organized as follows.
Section \ref{sec:general-setup} introduces our general framework and the motivating
example, proposes the class of MSD estimators including MSWD and MSCD estimators, and presents an illustrative example comparing the different distributions of MLE, MSWD, and MSCD estimators for a one-sided and a two-sided uniform models.
Section \ref{section:asym-theory} establishes the consistency and asymptotic normality of the MSD
estimator under high level assumptions. Section \ref{section:SC-conditional} verifies the high level
assumptions for the MSCD estimator in the conditional model including one-sided parameter-dependent support under primitive conditions. Section 
\ref{sec:numerical-results} presents numerical results using
synthetic data generated from the auction model in \cite{Paarsch1992}, \cite{Donald2002}, and \cite{Li2010}. Section \ref{sec:conclusion} concludes. A series
of appendices contains verification of assumptions for both MSCD and MSWD estimators in one-sided and two-sided uniform models, verification of high level assumptions for MSCD in two-sided parameter-dependent support models, and technical proofs of the main results in the paper.

\section{The General Set-up and Minimum Sliced Distance Estimation} \label{sec:general-setup}

\subsection{The General Set-up and Motivating Examples}

Let $\left\{ Z_{t}\right\}_{t=1}^T $ denote a random sample satisfying Assumption \ref{assumption:DGP} below. 
\begin{assumption} \label{assumption:DGP}
$\left\{ Z_{t}\right\}_{t=1}^T$ is
a random sample from either an \textit{unconditional model} characterized by a parametric distribution function $F\left(\cdot ;\psi _{0}\right) $ of $Z_{t}$ or a \textit{conditional model} characterized by a parametric conditional distribution function $F\left( \cdot |x,\psi _{0}\right)$ of $Y_{t}$ given $X_{t}=x\in\mathcal{X}$, where $Z_{t}=\left( Y_{t}^{\top
},X_{t}^{\top }\right) ^{\top }$ in the conditional model and $\psi_0 \in \Psi \subset \mathbb{R}^{d_{\psi}}$ for a finite integer $d_{\psi}$.
\end{assumption}
We note that in the conditional model, \textit{\ the distribution of }$X_{t}$\textit{\ is unspecified}. 
In both models, we are interested in the estimation and inference for $\psi _{0}$.

When the density function of either $F\left( \cdot ;\psi \right) $ or $%
F\left( \cdot |x,\psi \right) $ exists, MLE or the conditional MLE is a
popular approach to estimating the unknown parameter $\psi _{0}$. However, the
asymptotic distribution of MLE or conditional MLE and the associated
inference depend critically on  model assumptions. The classical Wald,
QLR, and score tests rely on smoothness assumptions on the density function
and the assumption that the true parameter $\psi _{0}$ is in the interior
of $\Psi$. Many important structural models in economics violate one or
more assumptions underlying the classical likelihood theory which motivates the
development of alternative methods of estimation and inference such as those
discussed in Section 1.

Below we present two examples of the conditional model stated in Assumption 2.1. Example \ref{example:parameter-dependent-support} includes the one-sided models in \cite{Chernozhukov2004} and \cite{Hirano2003} and two-sided models in \cite{Chernozhukov2004} for which likelihood-based estimation and inference are difficult to implement. We refer to both as \textit{parameter-dependent support models} and use them to illustrate our assumptions/results in subsequent sections. Example \ref{example:auction} is the independent private value procurement auction model formulated in \cite{Paarsch1992} and \cite{Donald2002} for which the winning bid follows the one-sided model in Example \ref{example:parameter-dependent-support}.

\begin{example}[Parameter-Dependent Support Models]
\label{example:parameter-dependent-support}	
	
A scalar random variable $Y$ given a vector of covariates $X$ follows
\begin{align*}
	Y = g(X, \theta_0) + \epsilon,
\end{align*}
where the conditional density function of the error term $\epsilon$ given $X$ is  $f_{\epsilon}(\epsilon|X, \theta_0, \gamma_0)$, and $\theta_0 \in \Theta \subset \mathbb{R}^{d_{\theta }}$ and $\gamma_0 \in
\Gamma \subset \mathbb{R}^{d_{\gamma }}$ are finite-dimensional parameters. 

In one-sided models, $f_{\epsilon}(\epsilon|X, \theta, \gamma) = 0$ for all $\epsilon \le 0$. 
Let $\mathcal{X}\subset R^{d_{x}}$ denote the support of $X$. \cite%
{Hirano2003} assume that \textit{for }$X$\textit{\ in some subset of }$%
\mathcal{X}$\textit{\ with positive probability, the conditional density of }
$\epsilon$ \textit{\ at its support boundary } $\epsilon = 0$ \textit{\
is strictly positive}. 

In two-sided models, 
\begin{equation*}
f_{\epsilon}(\epsilon|X,\theta,\gamma ):=%
\begin{cases}
f_{L, \epsilon}(\epsilon|X,\theta,\gamma ) & \text{ if }\epsilon < 0, \\ 
f_{U, \epsilon}(\epsilon|X,\theta,\gamma) & \text{ if } \epsilon \ge 0,%
\end{cases}%
\end{equation*}%
where as in \cite{Chernozhukov2004}, it holds that \textit{for any }$x\in \mathcal{X}$, 
\begin{align}
\lim_{\epsilon\uparrow 0 } f_{\epsilon}(\epsilon|x,\theta ,\gamma)
&= f_{L, \epsilon}(0|x,\theta,\gamma)\text{, }%
\lim_{\epsilon\downarrow 0} f_{\epsilon}(\epsilon|x,\theta,\gamma)=f_{U, \epsilon}(0 |x,\theta,\gamma ),  \notag \\
f_{U, \epsilon}(0|x, \theta,\gamma ) &> f_{L, \epsilon}(0|x,\theta,\gamma)+\eta \text{ \textit{for some }}\eta >0,  \text{\textit{ for all }}\left( \theta ,\gamma \right) \in \Theta \times
\Gamma \text{.} \label{Jump}
\end{align}%

Let $F\left( y|x,\theta ,\gamma \right) $ denote the conditional cdf of $Y$
given $X=x$. For one-sided models, 
\begin{equation}
F\left( y|x,\theta ,\gamma \right) =\int_{0
}^{y - g\left( x,\theta \right)}f_{\epsilon}\left( u|x,\theta ,\gamma \right) du\text{, }y\geq g\left( x,\theta
\right) \text{;}  \label{eq:one-sided}
\end{equation}%
for two-sided models, 
\begin{equation}
F(y|x,\theta ,\gamma )=%
\begin{cases}
\int_{-\infty }^{y-g\left( x,\theta \right)}f_{L, \epsilon}(u|x,\theta ,\gamma )\mathrm{d}u, & \text{ if }y\leq
g\left( x,\theta \right) , \\ 
\int_{-\infty }^{0}f_{L, \epsilon}(u|x,\theta ,\gamma )\mathrm{%
d}u+\int_{0}^{y-g\left( x,\theta \right)}f_{U,\epsilon}(u|x,\theta ,\gamma )\mathrm{d}%
u, & \text{ if }y>g\left( x,\theta \right) .\label{eq:two-sided}%
\end{cases}%
\end{equation}%
Let $\psi =(\theta ,\gamma )^{\prime }\in \Psi =\Theta \times \Gamma $.

\end{example}

Many structural econometrics models lead to one-sided or two-sided models
satisfying the assumptions in \cite{Hirano2003} and \cite{Chernozhukov2004} such as (\ref{Jump}) in two-sided models so that the non-normal asymptotic theory they develop is applicable. On the other hand, if (\ref{Jump}) does not hold, then the asymptotic distribution of MLE is normal. 

\begin{example}
\label{example:auction}

We consider the independent private value procurement
auction model formulated in \cite{Paarsch1992} and \cite{Donald2002}. 
In the first-price procurement auction model, there is only one buyer but multiple sellers. 
Sellers provide their bids and the lowest one is the winning bid. 
Let $Y$ denote the winning bid and $X$ denote the observable
auction characteristics. Suppose the bidder's private value $V$ follows a
conditional distribution of the form 
\begin{align*}
	f_{V}(v|X, \psi_0) I(v \ge g_V(X, \psi_0)),
\end{align*}
where 
\begin{equation*}
f_{V}\left( v|X, \psi \right) = \frac{1}{h\left( X,\psi \right) }%
\exp \left( -\frac{v}{h\left( X, \psi \right) }\right) \text{ and }%
g_{V}\left( X,\psi \right) =0.
\end{equation*}%
Then $E\left( V|X,\psi \right) =h\left( X,\psi \right)$. We are interested in estimating $\psi_0$. However, we only observe the winning bid $Y$. 
Assuming a Bayes-Nash Equilibrium solution concept, the equilibrium bidding function
satisfies 
\begin{equation}
\sigma \left( v\right) =v+\frac{\int_{v}^{\infty }\left( 1-F_{V}\left( \xi
|X, \psi \right) \right) ^{m-1}d\xi }{\left( 1-F_{V}\left( v|X, \psi\right) \right) ^{m-1}%
},  \label{BidFunction}
\end{equation}%
where  $
F_{V}\left(v|X, \psi\right) $ denotes the conditional cdf of $V$ given $X$. It is easy to show that the pdf of the winning bid $Y$ given $X$ is  
\begin{align*}
	f(y|X, \psi_0)I(y \ge g(X, \psi_0))
\end{align*}
where
\begin{equation*}
f\left( y|X,\psi \right) =\frac{m}{h\left( x,\psi \right) }\exp
\left( -\frac{m}{h\left( x,\psi \right) }\left( y-\frac{h\left( x,\psi
\right) }{m-1}\right) \right) \text{ and }g\left( X,\psi \right) =\frac{%
h\left( x,\psi \right) }{m-1},
\end{equation*}
in which $m$ is the number of bidders in the auction. The conditional cdf $F(y|X, \psi )$ of the winning bid is given by 
\begin{align*}
F(y|X,\psi ) =1-\exp \left( -\frac{m}{h(x,\psi )}\left( y-\frac{h(x,\psi )}{m-1}%
\right) \right) \text{ for }y\geq g(X, \psi ).
\end{align*}
\end{example}

\subsection{Minimum Sliced Distance Estimation}

Let $\mu_0$ and $\mu(\psi)$ denote respectively the true probability measure of $Z_t$ and the probability measure induced by the parametric model with parameter $\psi\in\Psi$. A minimum sliced distance estimator of $\psi_0$ is based on an estimator of a \textit{sliced distance} between $\mu_0$ and $\mu(\psi)$ denoted as $ \widehat{\mathcal{S}}(\psi)$.
To introduce $ \widehat{\mathcal{S}}(\psi)$, we first present a brief review of two popular sliced distances:  the sliced $2$-Wasserstein distance or simply the sliced $2$-Wasserstein distance and sliced Cr\'amer distance.

\subsubsection{Sliced Wasserstein Distance and Sliced Cr\'amer Distance}

Let $\mathcal{P}_{2}(\mathcal{Z})$ denote the space of probability measures
with support $\mathcal{Z}\subset \mathbb{R}^{d}$ and finite second moments. Further let $\mathbb{S}^{d-1}=\{u\in \mathbb{R}^{d}:\lVert u\rVert _{2}=1\}$ be the unit-sphere in $\mathbb{R}^{d} $. 
For two probability measures $\mu $ and $\nu $ from $\mathcal{P}_{2}(%
\mathcal{Z})$, we denote by $\mathcal{W}_{2}\left( \mu ,\nu \right) $ their $%
2$-Wasserstein distance or simply the Wassserstein distance. It is a finite
metric on $\mathcal{P}_{2}(\mathcal{Z})$ defined by the optimal transport
problem:%
\begin{equation*}
\mathcal{W}_{2}\left( \mu ,\nu \right) =\left[ \inf_{\gamma \in \Gamma (\mu
,\nu )}\int_{\mathbb{R}^{d}}\int_{\mathbb{R}^{d}}\lVert x-y\rVert ^{2}%
\mathrm{d}\gamma (x,y)\right] ^{1/2},
\end{equation*}%
where $\Gamma (\mu ,\nu )$ is the set of probability measures on $\mathbb{R}%
^{d}\times \mathbb{R}^{d}$ with marginals $\mu $ and $\nu $.

When $d=1$, the Wasserstein distance is easy to compute. Proposition 2.17 in \cite{Santambrogio2015} or Theorem 6.0.2 in 
\cite{Ambrosio2008} implies that 
\begin{equation}
\mathcal{W}_{2}^{2}(\mu ,\nu )=\int_{0}^{1}\left( F_{\mu }^{-1}(s)-F_{\nu
}^{-1}(s)\right) ^{2}\mathrm{d}s,  \label{W2}
\end{equation}%
where $F_{\mu }(\cdot )$ and $F_{\nu}(\cdot)$ are the distribution functions associated with the
measures $\mu$ and $\nu$, respectively, and $F_{\mu }^{-1}$ and $F_{\nu}^{-1}$ are the quantile functions.

When $d>1$, the Wasserstein distance $\mathcal{W}_{2}$ is difficult to compute. The sliced Wasserstein distance is introduced
to ease the computational burden associated with the Wasserstein distance $%
\mathcal{W}_{2}$ (c.f. \cite{Bonneel2015}). For $u\in \mathbb{S}^{d-1}$ and $z\in \mathbb{R}^{d}$, let $u^{\ast
}(z)=u^{\top }z$ be the 1D (or scalar) projection of $z$ to $u$. For a
probability measure $\mu $, we denote by $u_{\sharp }^{\ast }\mu $ the
push-forward measure of $\mu $ by $u^{\ast }$. The sliced Wasserstein
distance $\mathcal{SW}(\mu ,\nu )$ is defined as follows: 
\begin{equation}
\mathcal{SW}(\mu ,\nu )=\left[ \int_{\mathbb{S}^{d-1}}\mathcal{W}^{2}(u_{\#}^{\ast
}\mu ,u_{\#}^{\ast }\nu )d\varsigma (u)\right] ^{1/2},  \label{SWD}
\end{equation}%
where $\varsigma (u)$ is the uniform distribution on $\mathbb{S}^{d-1}$. It
is well known that $\mathcal{SW}$ is a well-defined metric, see \cite%
{Nadjahi2020SDProperty}. For each $u\in \mathbb{S}^{d-1}$, let 
\begin{equation*}
G_{\mu }(s;u)=\int I(u^{\top }z\leq s)\mathrm{d}F_{\mu }(z).
\end{equation*}%
Define $G_{\nu }(s;u)$ similarly. Since 
\begin{equation*}
\mathcal{W}^{2}(u_{\#}^{\ast }\mu ,u_{\#}^{\ast }\nu )=\int_{0}^{1}\left(
G_{\mu }^{-1}(s;u)-G_{\nu }^{-1}(s;u)\right) ^{2}\mathrm{d}s\text{,}
\end{equation*}%
we obtain that 
\begin{equation*}
\mathcal{SW}(\mu ,\nu )=\left[ \int_{\mathbb{S}^{d-1}}\int_{0}^{1}\left( G_{\mu
}^{-1}(s;u)-G_{\nu }^{-1}(s;u)\right) ^{2}\mathrm{d}s\mathrm{d}\varsigma (u)%
\right] ^{1/2}.
\end{equation*}%
For generality, we introduce a weighted version of $\mathcal{SW}(\mu ,\nu )$.
\begin{definition}
A weighted Sliced Wasserstein distance is defined as: 
\begin{equation}
\mathcal{SW}_{w}(\mu ,\nu )=\left[ \int_{\mathbb{S}^{d-1}}\int_{0}^{1}\left( G_{\mu
}^{-1}(s;u)-G_{\nu }^{-1}(s;u)\right) ^{2}w\left( s\right) \mathrm{d}%
sd\varsigma (u)\right] ^{1/2}, \label{Wasserstein}
\end{equation}%
where $w\left( \cdot\right) $ is a nonnegative function such that $%
\int_{0}^{1}w\left( s\right) ds=1$. 
\end{definition}

Similarly, we introduce a weighted sliced Cram\'er distance between two measures $\mu$ and $\nu$ with support $\mathcal{Z}\subset \mathbb{R}^{d}$. In the univariate case ($d=1$), the Cram\'er distance (see e.g. \cite{Cramer_1928} and \cite{Szekely_2017}) is defined as 
\begin{equation*}
\mathcal{C}_{2}^{2}(\mu ,\nu )=\int_{-\infty }^{\infty }(F_{\mu }(s)-F_{\nu
}(s))^{2}\mathrm{d}s.   
\end{equation*}
\begin{definition}
For $d>1$, we define a weighted Sliced Cram\'er distance as: 
\begin{equation}
S\mathcal{C}_{w}(\mu ,\nu )=\left( \int_{\mathbb{S}^{d-1}}\int_{-\infty
}^{\infty }(G_{\mu }(s;u)-G_{\nu }(s;u))^{2}w(s)\mathrm{d}sd\varsigma
(u)\right) ^{1/2}, \label{Cramer}
\end{equation}
where $w(\cdot)$ is a nonnegative function such that $%
\int_{0}^{1}w\left( s\right) ds=1$.
\end{definition}

\subsubsection{MSD, MSWD, and MSCD Estimators}

For $\psi \in \Psi \subset \mathbb{R}^{d_{\psi }}$, the sliced distance $\widehat{\mathcal{S}}(\psi)$ is defined as 
\begin{equation}
\widehat{\mathcal{S}}(\psi):=\int_{\mathbb{S}^{d-1}}\int_{\mathcal{S}}(Q_{T}(s;u)-%
\widehat{Q}_{T}(s;u,\psi ))^{2}w(s)\mathrm{d}sd\varsigma (u),
\label{SO}
\end{equation}
and the MSD estimator denoted by $\hat{\psi}_{T}$ is defined via
\begin{equation}
\widehat{\mathcal{S}}(\hat{
\psi}_T)=\inf_{\psi \in \Psi }\widehat{\mathcal{S}}(\psi)+o_{p}(T^{-1}),\label{Estimator}
\end{equation}%
where  $Q_{T}(s;u)$ is an \textquotedblleft empirical measure\textquotedblright\ of the projection of observed data $\{u^{\top}Z_t\}_{t=1}^{T}$ and  $\widehat{Q}_{T}(\cdot ;u,\psi )$ is the corresponding parametric or semiparametric measure induced by the structural model.
 
Two examples of $\widehat{\mathcal{S}}(\psi)$ we focus on are estimators of 
 $\mathcal{SW}^2_{w}(\mu_0 ,\mu(\psi) )$ and $\mathcal{SC}^2_{w}(\mu_0 ,\mu(\psi) )$, where
 \begin{itemize}
     \item for $\mathcal{SW}^2_{w}(\mu_0,\mu(\psi))$, $Q_{T}(s;u)$ is the empirical quantile function of $\left\{ u^{\top
}Z_{t}\right\} _{t=1}^{T}$ and $\widehat{Q}_{T}(\cdot ;u,\psi )$ is (an estimator of) a model induced quantile function;
\item  for $\mathcal{SC}^2_{w}(\mu_0, \mu(\psi))$,  $Q_{T}(s;u)$ is the empirical
distribution function of $\left\{ u^{\top }Z_{t}\right\} _{t=1}^{T}$ and $%
\widehat{Q}_{T}(\cdot ;u,\psi )$ is (an estimator of) a model induced distribution function.
 \end{itemize}

The form of $\widehat{Q}_{T}(\cdot ;u,\psi )$ differs for
unconditional and conditional models.
\begin{itemize}
    \item For \textit{unconditional models}, $\widehat{Q}_{T}(\cdot
;u,\psi )$ is deterministic denoted as $Q(\cdot ;u,\psi )$. It is the 
parametric quantile function of $u^{\top }Z_{t}$ induced by the parametric
distribution of $Z_{t}$ for $\mathcal{SW}^2_{w}(\mu_0,\mu(\psi))$ and the parametric distribution function of $%
u^{\top }Z_{t}$ for $\mathcal{SC}^2_{w}(\mu_0, \mu(\psi))$;
\item   
For \textit{conditional models}, $\widehat{Q}_{T}(\cdot ;u,\psi )$ is random. It is $\widehat{G}_{T}(\cdot ;u,\psi )$ for $\mathcal{SC}^2_{w}(\mu_0, \mu(\psi))$ and $\widehat{Q}%
	_{T}(\cdot ;u,\psi )=\widehat{G}_{T}^{-1}(\cdot ;u,\psi )$ for $\mathcal{SW}^2_{w}(\mu_0,\mu(\psi))$, where
 	\begin{align}
		\widehat{G}_{T}(s;u,\psi )& =\frac{1}{T}\sum_{t=1}^{T}\int_{-\infty
		}^{\infty }I(u_{1}y+u_{2}^{\top }X_{t}\leq s)f(y|X_{t},\psi )\mathrm{d}y \nonumber \\
		& =%
		\begin{cases}
			\frac{1}{T}\sum_{t=1}^{T}F(u_{1}^{-1}(s-u_{2}^{\top }X_{t})|X_{t},\psi ) & 
			\text{ if }u_{1}>0 \\ 
			\frac{1}{T}\sum_{t=1}^{T}I(u_{2}^{\top }X_{t}\leq s) & \text{ if }u_{1}=0 \\ 
			1-\frac{1}{T}\sum_{t=1}^{T}F(u_{1}^{-1}(s-u_{2}^{\top }X_{t})|X_{t},\psi ) & 
			\text{ if }u_{1}<0.%
		\end{cases}%
  \label{GT}
	\end{align}
 \end{itemize}
 
 To motivate $\widehat{G}_{T}(s;u,\psi )$ for conditional models,  let $Z=\left( Y,X^{\top }\right) ^{\top }$. We note that the cdf of $Z$ and $u^{\top }Z$ are given by 
	\begin{equation*}
		F\left( z;\psi \right) =\mathbb{E}\left[ F\left( y|X,\psi \right) I\left(
		X\leq x\right) \right] \text{ and }
	\end{equation*}%
	\begin{eqnarray}
		G(s;u,\psi ) &=&\Pr (u_{1}Y+u_{2}^{\top }X\leq s) \nonumber \\
		&=&\mathbb{E}\left[ \int_{-\infty }^{\infty }I(u_{1}y+u_{2}^{\top }X\leq
		s)f(y|X,\psi )\mathrm{d}y\right] \nonumber \\
		&=&%
		\begin{cases}
			\mathbb{E}\left[ F(u_{1}^{-1}(s-u_{2}^{\top }X)|X,\psi )\right] & \text{ if }%
			u_{1}>0 \\ 
			\mathbb{E}\left[ I(u_{2}^{\top }X\leq s)\right] & \text{ if }u_{1}=0 \\ 
			1-\mathbb{E}\left[ F(u_{1}^{-1}(s-u_{2}^{\top }X)|X,\psi )\right] & \text{
				if }u_{1}<0.%
		\end{cases}%
  \label{G}
	\end{eqnarray}%
 Since the population cdf of $u^{\top }Z$, i.e., $G(s;u,\psi_0)$, depends on the unknown distribution of $X$, we make use of its estimator $\widehat{G}_{T}(\cdot ;u,\psi )$ to define $\widehat{Q}_T$.

\begin{definition}
When $\widehat{\mathcal{S}}(\psi)$ is a consistent estimator of $\mathcal{SW}^2_{w}(\mu_0 ,\mu(\psi) )$,  we call $\hat{\psi}_{T}$ the MSWD estimator; When $\widehat{\mathcal{S}}(\psi)$ is a consistent estimator of $\mathcal{SC}^2_{w}(\mu_0 ,\mu(\psi) )$,  $\hat{\psi}_{T}$ is referred to as MSCD 
estimator. 
\end{definition}

\subsection{A Simple Illustrative Example}

\label{sec:illustrative_example}

Before establishing the formal asymptotic theory for $\hat{
\psi}_{T}$ in the next section, we illustrate the different asymptotic distributions of MSCD estimator, MSWD estimator, and MLE in a one-sided uniform model and a two-sided uniform model in this section. 

The one-sided uniform model is $Y\sim U[0,\psi _{0}]$, where $\psi _{0}>0$ is the true parameter. The two-sided uniform model is characterized by the density function
\begin{equation*}
f(y;\psi _{0})=%
\begin{cases}
\frac{1}{4\psi _{0}} & \ \text{if }0 < y\leq \psi_{0}, \\ 
\frac{3}{4(1-\psi _{0})} & \mathbf{\ }\text{if }\psi _{0}<y\leq 1, \\
0 & \text{ otherwise. }
\end{cases}%
\end{equation*}%
This model is an example of the two-sided parameter-dependent support model for all $\psi_0\in (0,1)$ except when $\psi_0=1/4$ so the MLE of $\psi_0$ follows asymptotically a non-normal distribution when $\psi_0\neq 1/4$. When $\psi_0=1/4$, there is no jump in the density function $f(y;\psi _{0})$. With straightforward but tedious algebra, we show that the conditions for asymptotic normality of MLE in Theorem 5.39 of \cite{Vaart1998} are satisfied when $\psi_0=1/4$. 

In addition to MSCD estimator, MSWD estimator, and MLE, we also included an oracle
Generative Adversarial Network (GAN) estimator\footnote{Here, we consider the oracle GAN estimator when the size of the simulated sample is infinity
to minimize variation caused by the simulated sample.} in the comparison. The oracle GAN estimator minimizes the Jensen–Shannon divergence (see Proposition 1 of \cite{Goodfellow2014}
and Example 3 in \cite{Kaji2020}): 
\begin{equation*}
\hat{\psi}_{JS}:=\argmin_\psi \left\{\frac{1}{2T}\sum_{t=1}^{T}\log \left( \frac{2f(y_{t};\psi _{0})}{f(y_{t};\psi_{0})+f(y_{t};\psi) }\right) +\frac{1
}{2}\int \log \left( \frac{2f(y;\psi )}{f(y;\psi _{0})+f(y;\psi )}
\right) f(y;\psi )\mathrm{d}y\right\},
\end{equation*}
where $f(y_t, \psi)$ is the density function of the model-induced parametric distribution.

We present three figures below. In each figure, the left panel plots the population objective functions (divergences/distances) of the four estimators and the right panel presents the corresponding QQ plots based on 3000 normalized values of the estimator, where each value is computed from a random sample of size 1000. We note that in this section and Section \ref{sec:numerical-results}, the normalized value of an estimator $\hat{\psi}$ is defined as $(\hat{\psi} - \psi_0) / s$, where $s$ is the Monte-Carlo standard deviation of the estimator.
\begin{figure}[tbph]
\centering
\includegraphics[width=0.7%
\linewidth]{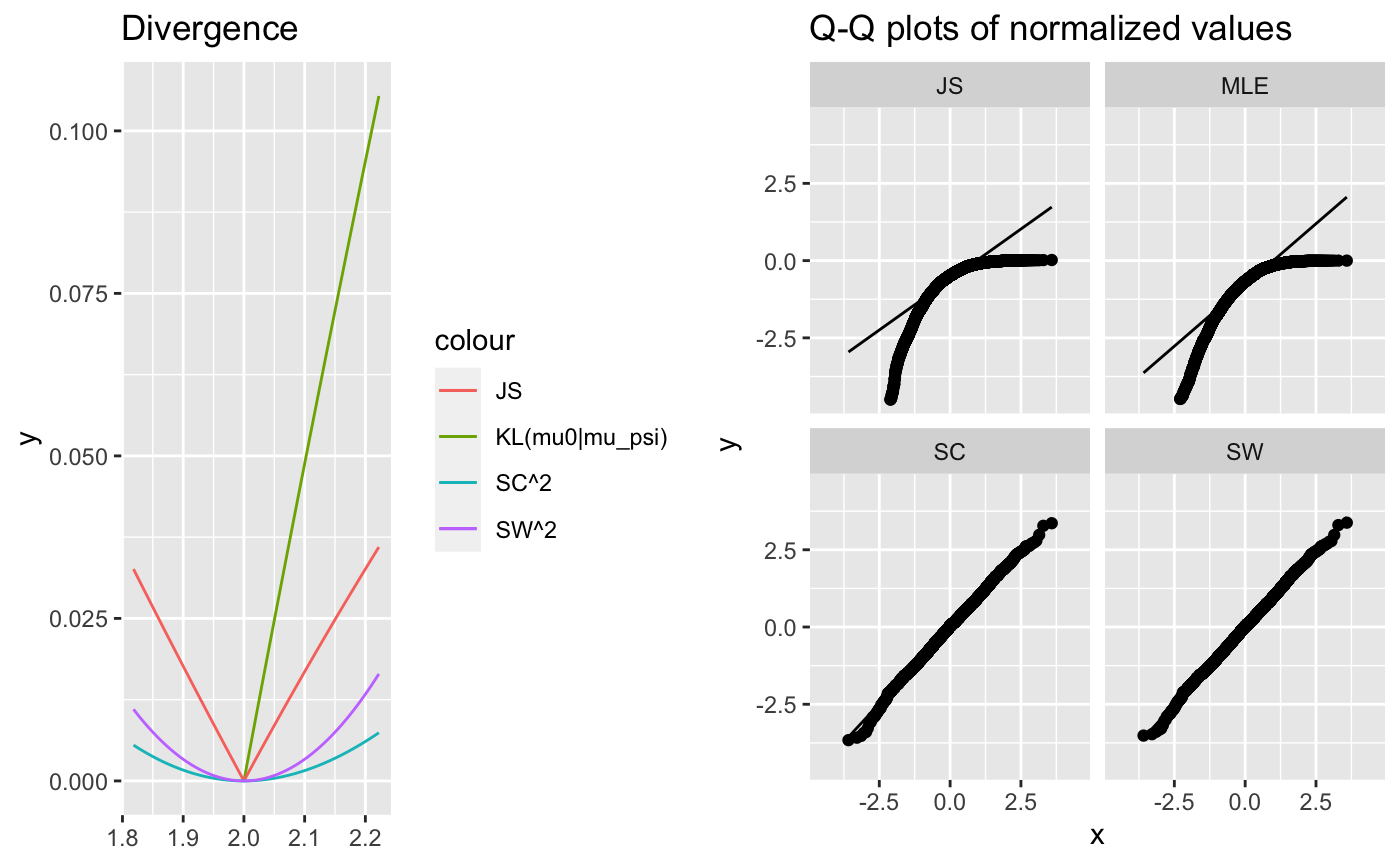}
\caption{One-sided uniform model with $\protect\psi_0 = 2$ and $T = 1000$.}
\label{fig:divergence-plot-unif-oneside-theta2}
\end{figure}

\begin{figure}[tbph]
\centering
\includegraphics[width=0.7%
\linewidth]{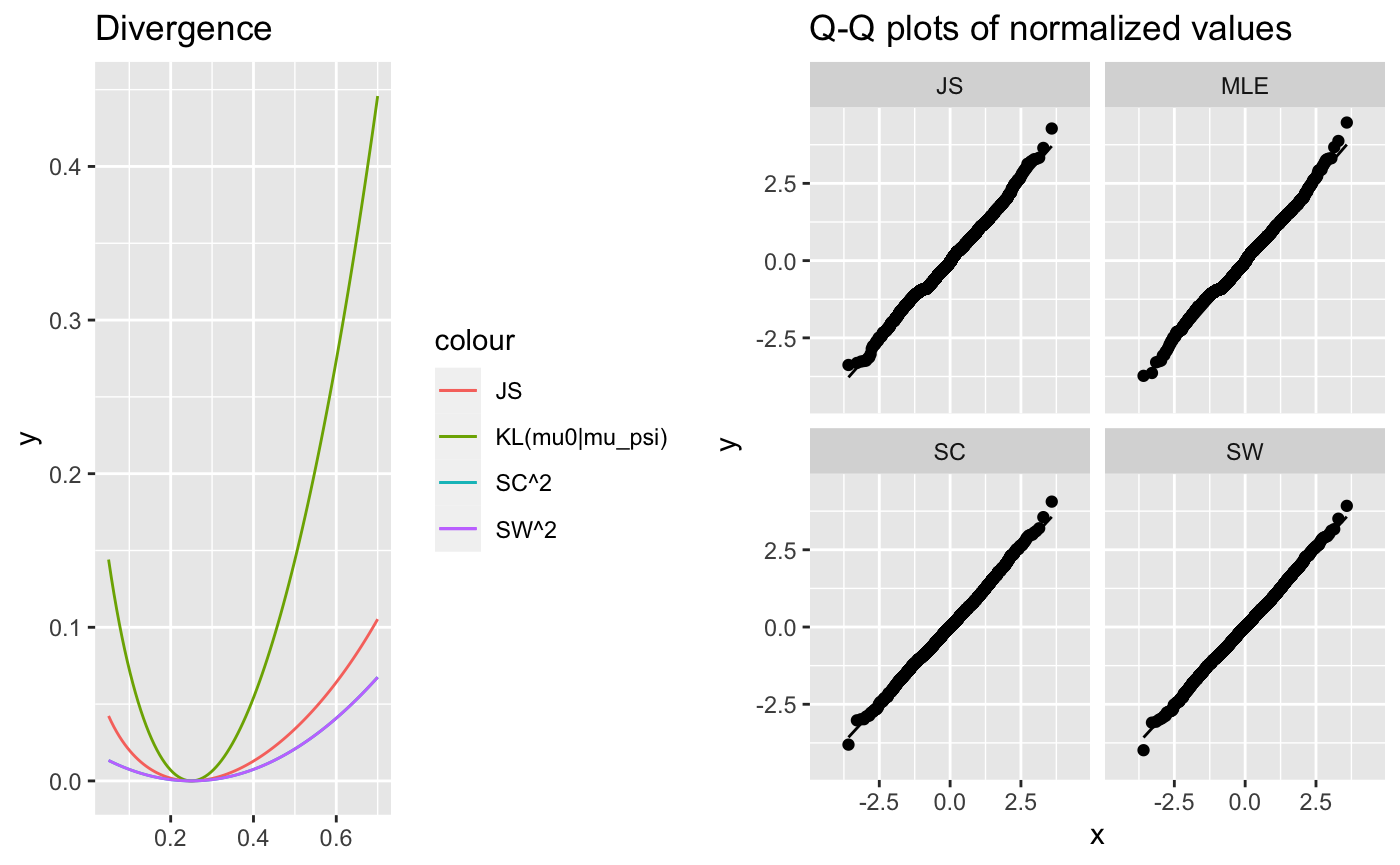}
\caption{Two-sided uniform model with $\protect\psi_0 = 1/4$
with $T = 1000$.}
\label{fig:divergence-plot-unif-twoside-025}
\end{figure}

\begin{figure}[tbph]
\centering
\includegraphics[width=0.7%
\linewidth]{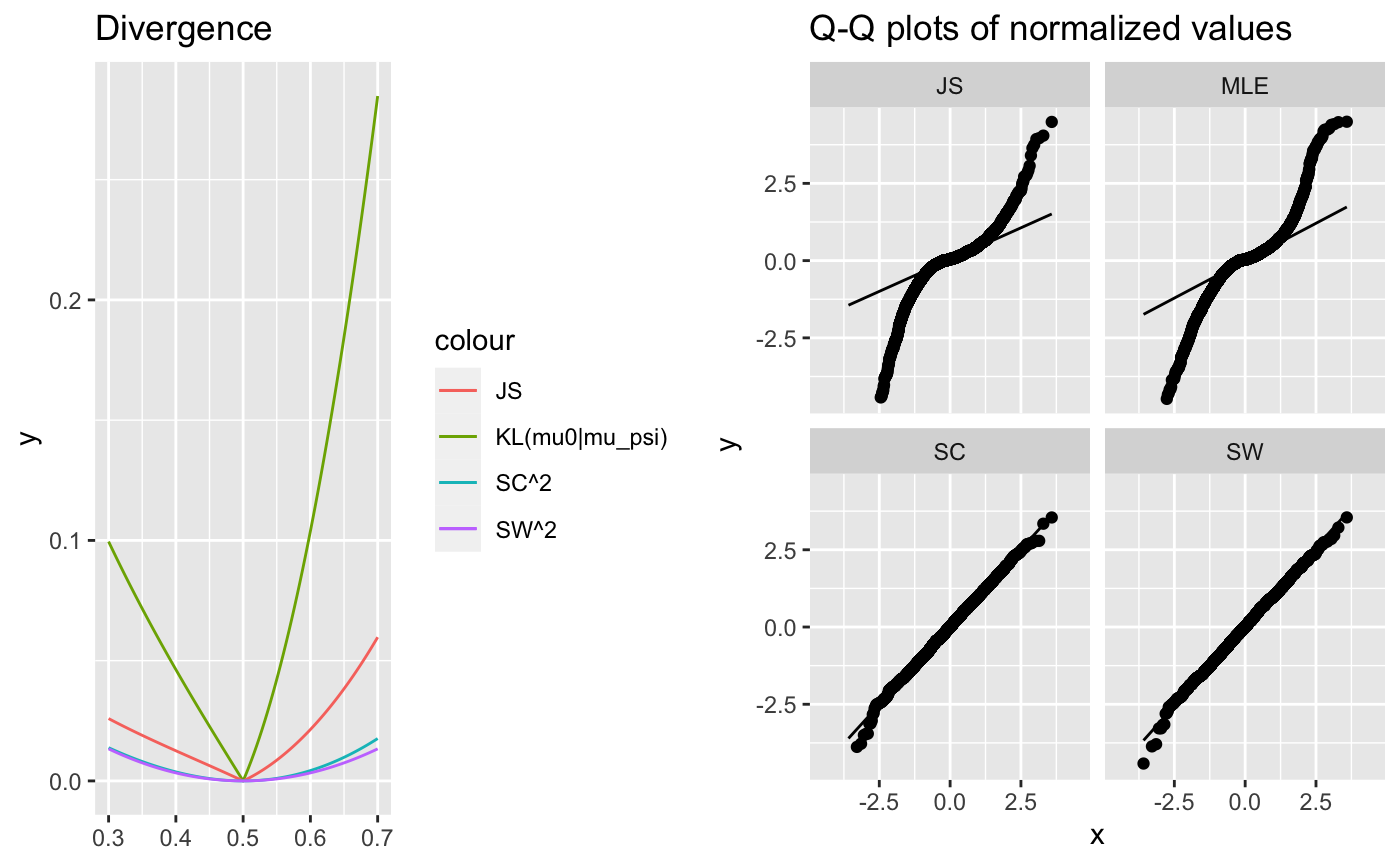}
\caption{Two-sided uniform model with $\protect\psi_0 = 1/2$
with $T = 1000$.}
\label{fig:divergence-plot-unif-twoside-050}
\end{figure}

\newpage
Several conclusions can be drawn from Figures \ref%
{fig:divergence-plot-unif-oneside-theta2}-\ref%
{fig:divergence-plot-unif-twoside-050}. First, the population objective functions for our MSWD and MSCD estimators are smooth in all cases. For the one-sided uniform model, the KL and JS divergences are not differentiable at the true parameter value: KL divergence is not defined when $\psi < \psi_0$. For the two-sided uniform model, the KL and JS divergences are 
bounded but may or may not be first-order differentiable at the true
parameter value. For example, both are not differentiable when $\psi_0 \ne 1/4$. Second, for the one-sided
uniform model, both MLE and oracle GAN estimators are non-normally distributed;
both MSWD and MSCD estimators are approximately normally distributed. Third,
for the two-sided uniform model, when $\psi _{0}=1/4,$ all four estimators
are close to being normally distributed; when $\psi_0 $ is far from $1/4$, MLE and oracle GAN estimators are non-normally distributed but MSWD and MSCD estimators
are again close to being normally distributed. In Section \ref{section:asym-theory} and Appendix \ref{appendix:MSWD-MSCD}, we will verify conditions for asymptotic normality in Theorem \ref{thm:asym-normality-general} for both MSCD and MSWD estimators in the one-sided and two-sided uniform models, providing theoretical justification for the results in Figures \ref%
{fig:divergence-plot-unif-oneside-theta2}-\ref%
{fig:divergence-plot-unif-twoside-050}.

\section{Asymptotic Theory}

\label{section:asym-theory}

In this section, we establish consistency and asymptotic normality of the MSD
estimator $\hat{\psi}_{T}$ defined in (\ref{Estimator}) under a set of
high-level assumptions on $Q_{T}(s;u)$ and $\widehat{Q}_{T}(s;u,\psi )$. The
high-level assumptions allow for the results to be applicable to a wide range of models and estimators. In the next section, we will verify them under sufficient primitive conditions for the MSCD estimator in conditional models. 

\subsection{Consistency}

\begin{assumption}
\label{assumption:convergence} (i) $\int_{\mathbb{S}^{d-1}}$ $\int_{\mathcal{S%
}} \left( Q_{T}(s;u)-Q(s;u)\right) ^{2}w(s)\mathrm{d}s d\varsigma (u) %
\xrightarrow{p}0$;

(ii) $\sup_{\psi \in \Psi }\int_{\mathbb{S}^{d-1}} \int_{\mathcal{S}} \left( 
\widehat{Q} _{T}(s;u,\psi )-Q(s;u,\psi )\right) ^{2}w(s)\mathrm{d}s
d\varsigma (u) \xrightarrow{p}0$.
\end{assumption}

Assumption $\ref{assumption:convergence}$ requires that $Q_{T}(\cdot ,\cdot 
\dot{)}$ and $\widehat{Q}_{T}(\cdot ,\cdot ,\psi )$ converge to $Q(\cdot
,\cdot )$ and $Q(\cdot ,\cdot ,\psi )$ in weighted $L_{2}$-norm,
respectively, and the latter is uniform in $\psi \in \Psi $. Under Assumption \ref{assumption:DGP}, $Q(\cdot
,\cdot )=Q(\cdot ,\cdot ,\psi_0 )$.

Assumption \ref%
{assumption:identification} below states that $\psi _{0}$ is well-separated.

\begin{assumption}
\label{assumption:identification} The true parameter $\psi_0$ is in the interior of $\Psi$ such
that for all $\epsilon >0$, 
\begin{align*}
\inf_{\psi \notin B(\psi _{0},\epsilon )} \int_{\mathbb{S} ^{d-1}}\int_{S}
(Q(s;u)-Q(s;u,\psi ))^{2}w(s)\mathrm{d}s \mathrm{d}\varsigma (u) > 0,
\end{align*}
where $B(\psi _{0},\epsilon ):=\{\psi \in \Psi :\lVert \psi -\psi _{0}\rVert
\leq \epsilon \}$.
\end{assumption}

\begin{theorem}[Consistency of $\hat{\protect\psi}_{T}$]
\label{thm:consistency-general} Suppose Assumptions \ref%
{assumption:convergence} and \ref{assumption:identification} hold. Then $%
\hat{\psi}_{T}\xrightarrow{p}\psi _{0}$ as $T \rightarrow \infty $.
\end{theorem}

\subsection{Asymptotic Normality}

To establish asymptotic normality of $\widehat{\psi}_T$, we follow \cite{Andrews1999} and \cite{Pollard1980} by verifying the following quadratic approximation of the sliced distance $\widehat{\mathcal{S}}(\psi)$ defined in (\ref{SO}), 
\begin{equation}
	\widehat{\mathcal{S}}(\psi)
	 \approx \int_{\mathbb{S}^{d-1}}\int_{\mathcal{S}} (Q_{T}(s;u)- \widehat{Q}
	_{T}(s;u,\psi_0))^{2}w(s)\mathrm{d}s \mathrm{d}\varsigma(u) - 2(\psi -
	\psi_0)^{\top} A_T /\sqrt{T} + (\psi - \psi_0)^{\top} B_T (\psi - \psi_0),
 \label{QA}
\end{equation}
where 
\begin{align*}
	A_T & = \int_{\mathbb{S}^{d-1}}\int_{\mathcal{S}} (Q_{T}(s;u)- \widehat{Q}
	_{T}(s;u,\psi_0)) \widehat{D}_{T}(s;u,\psi _{0}) w(s)\mathrm{d}s \mathrm{d}
	\varsigma(u); \\
	B_T & = \int_{\mathbb{S}^{d-1}} \int_{\mathcal{S}} \widehat{D}_{T}(s;u,\psi _{0}) 
	\widehat{D}_{T}^{\top}(s ;u,\psi _{0}) w(s)\mathrm{d}sd\varsigma(u).
\end{align*}

Let 
\begin{equation*}
\widehat{R}_{T}(s;u,\psi ,\psi _{0}):=\widehat{Q}_{T}(s;u,\psi )- \widehat{Q}
_{T}(s;u,\psi _{0})-(\psi -\psi _{0})^{\top }\widehat{D}_{T}(s;u,\psi _{0}),
\end{equation*}
where $\widehat{D}_T(\cdot; \dot, \psi_0)$ is an $L_2(\mathcal{S} \times 
\mathbb{S}^{d-1}, w(s)\mathrm{d}s \mathrm{d}\varsigma(u))$-measurable
function. We adopt the following assumption to show the quadratic approximation in (\ref{QA}). 
\begin{assumption}
\label{assumption:norm-diff} The model-induced function $\widehat{Q}_{T}(\cdot ;u,\psi )$ satisfies: for any $\tau _{T}\rightarrow 0$,
\begin{equation*}
\sup_{\psi \in \Psi ;\;\lVert \psi -\psi _{0}\rVert \leq \tau
_{T}}\left\vert \frac{T \int_{\mathbb{S}^{d-1}}\int_{S} \left( \widehat{R}
_{T}(s;u,\psi ,\psi _{0})\right) ^{2}w(s)dsd\varsigma (u)}{(1+\lVert \sqrt{T}
(\psi -\psi _{0})\rVert )^{2}}\right\vert =o_{p}(1).
\end{equation*}
\end{assumption}
First, we note that Assumption \ref{assumption:norm-diff} is implied by the \textit{norm-differentiability} of $\widehat{Q}_{T}(\cdot ;u,\psi )$: for any $\tau _{T}\rightarrow 0$,
\begin{equation*}
\sup_{\psi \in \Psi ;\;\lVert \psi -\psi _{0}\rVert \leq \tau
_{T}}\left\vert \frac{T \int_{\mathbb{S}^{d-1}}\int_{S} \left( \widehat{R}
_{T}(s;u,\psi ,\psi _{0})\right) ^{2}w(s)dsd\varsigma (u)}{\lVert \sqrt{T}
(\psi -\psi _{0})\rVert^{2}}\right\vert =o_{p}(1).
\end{equation*}
Moreover, norm-differentiability of $\widehat{Q}_{T}(\cdot ;u,\psi )$ implies twice continuous differentiability of the sliced distance  $\widehat{\mathcal{S}}(\psi)$.
Second, Assumption \ref{assumption:norm-diff} allows the model-induced function $\widehat{Q}_T$ to have a finite number of kink points with respect to $\psi$ such as in the parameter-dependent support models in Example \ref{example:parameter-dependent-support}. In these models, one can take 
\begin{equation*}
\widehat{D}_{T}(s;u,\psi _{0})=%
\begin{cases}
\frac{\partial \widehat{Q}_{T}(s;u,\psi )}{\partial \psi }\bigg|_{\psi =\psi _{0}} & \text{
when }s\text{ is not a kink point;} \\ 
0 & \text{ when }s\text{ is a kink point.}%
\end{cases}%
\end{equation*} 

\begin{example} \label{example:uniform-models-norm-diff}
In this example, we verify Assumption \ref{assumption:norm-diff} for the unweighted MSCD estimator in the one-sided and two-sided uniform models in Section \ref{sec:illustrative_example}. Appendix \ref{appendix:MSWD-MSCD} verifies it for the MSWD estimator. 

(i) Consider the unweighted MSCD estimator for the one-sided uniform model, $Y\sim U[0,\psi _{0}]$, where $\psi _{0}>0$ is the true parameter. Since this is an unconditional univariate model, one can take the model-induced function $\widehat{Q}_T(s;u,\psi )$ as $Q(s; \psi)$, where $Q(s; \psi_0)$ is the parametric distribution function of $Y$ given by 
\begin{equation*}
Q(s; \psi_0):=F(s;\psi_0 )=%
\begin{cases}
0 & \text{ if }s\leq 0, \\ 
\frac{s}{\psi_0 } & \text{ if }0<s<\psi_0 , \\ 
1 & \text{ if }s\geq \psi_0.%
\end{cases}%
\end{equation*}%
As a result, 
\begin{equation*}
\widehat{D}_{T}(s;u,\psi _{0}):=D(s;\psi _{0})=%
\begin{cases}
0 & \text{ if }s\leq 0, \\ 
-\frac{s}{\psi _{0}^{2}} & \text{ if }0<s<\psi _{0}, \\ 
0 & \text{ if }s\geq \psi _{0}.%
\end{cases}%
\end{equation*}
Letting $\widehat{R} _{T}(s;u,\psi ,\psi _{0})=R(s;\psi ,\psi _{0})$, where 
\[R(s;\psi ,\psi _{0})= Q(s, \psi) - Q(s; \psi_0) - (\psi -
\psi_0)^{\top} D(s; \psi_0),\] 
we obtain that 
\begin{equation*}
\int_{-\infty }^{\infty }\left(R(s;\psi ,\psi _{0})\right) ^{2}dy= 
\begin{cases}
-\frac{(\psi-\psi_{0})^{3}}{3\psi _{0}^{2}} & \text{ if }\psi <\psi_{0}, \\ 
\frac{(\psi -\psi_{0})^{3}}{3\psi \psi_{0}} & \text{ if }\psi \geq \psi _{0}, 
\end{cases}
\end{equation*}
and $Q$ satisfies norm-differentiability, i.e.,
\begin{align*}
\sup_{\psi \in \Psi ;\;\lVert \psi -\psi _{0}\rVert \leq \tau
_{T}}\left\vert \frac{\int_{-\infty}^{\infty} \left( 
R(s;\psi ,\psi _{0})\right) ^{2}ds}{\lVert
\psi -\psi _{0}\rVert ^{2}}\right\vert =o(1).
\end{align*}
Consequently, the population objective function $\mathcal{SC}^2$ is twice continuously differentiable as shown in Figure \ref{fig:divergence-plot-unif-oneside-theta2}.

(ii) We show below that norm-differentiability holds for the unweighted MSCD estimator in the two-sided uniform model for all parameter values including $1/4$. 
The distribution function is given by 
\begin{equation*}
F(y;\psi )=%
\begin{cases}
0 & \text{if } y \le 0 \\
\frac{y}{4\psi } & \text{\ if }0 < y \le \psi, \\ 
\frac{1}{4}+\frac{3(y-\psi )}{4(1-\psi )}=1-\frac{3(1-y)}{4(1-\psi )} & 
\mathbf{\ }\text{if }\psi < y\leq 1, \\
1 & \text{if } y > 1.%
\end{cases}%
\end{equation*}
Adopting the same notation as in the one-sided uniform model, we have
\begin{equation*}
D(y;\psi _{0})= 
\begin{cases}
0 & \text{ if }y\leq 0, \\ 
-\frac{y}{4\psi _{0}^{2}} & \text{ if }0<y<\psi _{0}, \\ 
0 & \text{ if }y=\psi _{0}, \\ 
-\frac{3(1-y)}{4(1-\psi _{0})^{2}} & \text{ if } \psi _{0}<y<1, \\ 
0 & \text{ if }y\geq 1%
\end{cases}%
\end{equation*}
and
\begin{align*}
\int_{-\infty }^{\infty }\left( 
R(y;\psi ,\psi _{0})\right) ^{2}dy 
  = 
\begin{cases}
-\frac{(\psi -\psi _{0})^{3}(4\psi (\psi _{0}-1)+12\psi _{0}^{2}-4\psi
_{0}+1)}{48(\psi -1)(\psi _{0}-1)\psi _{0}^{2}} & \text{ if }\psi <\psi _{0},
\\ 
\frac{(\psi -\psi _{0})^{3}(12\psi \psi _{0}+4\psi _{0}^{2}-8\psi _{0}+1)}{
48\psi \psi _{0}(1-\psi _{0})^{2}} & \text{ if }\psi \geq \psi _{0}.%
\end{cases}%
\end{align*}
As a result, norm-differentiability holds for all parameter values $\psi_0$, i.e., 
\begin{align*}
\sup_{\psi \in \Psi ;\;\lVert \psi -\psi _{0}\rVert \leq \tau
_{T}}\left\vert \frac{\int_{-\infty}^{\infty} \left( 
R(s;\psi ,\psi _{0})\right) ^{2}ds}{\lVert
\psi -\psi _{0}\rVert ^{2}}\right\vert =o(1).
\end{align*}
Furthermore, the population objective function $\mathcal{SC}^2$ is twice continuously differentiable as shown in Figures \ref{fig:divergence-plot-unif-twoside-025} and \ref{fig:divergence-plot-unif-twoside-050}.
\end{example}

Assumption \ref{assumption:big-oh} (i) below strengthens Assumption \ref%
{assumption:convergence} (i) by specifying the rate of convergence.

\begin{assumption}
\label{assumption:big-oh} The following conditions hold:

\begin{description}
\item[(i)] $T\int_{\mathbb{S}^{d-1}}\int_{\mathcal{S}}
(Q_{T}(s;u)-Q(s;u))^{2}w(s) \mathrm{\ d}sd\varsigma (u)=O_{p}(1)$;

\item[(ii)] $T\int_{\mathbb{S}^{d-1}}\int_{\mathcal{S}} (\widehat{Q}%
_{T}(s;u,\psi _{0 })-Q(s;u,\psi _{0}))^{2}w(s)\mathrm{d}sd\varsigma
(u)=O_{p}(1)$;

\item[(iii)] There exists an $L_{2}(\mathbb{R}\times \mathbb{S}^{d-1},w(s) 
\mathrm{d}s\mathrm{d}\varsigma )$-measurable function $D(\cdot ;\cdot ,\psi
_{0})$ such that 
\begin{align*}
\int_{\mathbb{S}^{d-1}}\int_{\mathcal{S}} \left\Vert \widehat{D}%
_{T}(s;u,\psi _{0})-D(s;u,\psi _{0})\right\Vert ^{2}w(s)\mathrm{d}%
sd\varsigma (u) &=o_{p}(1).
\end{align*}
\end{description}
\end{assumption}
The next assumption implies asymptotic normality of the gradient vector of the sample objective function evaluated at the true parameter value. It can be verified by appropriate CLTs. 
\begin{assumption}
\label{assumption:CLT} 
\begin{equation*}
\sqrt{T} 
\begin{pmatrix}
\int_{\mathbb{S}^{d-1}}\int_{\mathcal{S}} (Q_{T}(s;u)-Q(s;u))D(s;u.\psi
_{0})w(s)dsd\varsigma (u) \\ 
\int_{\mathbb{S}^{d-1}}\int_{\mathcal{S}} (\widehat{Q}_T(s;u,\psi
_{0})-Q(s;u,\psi _{0}))D(s;u,\psi _{0})w(s)dsd\varsigma (u)%
\end{pmatrix}
\xrightarrow{d}N(0,V_{0})
\end{equation*}
for some positive semidefinite matrix $V_{0}$.
\end{assumption}

\begin{assumption}
\label{assumption:hessian-pd} 
\begin{equation}
B_{0}:=\int_{\mathbb{S}^{d-1}}\int_{\mathcal{S}} D(s;u,\psi _{0})D^{\top
}(s;u,\psi _{0})w(s)\mathrm{d}sd\varsigma (u)  \label{eq:hessian}
\end{equation}
is positive definite.
\end{assumption}

We are ready to state our main theorem. 
\begin{theorem}[Asymptotic normality of $\hat{\protect\psi}_{T}$]
\label{thm:asym-normality-general} Suppose Assumption \ref{assumption:DGP} and Assumptions \ref%
{assumption:convergence} to \ref{assumption:hessian-pd} hold. %
Then, 
\begin{equation*}
\sqrt{n}(\hat{\psi}_{T}-\psi _{0})\xrightarrow{d}N(0,B_{0}^{-1}\Omega
_{0}B_{0}^{-1}),
\end{equation*}
where $\Omega _{0}=(e_{1}^{\top }, - e_{1}^{\top})V_{0}%
\begin{pmatrix}
e_{1} \\ 
- e_{1}%
\end{pmatrix}%
$ in which $e_{1}=(1,\dotsc ,1)^{\top }$ is a $d_{\psi}$ by $1$ vector of
ones.
\end{theorem}

\begin{remark}
When applied to specific models such as the conditional model, Wald-type inference could be done using plug-in estimators of $B_0$ and $V_0$, see Proposition \ref{prop:Asy Normality-one-sided} for expressions of $B_0$ and $V_0$.
\end{remark}

\section{MSCD Estimation of Conditional Models}

\label{section:SC-conditional}

In this section, we impose smoothness conditions to verify the high level assumptions for consistency and asymptotic normality of the MSCD estimator in conditional models, where $Q(\cdot ;u,\psi )=G(\cdot ;u,\psi )$ defined in (\ref{G}) and $\widehat{Q}_{T}(\cdot ;u,\psi )=\widehat{G}_{T}(\cdot ;u,\psi )$ defined in (\ref{GT}). 

Throughout this section, we assume that the parameter space $\Psi $ is compact and the weight function $w(\cdot)$ is integrable, i.e., $\int_{\mathbb{R}} w(s) \mathrm{d}s < \infty$. 

\subsection{Smoothness Assumptions}
 
We impose Conditions \ref{condition:Lip-F} and \ref{condition:norm-diff-true-Q} below on the conditional model and verify them for the one-sided parameter dependent support model.\footnote{Verification for the two-sided parameter-dependent support model is similar and postponed to Appendix \ref{appendix:two-sided}.}

\begin{condition}[Pointwise Lipschitz Continuity]
\label{condition:Lip-F} The conditional cdf $F(y|x, \psi)$ is pointwise Lipschitz with
respect to $\psi$ in the sense that there exists a function $M(y, x)$ such that for any $\psi, \psi^{\prime}\in \Psi$, 
\begin{align*}
|F(y|x, \psi) - F(y|x, \psi^{\prime})| \le M(y, x) \|\psi - \psi^{\prime}\|
\end{align*}
and $\int_{u \in \mathbb{S}^{d-1}, u_1 \ne 0} \int_{-\infty}^{\infty} \int
M^2(u_1^{-1}(s - u_2^{\top} x); x) \mathrm{d}F_X(x) w(s) \mathrm{d}s \mathrm{%
d}\varsigma(u) < \infty$, where $F_X(\cdot)$ is the cdf of $X$. 
\end{condition}

Condition \ref{condition:Lip-F} is implied by uniform Lipschitz continuity of $F(y|x,\psi )$ with respect to $\psi$, i.e., there exists a finite constant $C$ such that for all $y\in\mathcal{Y}$ and  $x\in\mathcal{X}$,
\begin{equation*}
|F(y|x,\psi )-F(y|x,\psi ^{\prime })|\leq C\Vert \psi -\psi ^{\prime }\Vert.
\end{equation*}
It is easy to see that Condition \ref{condition:Lip-F} is satisfied in the one-sided and two-sided uniform models. 

We now show that it is satisfied in the one-sided parameter-dependent support models in Example 2.1 under very mild conditions.

For one-sided parameter dependent support models, uniform Lipschitz continuity of $F(y|x,\psi )$ with respect to $\psi$ holds if $F(y|x,\psi )$ is absolutely continuous with respect to $\psi $ for
each $y$ and $x$, and
$\sup_{y\neq g(x,\psi )}|\partial F(y|x,\psi )/\partial \psi |$ is uniformly
bounded by a finite absolute constant. It follows from equation (\ref{eq:one-sided}) that 
\begin{equation*}
\frac{\partial F(y|x,\psi )}{\partial \psi }=%
\begin{cases}
0 & \text{ if }y<g(x, \theta) \\ 
-\frac{\partial g(x,\psi )}{\partial \psi }f_{\epsilon }(y-g(X,\psi )|x,\psi
) + \int_{0}^{y-g(x,\psi )}\frac{\partial f_{\epsilon }(\epsilon |x,\psi )}{%
\partial \psi }\mathrm{d}\epsilon & \text{ if }y>g(x,\theta).%
\end{cases}%
\end{equation*}
Here, $\psi = (\theta^{\top}, \gamma^{\top})^{\top}$ and $\partial g(x, \psi)/\partial \psi := [\partial g(x, \theta)/\partial \theta^{\top}, 0]^{\top}$ because $g(x, \theta)$ does not contain $\gamma$.

\begin{lemma}[One-sided model] 
\label{lemma:Lipchitz-one-sided} 
Suppose that $f_{\epsilon }(\epsilon |x,\psi )$ is continuous in $(\epsilon ,\psi ) \in \mathbb{R}_{+} \times \Psi$ for each $x$, 
and $f_{\epsilon }(\epsilon |x,\psi )$ is differentiable with respect to $\psi$ for each $\epsilon$ and $\psi$; 
and $g(x, \psi)$ is continuously differentiable with respect to $\psi$ for each $x$, 
and there exists $\bar{f}_{\epsilon }(\epsilon |x)$
such that 

(i)
$\left\Vert \frac{f_{\epsilon }(\epsilon |x,\psi )}{\partial \psi }%
\right\Vert \leq \bar{f}_{\epsilon }(\epsilon |x)\text{ for all }\epsilon
,x,\psi$, (ii) $\int_{0}^{\infty}\bar{f}_{\epsilon }(\epsilon |x)
\mathrm{d}\epsilon < \infty \text{ for each } x$, 

(iii)
$\sup_{x,\psi }\left\Vert \frac{\partial g(x,\psi )}{\partial \psi }%
\right\Vert <\infty$, (iv) $\sup_{\epsilon .x,\psi }|f_{\epsilon }(\epsilon
|x,\psi )|<\infty$, and 

(v) $\sup_{x, \psi}
\int_{0}^{\infty} \left\Vert \frac{\partial f_{\epsilon }(\epsilon |x,\psi )}{%
\partial \psi } \right\Vert \mathrm{d}\epsilon  <\infty.$

Then $F(y|x, \psi)$ is Lipschitz with respect to $%
\psi $ uniformly in $y$ and $x$. 
\end{lemma}

The first two assumptions (i) and (ii) in Lemma \ref{lemma:Lipchitz-one-sided} ensure
validity of interchanging the differentiation operation and integration in equation (\ref{eq:one-sided}). The third
and fourth assumptions (iii) and (iv) are included in Conditions C2 and C3 in \cite%
{Chernozhukov2004}. The fifth condition (v) is stronger than the condition 
$\sup_{\psi} \int
\int \left\Vert \frac{\partial f_{\epsilon }(y-g(x, \psi) |x,\psi )}{%
\partial \psi } \right\Vert \mathrm{d}y \mathrm{d}F_X(x) < \infty$ in C2 in \cite%
{Chernozhukov2004}.
However, the assumption: $f_{\epsilon
}(0|x,\psi )>\eta >0$ in \cite{Chernozhukov2004} (see also (\ref{Jump}) in  Example 2.1) is not required in Lemma \ref{lemma:Lipchitz-one-sided}. Consequently, the asymptotic distribution of our MSCD estimator is normal regardless of the presence/absence of a jump in $f_{\epsilon}(0|x,\psi )$.

Let 
\begin{align*}
R(s; u, \psi, \psi_0) := Q(s, u, \psi) - Q(s; u, \psi_0) - (\psi -
\psi_0)^{\top} D(s; u, \psi_0),
\end{align*}
where $D(\cdot;\cdot,\psi _{0})$ is a $L_{2}(\mathbb{R} \times \mathbb{S}
^{d-1} ,w(s)\mathrm{d}s \mathrm{d}\varsigma)$-measurable function. We impose the analog of 
Assumption \ref{assumption:norm-diff} on the population function $Q$ and use it to verify Assumption \ref{assumption:norm-diff} for the random $\widehat{Q}_T$.
\begin{condition}
\label{condition:norm-diff-true-Q} $Q(\cdot ;u,\psi )$ satisfies: for any $\tau _{T} \rightarrow 0$,
\begin{equation*}
\sup_{\psi \in \Psi ;\;\lVert \psi - \psi_{0} \rVert \leq \tau_{T}} \left| 
\frac{T \int_{\mathbb{S}^{d-1}}\int_{-\infty}^{\infty} \left(
R(s;u,\psi,\psi _{0})\right) ^{2}w(s)\mathrm{d}s \mathrm{d}\varsigma (u)}{(1
+ \lVert \sqrt{T}(\psi-\psi _{0})\rVert) ^{2}} \right| = o(1).
\end{equation*}
\end{condition}

We now verify Condition \ref{condition:norm-diff-true-Q} in the one-sided parameter-dependent support model.

\begin{lemma} \label{lemma:sup-bound-hessian}
Suppose that all the assumptions in Lemma \ref{lemma:Lipchitz-one-sided} hold and $F(y|x, \psi)$ continuously second-order differentiable with respect to $\psi$ except for $y = g(x, \psi)$.
Then Condition \ref{condition:norm-diff-true-Q} is implied by 
\begin{equation*}
\sup_{y \ne g(x, \psi) ,x,\psi }\left\Vert \frac{\partial ^{2}F(y|x, \psi)}{\partial
\psi \partial \psi ^{\prime }}\right\Vert <\infty .
\end{equation*}
\end{lemma}

In the one-sided model, when $y > g(x, \psi)$, 
\begin{align*}
\frac{\partial^2 F(y|x, \psi)}{\partial \psi \partial \psi^{\prime}} & = - 
\frac{\partial^2 g(x, \psi)}{\partial \psi \partial \psi^{\prime}}
f_{\epsilon}(y - g(x, \psi)|x, \psi) \\
& \quad + \frac{\partial g(x, \psi)}{\partial \psi} \left[\frac{\partial
g(x, \psi)}{\partial \psi^{\prime}} \frac{\partial f_{\epsilon}(y - g(x, \psi)|x,
\psi)}{\partial \epsilon} + \frac{\partial f_{\epsilon}(y - g(x, \psi)|x, \psi) }{\partial
\psi^{\prime}} \right] \\
& \quad - \frac{\partial g(x, \psi)}{\partial \psi} \frac{f_{\epsilon}(y -
g(x, \psi)|x, \psi) }{\partial \psi^{\prime}} + \int_{0}^{y - g(x, \psi)} 
\frac{\partial^2 f_{\epsilon}(\epsilon, x, \psi)}{\partial \psi \partial
\psi^{\prime}} \mathrm{d}\epsilon.
\end{align*}

\begin{lemma}[One-sided model] 
	\label{lemma:norm-diff-one-sided-general} 
 In addition to the conditions in Lemma \ref{lemma:Lipchitz-one-sided}, 
 suppose that  $\partial f_{\epsilon }(\epsilon |x,\psi )$ is continuously differentiable with respect to $(\epsilon, \psi)$ (except for $\epsilon = 0$) for each $x$,
 and $f_{\epsilon }(\epsilon |x,\psi )$ is second-order continuously differentiable with respect to $\psi$ for each $\epsilon$ and $\psi$; 
and $g(x, \psi)$ is continuously second-order differentiable with respect to $\psi$ for each $x$. Further,
assume that there exists $\tilde{f}_{\epsilon}(\epsilon |x)$ such that 

(i) $\left\Vert \frac{\partial f_{\epsilon }(\epsilon |x,\psi )}{\partial \psi
			\partial \psi ^{\prime }}\right\Vert \leq \tilde{f}_{\epsilon}(\epsilon|x)$ and $\int_{0}^{\infty} \tilde{f}_{\epsilon }(\epsilon |x)\mathrm{d}\epsilon <\infty \text{
			for each } x$. 
   
   (ii) Further, suppose $\sup_{x,\psi }\left\Vert \frac{\partial ^{2}g(x,\psi )}{\partial \psi
			\partial \psi ^{\prime }}\right\Vert ,\; \sup_{\epsilon \ge 0,x,\psi }\left\Vert 
		\frac{\partial f_{\epsilon }(\epsilon |x,\psi )}{\partial \psi }\right\Vert
		,\; 
  \sup_{\epsilon > 0, x,\psi }\left\Vert \frac{\partial f_{\epsilon }(\epsilon
			|x,\psi )}{\partial \epsilon }\right\Vert$, and 
   
  $ \sup_{x,\psi}
		\int_{0}^{\infty} \left\| \frac{\partial ^{2}f_{\epsilon }(\epsilon |x,\psi )}{%
			\partial \psi \partial \psi ^{\prime }} \right\| \mathrm{d}\epsilon$
	are bounded above by finite constants. 
 
 Then $\widehat{Q}_{T}(\cdot ;u,\psi
	) $ is norm-differentiable at $\psi =\psi _{0}$ with $D(s;u,\psi )=\mathbb{E}%
	\left[ \frac{\partial F(u_{1}^{-1}(s-u_{2}^{\top }X_{t})|X_{t},\psi )}{%
		\partial \psi }\right] $.
\end{lemma}

Assumption (i) in Lemma \ref{lemma:norm-diff-one-sided-general} ensures
validity of interchanging the differentiation operation and integration. The first three conditions in Assumption (ii) are included in Conditions C2 and C3 in \cite%
{Chernozhukov2004}. It is important to note that like in Lemma \ref{lemma:Lipchitz-one-sided}, the assumption: $f_{\epsilon
}(0|x,\psi )>\eta >0$ in \cite{Chernozhukov2004} (see also (\ref{Jump}) in  Example 2.1) is not required in Lemma \ref{lemma:norm-diff-one-sided-general}.

\subsection{Consistency and Asymptotic Normality}

We are now ready to state the main result in this section.
\begin{proposition} 
\label{prop:Asy Normality-one-sided}
Suppose that Assumptions \ref{assumption:DGP}, \ref{assumption:convergence} (i), \ref{assumption:identification}, 
\ref{assumption:hessian-pd} and Conditions \ref{condition:Lip-F}- \ref{condition:norm-diff-true-Q} hold. In addition, $\int_{u \in \mathbb{S}^{d-1}} \int_{-\infty}^{\infty} \| D(s, u, \psi_0) \| w(s)  \mathrm{d}s \mathrm{d}\varsigma(u)
< \infty$. Then $\widehat{\psi}_T$ is consistent and asymptotically normally distributed,  i.e., 
\begin{equation*}
\sqrt{n}(\hat{\psi}_{T}-\psi _{0})\xrightarrow{d}N(0,B_{0}^{-1}\Omega
_{0}B_{0}^{-1}),
\end{equation*}
where 
$
B_{0}=\int_{\mathbb{S}^{d-1}}\int_{\mathcal{S}} D(s;u,\psi _{0})D^{\top
}(s;u,\psi _{0})w(s)\mathrm{d}sd\varsigma (u)
$
with $D(s;u,\psi )=\mathbb{E}%
	\left[ \frac{\partial F(u_{1}^{-1}(s-u_{2}^{\top }X_{t})|X_{t},\psi )}{%
		\partial \psi }\right] $, and 
$\Omega _{0}=(e_{1}^{\top }, - e_{1}^{\top})V_{0}%
\begin{pmatrix}
e_{1} \\ 
- e_{1}%
\end{pmatrix}%
$ in which $e_{1}=(1,\dotsc ,1)^{\top }$ is a $d_{\psi}$ by $1$ vector of
ones with 
\begin{align*}
V_0 = \int_{u \in \mathbb{S}^{d-1}} \int_{v \in \mathbb{S}^{d-1}} \int \int 
\begin{pmatrix}
A_{11}(t, s; u, v) & A_{12}(t, s; u, v) \\ 
A_{21}(t, s; u, v) & A_{22}(t,s; u, v)%
\end{pmatrix}
\otimes D(s; u, \psi_0) D(t; u, \psi_0)^{\top} w(s) w(t) \mathrm{d}s \mathrm{%
\ d}t \mathrm{d}\varsigma(u) \mathrm{d}\varsigma(v)
\end{align*}
in which $\otimes$ is the Kronecker product, and 
\begin{align*}
A_{11}(t, s; u, v) & = \mathbb{E}[I(u^{\top} Z \le t) I(v^{\top} Z \le s)] -
G(t; u) G(s, v), \\
A_{22}(t, s; u, v) & = \mathbb{E}\left\{\mathbb{E}[I(u^{\top} Z \le t|X)] 
\mathbb{E}[I(v^{\top} Z \le s)|X] \right\} - G(t; u) G(s, v), \\
A_{12}(t, s; u, v) & = \mathbb{E}\left\{I(u^{\top} Z \le t) \mathbb{E}
[I(v^{\top} Z \le s)|X] \right\} - G(t; u) G(s, v), \\
A_{21}(t, s; u, v) & = \mathbb{E}\left\{\mathbb{E}[I(u^{\top} Z \le t|X)]
I(v^{\top} Z \le s) \right\} - G(t; u) G(s, v).
\end{align*}
\end{proposition}

By Theorems 3.1 and 3.2, it is sufficient to verify Assumption \ref{assumption:convergence} (ii), Assumption \ref%
{assumption:norm-diff}, Assumption \ref{assumption:big-oh} (ii) and (iii),
and Assumption \ref{assumption:CLT}, given that $\{Z_t\}$ satisfies Assumption \ref{assumption:DGP}. Assumption \ref{assumption:big-oh} (ii) is straightforward to verify, Assumption \ref{assumption:CLT} relies on CLT, and the other assumptions make heavy use of tools for degenerate U-statistics. We discuss details in the rest of this section.

\paragraph{Verification of Assumption \ref{assumption:convergence} (ii)} Noting that $$\widehat{G}_{T}(s;u,\psi ) = \frac{1}{T}\sum_{t=1}^{T} \mathbb{E}[u^{\top} Z_t
\le s | X_t, \psi],$$ we have 
\begin{align*}
\widehat{G}_T(s; u. \psi) - G(s; u, \psi) = \frac{1}{T} \sum_{t=1}^T \bigg( 
\mathbb{E}[I(u^{\top} Z_t \le s)|X_t, \psi] - \mathbb{E}[I(u^{\top} Z_t \le
s)|\psi]\bigg).
\end{align*}
Hence, $\int_{u \in \mathbb{S}^{d-1}} \int_{-\infty}^{\infty} (\widehat{G}%
_T(s; u. \psi) - G(s; u, \psi))^2 w(s) \mathrm{d}s \mathrm{d}\varsigma(u)$
can be represented as a degenerate $V$-statistic of order 2, i.e.,
\begin{align*}
\int_{u \in \mathbb{S}^{d-1}} \int_{-\infty}^{\infty} (\widehat{G}_T(s; u.
\psi) - G(s; u, \psi))^2 w(s) \mathrm{d}s \mathrm{d}\varsigma(u) = \frac{1}{
T^2} \sum_{t=1}^T \sum_{j=1}^T k(X_t, X_j; \psi),
\end{align*}
where $k(\cdot,\cdot;\psi)$ is a degenerate symmetric kernel function indexed by $\psi$:
\begin{align*}
& k(X_t, X_j; \psi) = \int_{u \in \mathbb{S}^{d-1}} \int_{-\infty}^{\infty} \left(\mathbb{E}
[I(u^{\top} Z_t \le s)|X_t, \psi] - \mathbb{E}[I(u^{\top} Z_t \le
s)|\psi]\right)\\
&\times \left(\mathbb{E}[I(u^{\top} Z_j \le s)|X_j, \psi] - \mathbb{%
E }[I(u^{\top} Z_j \le s)|\psi]\right) w(s) \mathrm{d}s \mathrm{d}%
\varsigma(u)
\end{align*}

 Lemma \ref{lemma:Lipchitz-k} in Appendix \ref{appendix:proofs-SC} shows that Lipschitz continuity of $k$ is implied by that of $F$ in Condition \ref{condition:Lip-F}. Under the Lipschitz continuity of $k(x, x^{\prime},\psi )$ with respect to $%
\psi$ for every $x$ and $x^{\prime}$, Corollary 4.1 of \cite{Newey1991} or
Lemma 4 in the appendix of \cite{Briol2019} can be used to verify Assumption \ref%
{assumption:convergence} (ii).

To sum up, the following lemma holds. 
\begin{lemma}
\label{lemma:small-oh-SC} 
Suppose Condition \ref{condition:Lip-F} holds. Then Assumption \protect\ref{assumption:convergence}
(ii) holds, i.e., 
\begin{align*}
\sup_{ \psi \in \Psi} \int_{u \in \mathbb{S}^{d-1}} \int_{-\infty}^{\infty}
( \widehat{G}_T(s; u. \psi) - G(s; u, \psi))^2 w(s) \mathrm{d}s \mathrm{d}
\varsigma(u) = o_p(1).
\end{align*}
\end{lemma}

\paragraph{Verification of Assumption \ref{assumption:big-oh}
(ii)} Note that under i.i.d assumption, we have 
\begin{align*}
& T \mathbb{E}\left[\int_{-\infty}^{\infty} (\widehat{G}_T(s; u. \psi_0) -
G(s; u, \psi_0))^2 w(s) \mathrm{d}s \mathrm{d}\varsigma(u) \right] \\
&\qquad =  \int_{-\infty}^{\infty} \mathbb{E} \left[T (\widehat{G}_T(s; u, \psi_0)
- G(s; u, \psi_0))^2\right] w(s) \mathrm{d}s \mathrm{d}\varsigma(u) \\
& \qquad =  \int_{-\infty}^{\infty} \mathbb{E} \left[\bigg(\mathbb{E}[I(u^{\top} Z_t
\le s)|X_t, \psi_0] - G(s; u, \psi_0)\bigg)^2\right] w(s) \mathrm{d}s 
\mathrm{d}\varsigma(u) \\
& \qquad =  \int_{-\infty}^{\infty} \left\{\mathbb{E} \left[\bigg(\mathbb{E}
[I(u^{\top} Z_t \le s)|X_t, \psi_0] \bigg)^2\right] - G^2(s; u,
\psi_0)\right\} w(s) \mathrm{d}s \mathrm{d}\varsigma(u) \\
& \qquad \le  \int_{-\infty}^{\infty} \left\{\mathbb{E} \left[\bigg(\mathbb{E}
[I(u^{\top} Z_t \le s)^2|X_t, \psi_0] \bigg)\right] - G^2(s; u,
\psi_0)\right\} w(s) \mathrm{d}s \mathrm{d}\varsigma(u) \\
& \qquad =  \int_{-\infty}^{\infty} G(s; u, \psi_0) \bigg(1 - G(t; u, \psi_0)\bigg) %
w(s) \mathrm{d}s \mathrm{d}\varsigma(u).
\end{align*}
So Assumption \ref{assumption:big-oh} (ii) holds.

\paragraph{Verification of Assumptions \ref{assumption:norm-diff} and \ref{assumption:big-oh} (iii)}

Under Condition \ref{condition:norm-diff-true-Q}, it is sufficient to show that
\begin{align*}
\sup_{|\psi - \psi_0| \le \tau_T} \frac{T \left(\int_{u \in \mathbb{S}
^{d-1}} \int_{-\infty}^{\infty} \left( \big[\widehat{Q}_{T}(s;u,\psi ) -
Q(s;u,\psi )\big] - \big[\widehat{Q}_{T}(s;u,\psi_0) - Q(s;u,\psi _{0})\big] %
\right)^2 w(s) \mathrm{d}s \mathrm{d}\varsigma(u) \right)}{(1 +  \|\sqrt{T}(\psi -
\psi_0)\|)^2} & = o_p(1).
\end{align*}
Note that the integral in the numerator on the left hand side of the above equation can be represented as a degenerate
V-statistic of order 2: 
\begin{align*}
& \left(\int_{u \in \mathbb{S}^{d-1}} \int_{-\infty}^{\infty} \left(\big[%
\widehat{Q}_{T}(s;u,\psi ) - Q(s;u,\psi )\big] - \big[\widehat{Q}%
_{T}(s;u,\psi_0) - Q(s;u,\psi _{0}) \big] \right)^2 w(s) \mathrm{d}s \mathrm{%
d}\varsigma(u) \right) \\
& = \frac{1}{ T^2} \sum_{t=1}^T \sum_{j=1}^{T} k_2(X_t, X_j, \psi, \psi_0),
\end{align*}
where $ k_2(X_t, X_j, \psi, \psi_0)$ is a degenerate symmetric kernel given by 
\begin{align*}
&  \int_{u \in \mathbb{S}^{d-1}} \int_{-\infty}^{\infty} \Bigg\{ \bigg[\Big(%
\mathbb{E} [I(u^{\top} Z_t \le s)|X_t, \psi] - G(s; u, \psi)\Big) - \Big(%
\mathbb{E} [I(u^{\top} Z_t \le s)|X_t, \psi_0] - G(s; u, \psi_0)\Big)\bigg]
\\
& \quad \quad \times \bigg[\Big(\mathbb{E}[I(u^{\top} Z_j \le s)|X_j, \psi]
- G(s; u, \psi)\Big) - \Big(\mathbb{E}[I(u^{\top} Z_j \le s)|X_j, \psi_0] -
G(s; u, \psi_0)\Big)\bigg] \Bigg\} w(s) \mathrm{d}s \mathrm{d}\varsigma(u).
\end{align*}
When $k_{2}(x, x^{\prime},\psi ,\psi _{0})$ is Lipschitz continuous with
respect to $\psi $ for each $x$, $x^{\prime}$, we can use Corollary 8 in 
\cite{Sherman1994} and the proof of Lemma 4 in the Appendix of \cite{Briol2019}
to verify Assumptions \ref{assumption:norm-diff} and \ref{assumption:big-oh}
(iii). Furthermore, Lipchitz continuity of $k_2$ is implied by that of $F(\cdot|x, \psi)$
in Condition \ref{condition:Lip-F}, see Lemma \ref{lemma:Lipchitz-k} in Appendix \ref{appendix:proofs-SC}. 

Summing up, we obtain the following result. 
\begin{lemma}
\label{lemma:proof-norm-diff-approx-Q-SC}

Suppose Conditions \ref{condition:Lip-F} and \ref{condition:norm-diff-true-Q} hold. Then $\widehat{Q}%
_T(\cdot; \cdot, \psi)$ satisfies Assumptions %
\ref{assumption:norm-diff} and \ref{assumption:big-oh} (iii) with $%
\widehat{D}_T(s; u, \psi) = D(s; u, \psi)$ in the definition of $R$ in Condition \ref{condition:norm-diff-true-Q}.

\end{lemma}

\paragraph{Verification of Assumption \ref{assumption:CLT}} Note that 
\begin{align*}
\sqrt{T}(\widehat{Q}_{T}(s; u)-Q(s;u )) & = \frac{1}{\sqrt{n}} \sum_{i=1}^n
\left(I(u^{\top} Z_t \le s) - G(s; u) \right) \text{ and } \\
\sqrt{T}(\widehat{Q}_{T}(s; u, \psi_0)-Q(s;u, \psi_0)) & = \frac{1}{\sqrt{T}}
\sum_{i=1}^n \left(\mathbb{E}[I(u^{\top} Z_t \le s)|X_t, \psi_0] - G(s; u,
\psi_0)\right).
\end{align*}

By applying CLT, we obtain the following result.
\begin{lemma}
Suppose $\int_{u \in \mathbb{S}^{d-1}} \int_{-\infty}^{\infty} \| D(s, u, \psi_0) w(s) \| \mathrm{d}s \mathrm{d}\varsigma(u)
< \infty$. Then, Assumption \ref{assumption:CLT} holds: 
\begin{align*}
& \sqrt{T} 
\begin{pmatrix}
\int_{u \in \mathbb{S}^{d-1}} \int_{-\infty}^{\infty} (\widehat{Q}_{T}(s;
u)-Q(s;u )) D(s; u, \psi_0) w(s) \mathrm{d}s \mathrm{d}u \\ 
\int_{u \in \mathbb{S}^{d-1}} \int_{-\infty}^{\infty} (\widehat{Q}_T(s;u,
\psi_0) -Q(s;u, \psi_0)) D(s; u, \psi_0) w(s) \mathrm{d}s \mathrm{d}u%
\end{pmatrix}
\\
& \qquad =  \frac{1}{\sqrt{T}} \sum_{t=1}^T 
\begin{pmatrix}
\int_{u \in \mathbb{S}^{d-1}} \int_{-\infty}^{\infty} \left(I(u^{\top} Z_t
\le s) - G(s; u)\right) D(s; u, \psi_0) w(s) \mathrm{d}s \mathrm{d}u \\ 
\int_{u \in \mathbb{S}^{d-1}} \int_{-\infty}^{\infty} (\mathbb{E}[I(u^{\top}
Z_t \le s)|X_t, \psi_0] -Q(s;u; \psi_0)) D(s; u, \psi_0) w(s) \mathrm{d}s 
\mathrm{d}u%
\end{pmatrix}
\\
& \qquad \xrightarrow{d} N(0, V_0).
\end{align*}

\end{lemma}

\section{Numerical Results}

\label{sec:numerical-results}

In this section, we report simulation results on the accuracy of the asymptotic normal distribution of MSCD estimator in the auction model in Example \ref{example:auction}.

In the simulation, we follow \cite{Li2010}, where $h(x, \psi) = \exp(\psi_1 + \psi_2 x)$ and $X$ follows the square of $U[0, 2]$. We set $m = 6$ and  $\psi_0 = (1, 0.5)$, $(1, 3)$.  The number of Monte-Carlo simulations is $2000$.

In addition to our MSCD estimator, we also apply \cite{Li2010}'s indirect inference estimator. Both estimators are asymptotically normally distributed leading to simple Wald-type inference.
For the MSCD estimator, we choose $w(s) = \mathds{1}(s \in [-50000, 50000])$ and 100 projection vectors. The algorithm is similar to the Simulation Algorithm in \citet{Deshpande2018} except that we only draw projection vectors during optimization. A description of our algorithm is given in Algorithm \ref{algorithm}. 
\begin{algorithm} 
    \caption{Algorithm for MSCD Estimator}
    \label{algorithm}
    \KwData{$\{Z_t\}_{t=1}^{T}$ where $Z_t = (Y_t, X_t')'$. 
    $n$: the number of observation, $n_{epoch}: $ the number of epochs (iteration), $m: $ the number of projection vectors.}
    \For{$i\gets0$ \KwTo $n_{epoch}$}{
    Initial Loss: $L \to 0$; \\
    Draw $U:=\{u_i\}_{i=1}^{m}$ uniformly from $\mathds{S}^{d-1}$;\\
    Calculate sliced distance; \\
    \For{each $u \in U$}{
        Construct $Q_{T}(s;u)$ using $\{Z_t\}_{t=1}^{T}$; \\
        $L$ \KwTo $L + \int_{\mathcal{S}}(Q_{T}(s;u)-%
\widehat{Q}_{T}(s;u,\psi ))^{2}w(s)\mathrm{d}s$
     }
     $L$ \KwTo $L/m$; \\
     Update $\psi$ via Adam algorithm.
    }
\end{algorithm}

For \cite{Li2010}'s estimator, we draw 100 synthetic samples and  run the auxiliary linear regression with regressor $(1, x)$. 
Also, we use the optimal weighting matrix for \cite{Li2010}'s estimator. 

For both estimators, optimization is done via Adam optimizer, where the learning rate is 0.1 and tuning parameter $(\beta_1, \beta_2)$ is $(0.9, 0.999)$ in options. We run 1000 epochs for all optimization. The code is implemented in Pytorch.

We consider and summarize results in Case 1 and Case 2 below. 

\paragraph{Case 1: $\psi_0 = (1, 0.5)$ and $T = 100$}

Figures \ref{fig:qqplot-auction-theta_10_05_init_00} and \ref{fig:hist-auction-theta_10_05_init_00} present respectively the qqplots and histograms of the normalized values of our MSCD estimator and \cite{Li2010}'s estimator. 
They imply that when $\psi_0 = (1, 0.5)$, the normal distribution approximates the finite sample distributions of both our estimator and \cite{Li2010}'s  estimator well even when $T=100$.

\paragraph{Case 2: $\psi_0 = (1, 3)$ and $T = 100, 200$}

Figures \ref{fig:qqplot-auction-theta_10_30_n100} and \ref{fig:hist-auction-theta_10_30_n100} show respectively the qqplots and histograms of normalized values of our MSCD estimator and \cite{Li2010}'s estimator when $T=100$. 
The normal approximation for both estimators is less accurate than in Case 1. Furthermore, we find that \cite{Li2010}'s estimator could be sensitive to the initial value when $T=100$. We then increased the sample size to $T=200$ and experimented with different initial values for \cite{Li2010}'s estimator. 
Figures \ref{fig:qqplot-auction-theta_10_30-n200} and \ref{fig:hist-auction-theta_10_30-n200} show respectively the qqplots and histograms of
normalized values of our MSCD estimator and \cite{Li2010}'s estimator when $T=200$. As expected, the accuracy of the normal approximation improves for both estimators as $T$ increases. Moreover, \cite{Li2010}'s estimator with initial value close to the true value performs comparably with the MSCD estimator.

\section{Concluding Remarks}
\label{sec:conclusion} 

Motivated by simple inference in nonregular structural models in economics, we have proposed the method of minimum sliced distance estimation based on empirical and model-induced measures of the true distribution. We have
developed a unified asymptotic theory under high level assumptions and
verified them for the parameter-dependent support model. The numerical illustration on an auction model confirms the efficacy of the proposed methodology. 

\begin{figure}[h]
    \centering
    \begin{subfigure}[h]{\linewidth}
    	\centering
    	\includegraphics[width=0.6\linewidth]{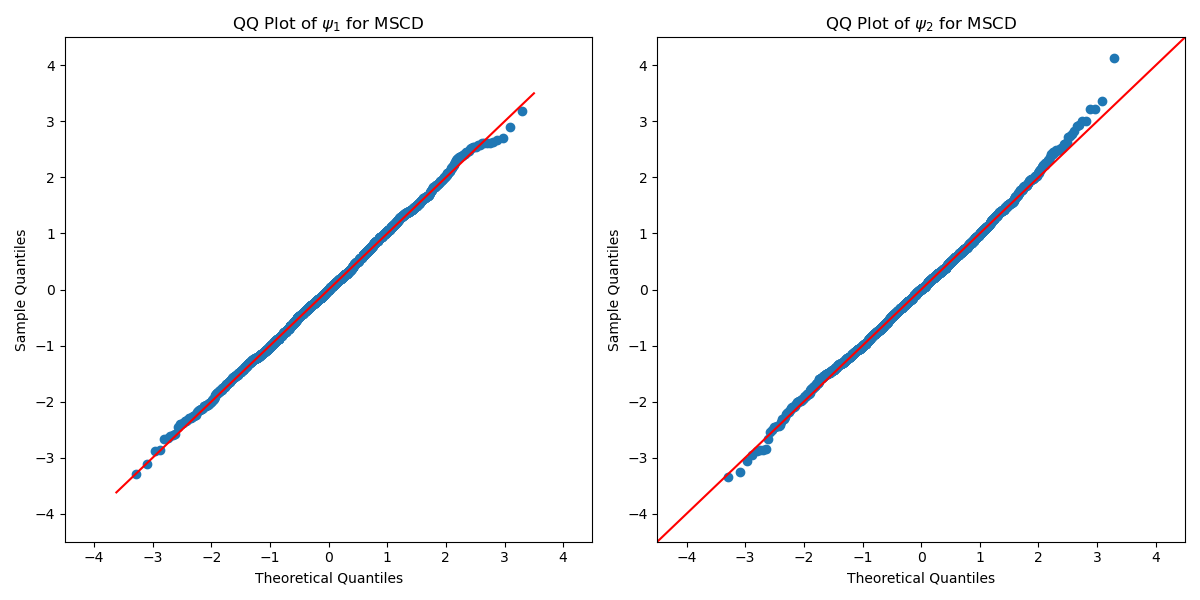}
    	\caption{QQ plots of normalized values of our MSCD estimator.}
    	\label{fig:qqplot-auction-MSCD-theta_10_05_init_00}
    \end{subfigure}
    \begin{subfigure}[h]{\linewidth}
    	\centering
    	\includegraphics[width=0.6\linewidth]{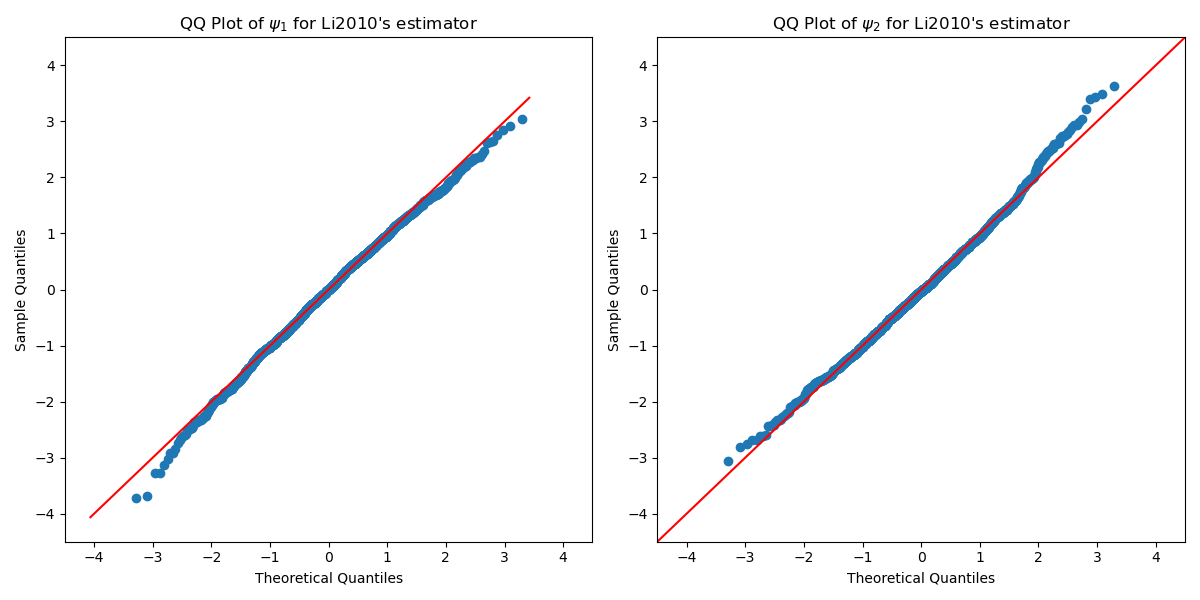}
    	\caption{QQ plots of normalized values of \cite{Li2010}'s estimator.}
    	\label{fig:qqplot-auction-Li2010-theta_10_05_init_00}
    \end{subfigure}
    \caption{QQ plots of normalized values of our MSCD estimator and \cite{Li2010}'s estimator when $\psi = (1, 0.5)$ and initial value is set to $(0, 0)$.}
    \label{fig:qqplot-auction-theta_10_05_init_00}
\end{figure}

\begin{figure}[h]
    \centering
    \begin{subfigure}[h]{\textwidth}
    \centering
   \includegraphics[width=0.6\linewidth]{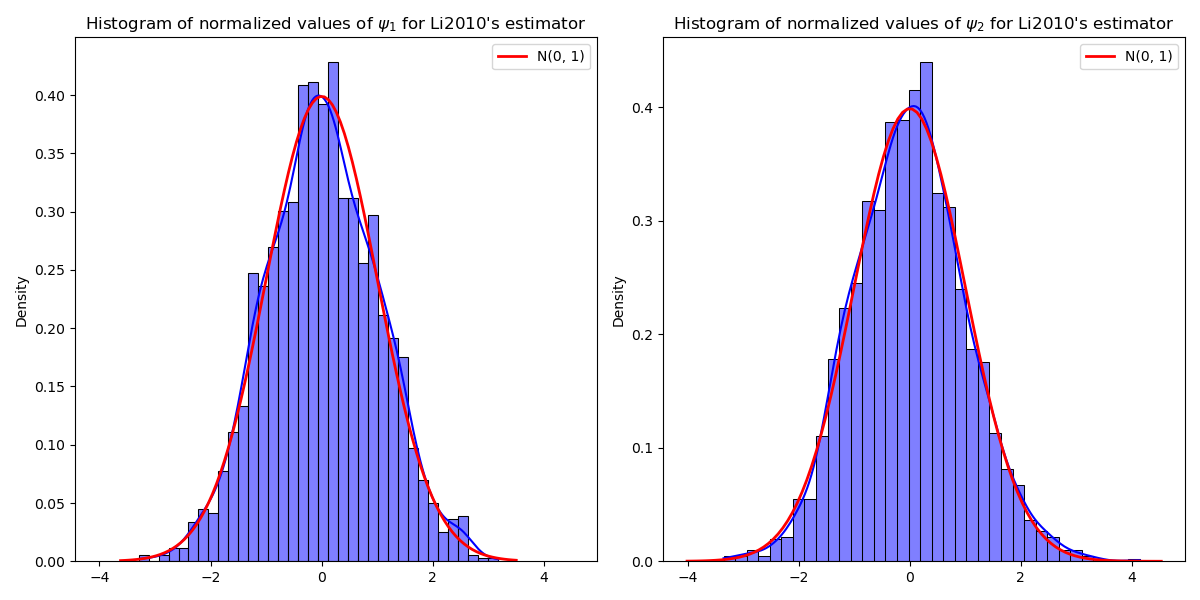}
	\caption{Histograms of normalized values of our MSCD estimator.}
 \label{fig:hist-auction-MSCD-theta_10_05_init_00}
    \end{subfigure} \\
    \begin{subfigure}[h]{\textwidth}
    \centering
	\includegraphics[width=0.6\linewidth]{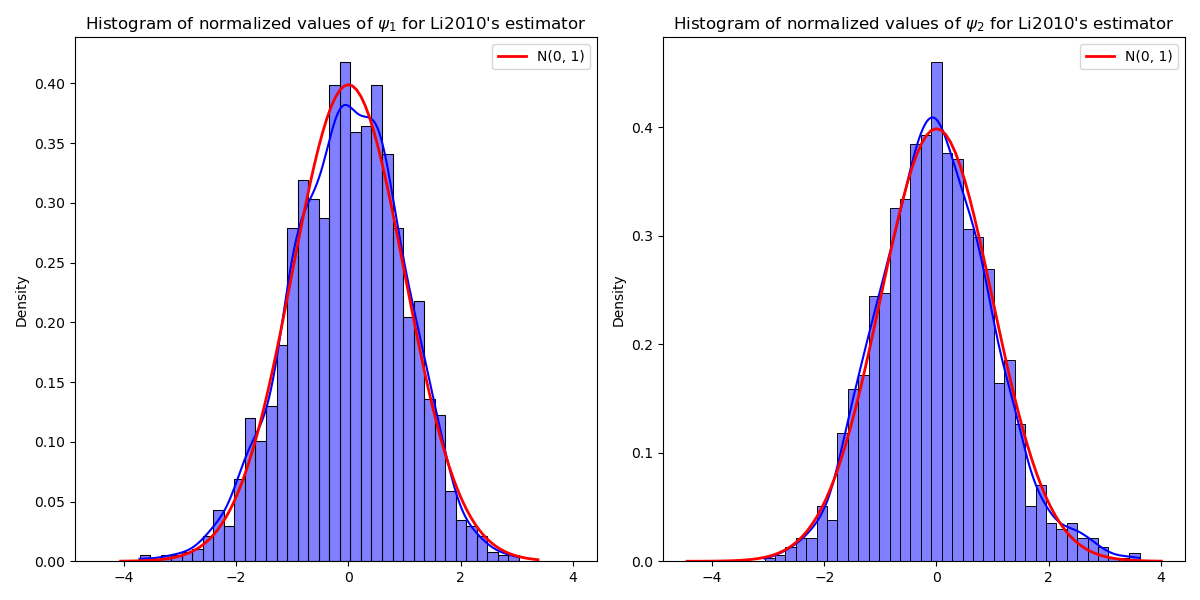}
	\caption{Histograms of normalized values of \cite{Li2010}'s estimator.}
    \label{fig:hist-auction-Li2010-theta_10_05_init_00}
    \end{subfigure}
    \caption{Histograms of normalized values of our MSCD estimator and \cite{Li2010}'s estimator when $\psi = (1, 0.5)$ and initial value is set to $(0, 0)$. The red line shows the density of standard normal distribution.}
    \label{fig:hist-auction-theta_10_05_init_00}
\end{figure}

\begin{figure}[h]
    \centering
    \begin{subfigure}[h]{\textwidth}
	\centering
	\includegraphics[width=0.6\linewidth]{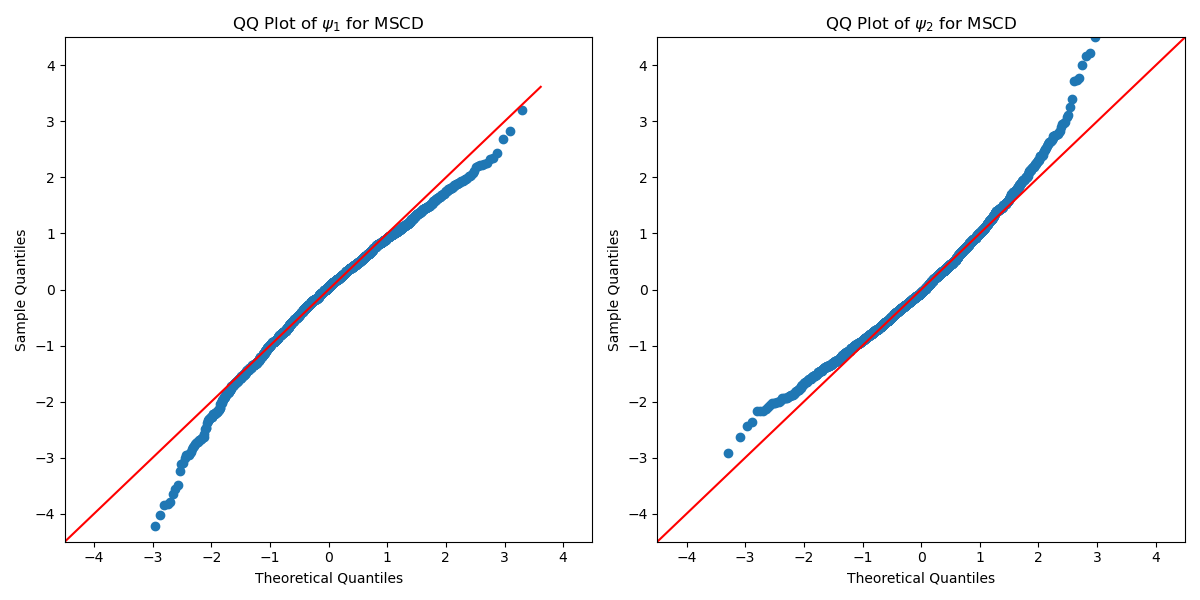}
    \caption{QQ plots of normalized values of our MSCD estimator.}
	\label{fig:qqplot-auction-MSCD-theta_10_30_init_00}
    \end{subfigure} \\
    \begin{subfigure}[h]{\textwidth}
	\centering
	\includegraphics[width=0.6\linewidth]{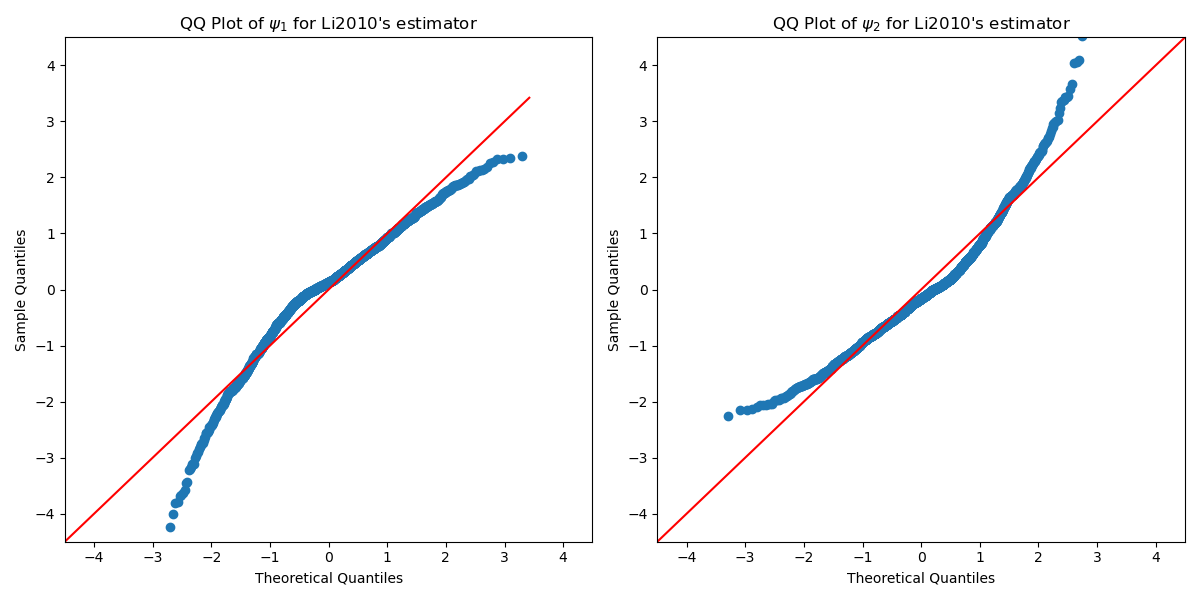}
    \caption{QQ plots of normalized values of \cite{Li2010}'s estimator.}
	\label{fig:qqplot-auction-Li2010-theta_10_30_init_true_value}
    \end{subfigure}
    \caption{QQ plots of normalized values of our MSCD estimator and \cite{Li2010}'s estimator when $T=100$, $\psi = (1, 3)$. For \cite{Li2010}'s estimator, initial value is set to true value.}
    \label{fig:qqplot-auction-theta_10_30_n100}
\end{figure}

\begin{figure}
    \centering
   \begin{subfigure}[h]{\textwidth}
	\centering
	\includegraphics[width=0.6\linewidth]{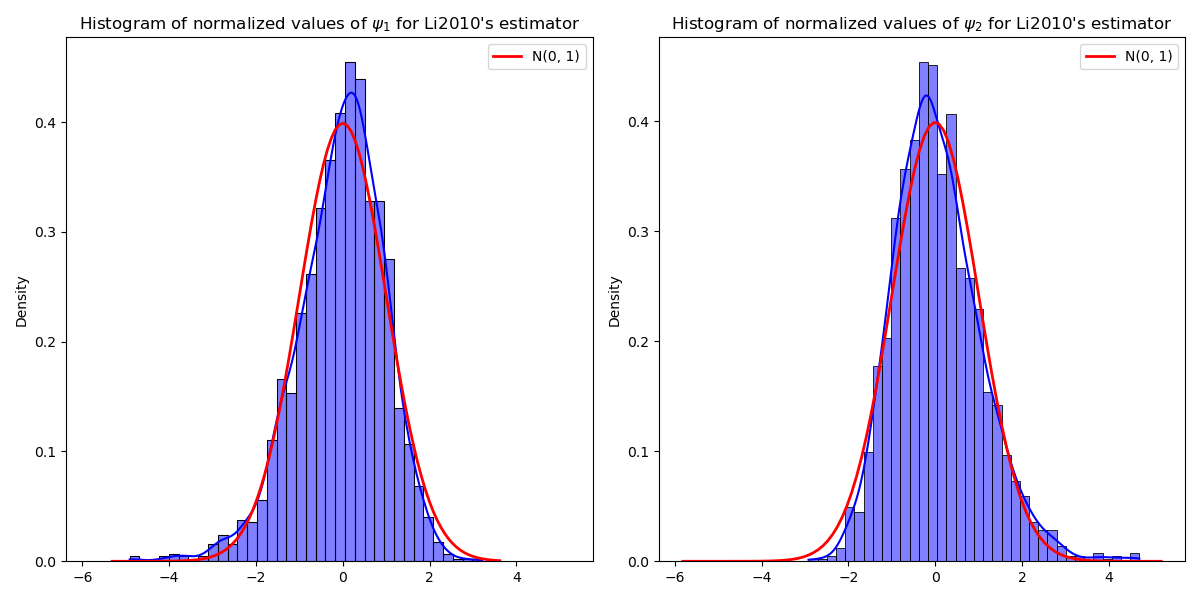}
    \caption{Histograms of normalized values of our MSCD estimator.}
	\label{fig:hist-auction-MSCD-theta_10_30_init_00}
    \end{subfigure} \\
    \begin{subfigure}[h]{\textwidth}
	\centering
	\includegraphics[width=0.6\linewidth]{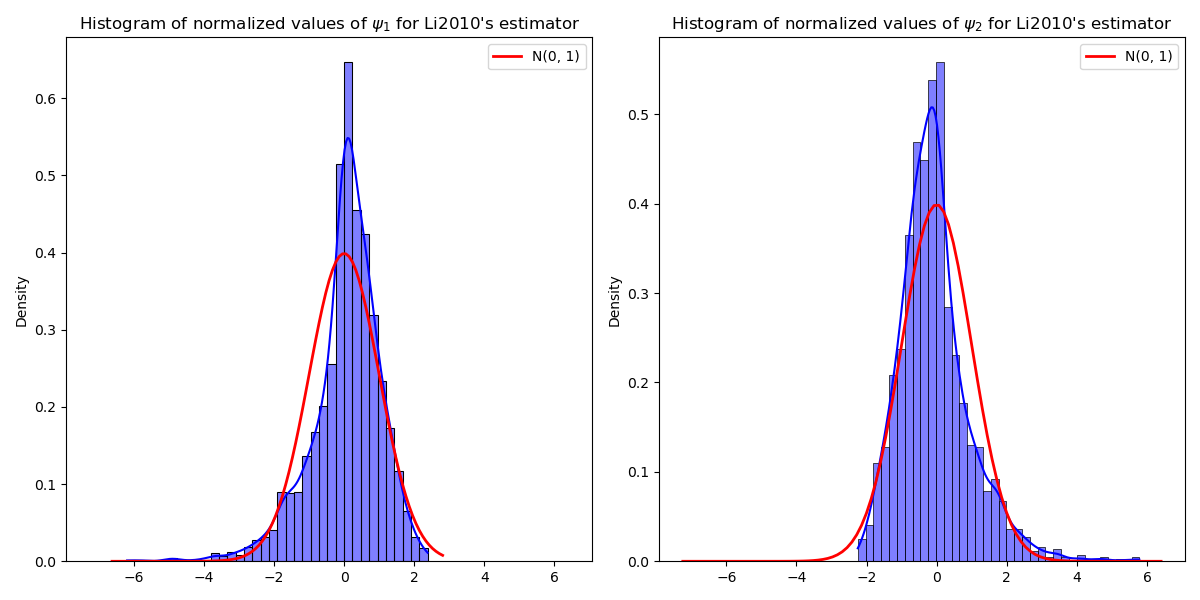}
     \caption{Histograms of normalized values of \cite{Li2010}'s estimator.}
	\label{fig:hist-auction-Li2010-theta_10_30_init_true_value}
    \end{subfigure}
    \caption{Histograms of normalized values of our MSCD estimator and \cite{Li2010}'s estimator when $T=100$, $\psi = (1, 3)$. The red line shows the density of standard normal distribution. For \cite{Li2010}'s estimator, the initial value is set to true value.}
    \label{fig:hist-auction-theta_10_30_n100}
\end{figure}

\begin{figure}[h]
    \centering
    \begin{subfigure}[h]{\textwidth}
	\centering
	\includegraphics[width=0.6\linewidth]{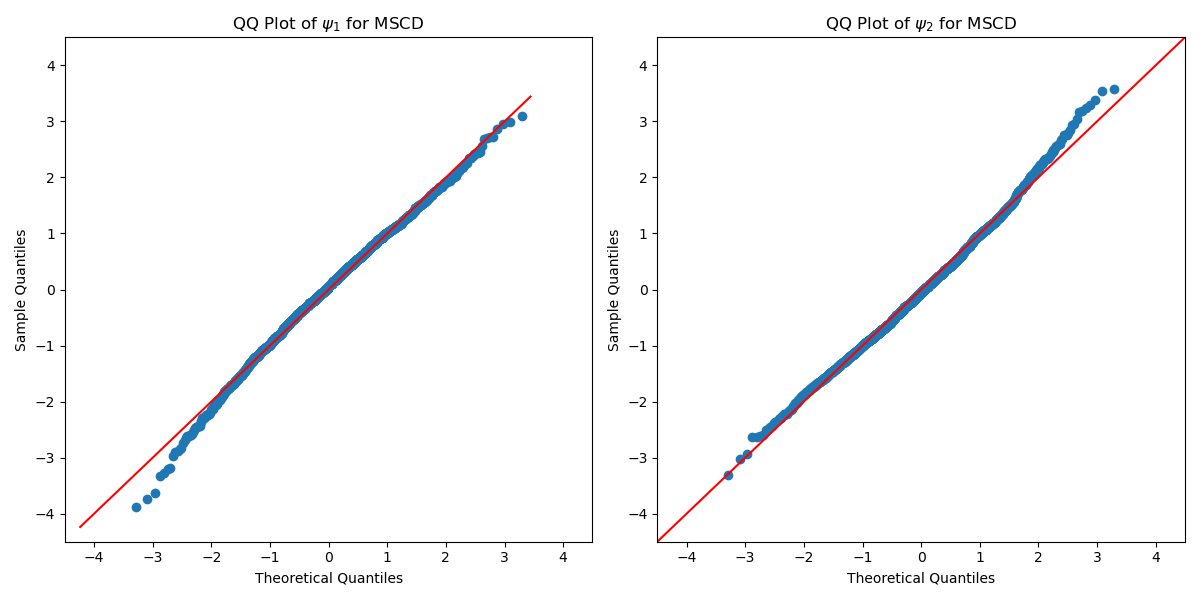}
    \caption{QQ plots of normalized values of our MSCD estimator.}
	\label{fig:qqplot-auction-MSCD-theta_10_30_init_00-n200}
    \end{subfigure} 
    \\
    \begin{subfigure}[h]{\textwidth}
	\centering
	\includegraphics[width=0.6\linewidth]{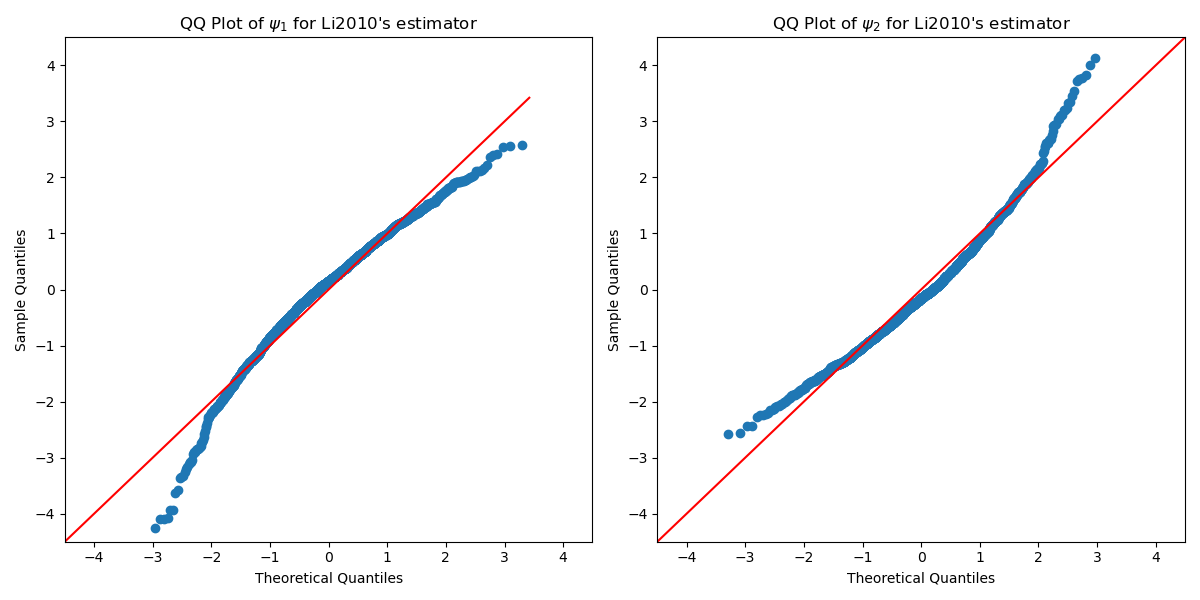}
     \caption{QQ plots of normalized values of \cite{Li2010}'s estimator.}
	\label{fig:qqplot-auction-Li2010-theta_10_30_init_true_value-n200}
    \end{subfigure}
    \caption{QQ plots of normalized values of our MSCD estimator and \cite{Li2010}'s estimator when $T=200$ and $\psi = (1, 3)$. For \cite{Li2010}'s estimator, the initial value is set to true value.}
    \label{fig:qqplot-auction-theta_10_30-n200}
\end{figure}

\begin{figure}
    \centering
    \begin{subfigure}[h]{\textwidth}
	\centering
	\includegraphics[width=0.6\linewidth]{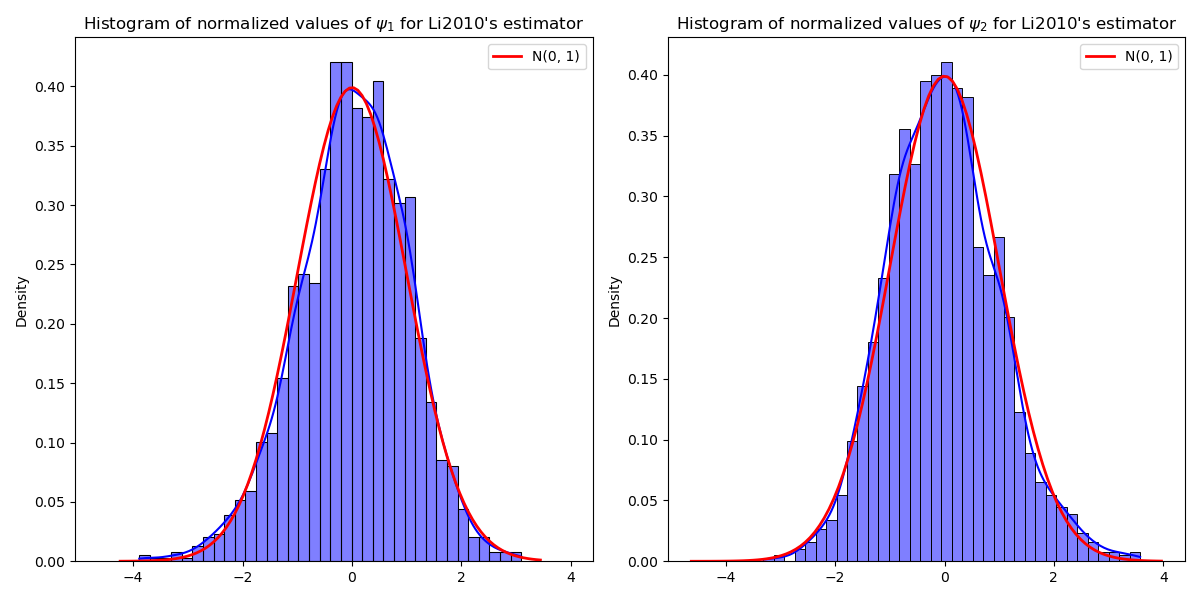}
	\caption{Histograms of normalized values of our MSCD estimator.}
	\label{fig:hist-auction-MSCD-theta_10_30_init_00-n200}
    \end{subfigure} 
    \\
    \begin{subfigure}[h]{\textwidth}
	\centering
	\includegraphics[width=0.6\linewidth]{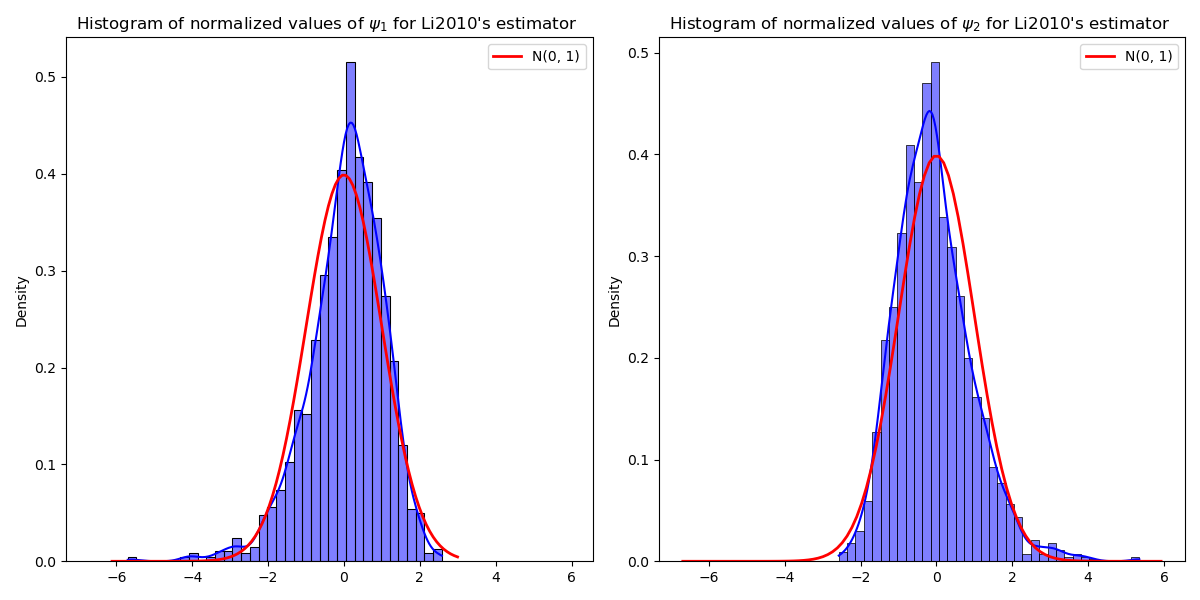}
    \caption{Histograms of normalized values of \cite{Li2010}'s estimator.}
	\label{fig:hist-auction-Li2010-theta_10_30_init_true_value-n200}
    \end{subfigure} 
    \caption{Histograms of normalized values of our MSCD estimator and \cite{Li2010}'s estimator when $T=200$ and $\psi = (1, 3)$. The red line shows the density of standard normal distribution. For \cite{Li2010}'s estimator, the initial value is set to true value.}
    \label{fig:hist-auction-theta_10_30-n200}
\end{figure}

\clearpage

\bibliographystyle{plainnat}
\bibliography{SW-PDS-reference.bib}

\begin{thebibliography}{29}
\providecommand{\natexlab}[1]{#1}
\providecommand{\url}[1]{\texttt{#1}}
\expandafter\ifx\csname urlstyle\endcsname\relax
  \providecommand{\doi}[1]{doi: #1}\else
  \providecommand{\doi}{doi: \begingroup \urlstyle{rm}\Url}\fi

\bibitem[Ambrosio et~al.(2008)Ambrosio, Gigli, and Savare]{Ambrosio2008}
Luigi Ambrosio, Nicola Gigli, and Giuseppe Savare.
\newblock \emph{{Gradient Flows}}.
\newblock Birkh{\"{a}}user-Verlag, 2008.
\newblock ISBN 9783764373092.
\newblock \doi{10.1007/b137080}.

\bibitem[Andrews(1999)]{Andrews1999}
Donald W~K Andrews.
\newblock {Estimation When a Parameter is on a Boundary}.
\newblock \emph{Econometrica}, 67\penalty0 (6):\penalty0 1341--1383, 1999.
\newblock ISSN 0012-9682.
\newblock \doi{10.1111/1468-0262.00082}.

\bibitem[Bernton et~al.(2019)Bernton, Jacob, Gerber, and Robert]{Bernton2019}
Espen Bernton, Pierre~E Jacob, Mathieu Gerber, and Christian~P Robert.
\newblock {On Parameter Estimation with the Wasserstein Distance}.
\newblock \emph{Information and Inference: A Journal of the IMA}, 8\penalty0 (4):\penalty0 657--676, 2019.
\newblock ISSN 2049-8764.
\newblock \doi{10.1093/imaiai/iaz003}.

\bibitem[Bobkov and Ledoux(2019)]{Bobkov2019}
Sergey Bobkov and Michel Ledoux.
\newblock {One-Dimensional Empirical Measures, Order Statistics, and Kantorovich Transport Distances}.
\newblock \emph{Memoirs of the American Mathematical Society}, 261\penalty0 (1259):\penalty0 0, 2019.
\newblock \doi{10.1090/memo/1259}.

\bibitem[Bonneel et~al.(2015)Bonneel, Rabin, Peyr{\'{e}}, and Pfister]{Bonneel2015}
Nicolas Bonneel, Julien Rabin, Gabriel Peyr{\'{e}}, and Hanspeter Pfister.
\newblock {Sliced and Radon Wasserstein Barycenters of Measures}.
\newblock \emph{Journal of Mathematical Imaging and Vision}, 51\penalty0 (1), 1 2015.
\newblock ISSN 0924-9907.
\newblock \doi{10.1007/s10851-014-0506-3}.

\bibitem[Bowlus et~al.(2001)Bowlus, Kiefer, and Neumann]{Bowlus2001}
Audra~J. Bowlus, Nicholas~M. Kiefer, and George~R. Neumann.
\newblock {Equilibrium search models and the transition from school to work}.
\newblock \emph{International Economic Review}, 42\penalty0 (2), 2001.
\newblock ISSN 00206598.
\newblock \doi{10.1111/1468-2354.00112}.

\bibitem[Briol et~al.(2019)Briol, Barp, Duncan, and Girolami]{Briol2019}
Francois-Xavier Briol, Alessandro Barp, Andrew~B. Duncan, and Mark Girolami.
\newblock {Statistical Inference for Generative Models with Maximum Mean Discrepancy}.
\newblock 6 2019.

\bibitem[Chernozhukov and Hong(2004)]{Chernozhukov2004}
Victor Chernozhukov and Han Hong.
\newblock {Likelihood Estimation and Inference in a Class of Nonregular Econometric Models}.
\newblock \emph{Econometrica}, 72\penalty0 (5):\penalty0 1445--1480, 2004.
\newblock \doi{10.1111/j.1468-0262.2004.00540.x}.

\bibitem[Cram{\'{e}}r(1928)]{Cramer_1928}
Harald Cram{\'{e}}r.
\newblock {On the Composition of Elementary Errors: Second Paper: Statistical Applications}.
\newblock \emph{Scandinavian Actuarial Journal}, 1928\penalty0 (1):\penalty0 141--180, January 1928.
\newblock ISSN 1651-2030.
\newblock \doi{10.1080/03461238.1928.10416872}.

\bibitem[Cs{\"{o}}rg{\H{o}}(1983)]{Csorgo1983}
Mikl{\'{o}}s Cs{\"{o}}rg{\H{o}}.
\newblock \emph{{Quantile Processes with Statistical Applications}}.
\newblock Society for Industrial and Applied Mathematics, 1 1983.
\newblock ISBN 978-0-89871-185-1.
\newblock \doi{10.1137/1.9781611970289}.

\bibitem[Deshpande et~al.(2018)Deshpande, Zhang, and Schwing]{Deshpande2018}
Ishan Deshpande, Ziyu Zhang, and Alexander Schwing.
\newblock {Generative Modeling Using the Sliced Wasserstein Distance}.
\newblock In \emph{Proceedings of the IEEE Computer Society Conference on Computer Vision and Pattern Recognition}, 2018.
\newblock \doi{10.1109/CVPR.2018.00367}.

\bibitem[Donald and Paarsch(2002)]{Donald2002}
Stephen~G Donald and Harry~J Paarsch.
\newblock {Superconsistent Estimation and Inference in Structural Econometric Models Using Extreme Order Statistics}.
\newblock \emph{Journal of Econometrics}, 109\penalty0 (2):\penalty0 305--340, 2002.
\newblock \doi{10.1016/s0304-4076(02)00116-1}.

\bibitem[Goodfellow et~al.(2014)Goodfellow, Pouget-Abadie, Mirza, Xu, Warde-Farley, Ozair, Courville, and Bengio]{Goodfellow2014}
Ian Goodfellow, Jean Pouget-Abadie, Mehdi Mirza, Bing Xu, David Warde-Farley, Sherjil Ozair, Aaron Courville, and Yoshua Bengio.
\newblock {Generative Adversarial Nets}.
\newblock In Z~Ghahramani, M~Welling, C~Cortes, N~Lawrence, and K~Q Weinberger, editors, \emph{Advances in Neural Information Processing Systems}, volume~27. Curran Associates, Inc., 2014.

\bibitem[Hirano and Porter(2003)]{Hirano2003}
Keisuke Hirano and Jack~R Porter.
\newblock {Asymptotic Efficiency in Parametric Structural Models with Parameter-Dependent Support}.
\newblock \emph{Econometrica}, 71\penalty0 (5):\penalty0 1307--1338, 2003.
\newblock \doi{10.1111/1468-0262.00451}.

\bibitem[Kaji et~al.(2020)Kaji, Manresa, and Pouliot]{Kaji2020}
Tetsuya Kaji, Elena Manresa, and Guillaume Pouliot.
\newblock {An Adversarial Approach to Structural Estimation}.
\newblock \emph{SSRN Electronic Journal}, 2020.
\newblock ISSN 1556-5068.
\newblock \doi{10.2139/ssrn.3706365}.

\bibitem[Li(2010)]{Li2010}
Tong Li.
\newblock {Indirect Inference in Structural Econometric Models}.
\newblock \emph{Journal of Econometrics}, 157\penalty0 (1):\penalty0 120--128, 7 2010.
\newblock ISSN 03044076.
\newblock \doi{10.1016/j.jeconom.2009.10.027}.

\bibitem[Nadjahi et~al.(2020{\natexlab{a}})Nadjahi, Durmus, Chizat, Kolouri, Shahrampour, and {\c{S}}im{\c{s}}ekli]{Nadjahi2020SDProperty}
Kimia Nadjahi, Alain Durmus, Lénaïc Chizat, Soheil Kolouri, Shahin Shahrampour, and Umut {\c{S}}im{\c{s}}ekli.
\newblock {Statistical and Topological Properties of Sliced Probability Divergences}.
\newblock 3 2020{\natexlab{a}}.

\bibitem[Nadjahi et~al.(2020{\natexlab{b}})Nadjahi, Durmus, {\c{S}}im{\c{s}}ekli, and Badeau]{Nadjahi2020AsymSWD}
Kimia Nadjahi, Alain Durmus, Umut {\c{S}}im{\c{s}}ekli, and Roland Badeau.
\newblock {Asymptotic Guarantees for Learning Generative Models with the Sliced-Wasserstein Distance}.
\newblock 6 2020{\natexlab{b}}.

\bibitem[Newey(1991)]{Newey1991}
Whitney~K. Newey.
\newblock {Uniform Convergence in Probability and Stochastic Equicontinuity}.
\newblock \emph{Econometrica}, 59\penalty0 (4), 7 1991.
\newblock ISSN 00129682.
\newblock \doi{10.2307/2938179}.

\bibitem[Paarsch(1992)]{Paarsch1992}
Harry~J Paarsch.
\newblock {Deciding between the Common and Private Value Paradigms in Empirical Models of Auctions}.
\newblock \emph{Journal of Econometrics}, 51\penalty0 (1-2):\penalty0 191--215, 10 1992.
\newblock \doi{10.1016/0304-4076(92)90035-p}.

\bibitem[Park(2022)]{Park2022_dissertation}
Hyeonseok Park.
\newblock {Non-Likelihood Based Methods for the Estimation and Inference in Econometric Models}, 2022.

\bibitem[Pollard(1980)]{Pollard1980}
D~Pollard.
\newblock {The Minimum Distance Method of Testing}.
\newblock \emph{Metrika}, 27\penalty0 (1):\penalty0 43--70, 1980.
\newblock ISSN 0026-1335.
\newblock \doi{10.1007/bf01893576}.

\bibitem[Santambrogio(2015)]{Santambrogio2015}
Filippo Santambrogio.
\newblock \emph{{Optimal Transport for Applied Mathematicians}}, volume~87 of \emph{Progress in Nonlinear Differential Equations and Their Applications}.
\newblock Springer International Publishing, Cham, 2015.
\newblock ISBN 978-3-319-20827-5.
\newblock \doi{10.1007/978-3-319-20828-2}.

\bibitem[Sherman(1994)]{Sherman1994}
Robert~P. Sherman.
\newblock {Maximal Inequalities for Degenerate {$U$}-Processes with Applications to Optimization Estimators}.
\newblock \emph{The Annals of Statistics}, 22\penalty0 (1), 3 1994.
\newblock ISSN 0090-5364.
\newblock \doi{10.1214/aos/1176325377}.

\bibitem[Shorack and Wellner(2009)]{Shorack2009}
Galen~R Shorack and Jon~A Wellner.
\newblock \emph{{Empirical Processes with Applications to Statistics}}.
\newblock Society for Industrial and Applied Mathematics, 10 2009.
\newblock ISBN 0898716845.
\newblock \doi{10.1137/1.9780898719017}.

\bibitem[Sz{\'{e}}kely and Rizzo(2017)]{Szekely_2017}
G{\'{a}}bor~J. Sz{\'{e}}kely and Maria~L. Rizzo.
\newblock {The Energy of Data}.
\newblock \emph{Annual Review of Statistics and Its Application}, 4\penalty0 (1):\penalty0 447--479, March 2017.
\newblock ISSN 2326-831X.
\newblock \doi{10.1146/annurev-statistics-060116-054026}.

\bibitem[van~de Vaart(1998)]{Vaart1998}
Aad~W. van~de Vaart.
\newblock \emph{{Asymptotic Statistics}}.
\newblock Cambridge University Press, 10 1998.
\newblock ISBN 9780521496032.
\newblock \doi{10.1017/CBO9780511802256}.

\bibitem[van~der Vaart and Wellner(1996)]{Vaart1996}
Aad~W. van~der Vaart and Jon~A. Wellner.
\newblock \emph{{Weak Convergence and Empirical Processes}}.
\newblock Springer Series in Statistics. Springer New York, New York, NY, 1996.
\newblock ISBN 978-1-4757-2547-6.
\newblock \doi{10.1007/978-1-4757-2545-2}.

\bibitem[Zhu et~al.(1997)Zhu, Fang, and Bhatti]{Zhu1997}
Li~Xing Zhu, Kai~Tai Fang, and M.~Ishaq Bhatti.
\newblock {On Estimated Projection Pursuit-Type Cr{\'{a}}mer-Von Mises Statistics}.
\newblock \emph{Journal of Multivariate Analysis}, 63\penalty0 (1), 1997.
\newblock ISSN 0047259X.
\newblock \doi{10.1006/jmva.1997.1673}.

\end{thebibliography}

\begin{appendices}

\section{Verification of the Remaining Assumptions in One-sided and Two-Sided Uniform Models}
\label{appendix:MSWD-MSCD}

In this section, we will verify that the assumptions for consistency and
asymptotic normality of MSWD and MSCD estimators in Section \ref{section:asym-theory} are satisfied in
one-sided and two-sided uniform models.

We will verify Assumptions \ref{assumption:norm-diff}, \ref%
{assumption:identification}, \ref{assumption:big-oh} (i), and \ref%
{assumption:CLT}, and \ref{assumption:hessian-pd} because all two examples
are unconditional models, and Assumption \ref{assumption:convergence} (i) is
implied by Assumption \ref{assumption:big-oh}(i).

Since we deal with unconditional models for the univariate variable $Y_t$, we can take $\widehat{Q}(s; u, \psi) := Q(s; \psi)$, and $\widehat{D}(\cdot, u, \psi_0) = D(\cdot, \psi_0)$ where $D(\cdot, \psi_0)$ is a deterministic $L_2(\mathcal{S}, w(s)\mathrm{d}s)$-measurable function. $Q(s; \psi)$ will be the parametric distribution function of $Y_t$ for MSCD estimators, and the parametric quantile function for MSWD estimators.
In one-sided and two-sided models, we will consider the case where $w(s) = 1$.

We note here that when $\widehat{Q}_{T}(\cdot ;u,\psi )=Q(\cdot ;u,\psi )$
such as in unconditional models, we have $\widehat{D}_{T}(s;u,\psi
_{0})=D(s;u,\psi _{0})$ so Assumptions \ref{assumption:big-oh} (ii) and
(iii) automatically hold. \

\subsection{One-sided Uniform Model}

\label{sec:one-sided}

We deal with real-valued random variable, and do not require projection. We can take $\widehat{Q}_T(\cdot, u, \psi) = Q(\cdot, \psi)$. We assume that $\psi_0 > 0$.

\paragraph{The MSCD Estimator}

In this case, $Q(s, \psi_0) = F(s, \psi_0)$. Let $F_T(s) = I(Z_t \le
s)$. Some algebra can show 
\begin{align*}
\mathcal{C}_2^2(\mu_{\psi}, \mu_0) = \int_{-\infty}^{\infty} (F(s ;
\psi) - F(s; \psi_0))^2 \mathrm{d}s = 
\begin{cases}
\frac{(\psi - \psi_0)^2}{3\psi_0} & \text{ if } \psi \le \psi_0,
\\ 
\frac{(\psi - \psi_0)^2}{3\psi} & \text{ if } \psi > \psi_0.%
\end{cases}%
\end{align*}
Since 
\begin{align*}
\frac{\partial \mathcal{C}_2^2(\mu_{\psi}, \mu_0) }{\partial \psi} = 
\begin{cases}
\frac{2 (\psi - \psi_0)}{3 \psi_0} < 0 & \text{ if } \psi < \psi_0,
\\ 
0 & \text{ if } \psi = \psi_0, \\ 
\frac{\psi^2 - \psi_0^2}{3 \psi^2} > 0 & \text{ if } \psi > \psi_0.%
\end{cases}%
\end{align*}
Assumption \ref{assumption:identification} holds when $\psi_0 > 0$.

Because $\{Y_t\}_{t=1}^T$ is a random sample, 
$\{\sqrt{T}(F_T(s) - F(s, \psi_0)): s \in \mathbb{R}\}$
is P-Donsker, and 
\begin{align*}
T \int_{-\infty}^{\infty} ((F_T(s) - F(s, \psi_0)))^2 \mathrm{d}s 
& = T
\int_{0}^{\psi_0} ((F_T(s) - F(s, \psi_0)))^2 \mathrm{d}s \\
& \le \psi_0 %
\left[\sup_{s \in \mathbb{R}} (\sqrt{T}(F_T(s) - F(s, \psi_0))) \right]^2
= O_p(1).
\end{align*}
Therefore, Assumption \ref{assumption:big-oh} (i) holds.

Example \ref{example:uniform-models-norm-diff} (i) shows that $F(\cdot, \psi)$ is
norm-differentiable at $\psi = \psi_0$, and $D(s; \psi_0)$ is $L_2(%
\mathbb{R}, \mathrm{d}s)$%
-measurable, and
\begin{align*}
\int_{-\infty}^{\infty} D^2(s; \psi_0) 
\mathrm{d}s  = \int_{0}^{\psi_0} \frac{s^2}{%
\psi_0^4} \mathrm{d}s = \frac{1}{3\psi_0} < \infty.
\end{align*}
Therefore, Assumptions \ref{assumption:norm-diff} and \ref%
{assumption:hessian-pd} hold.

Since $\{\sqrt{T}(F_T(s) - F(s, \psi_0)): s \in \mathbb{R}\}$ is P-Donsker and
$
	\int_{u \in \mathbb{S}^{d-1}}
	\int_{-\infty}^{\infty} D^2(s; \psi_0) \mathrm{d}s \mathrm{d}\varsigma(u) < \infty, 
$
Assumption \ref{assumption:CLT} is satisfied. %
That is, $\sqrt{T} \int_{-\infty}^{\infty} (F_T(s) - F(s, \psi_0)) D(s; \psi_0) d s$ follows normal distribution.

\paragraph{The MSWD Estimator}

In this case, $Q(s, \psi_0) = F^{-1}(s, \psi_0)$, and $Q_T(s) = F_T^{-1}(s)$, where $F_T(s) = \frac{1}{T}\sum_{t=1}^T I(Z_t \le s)$.

Since 
\begin{align*}
\int_{0}^{1} ((F^{-1}(s, \psi_0) - F^{-1}(s, \psi)))^2 \mathrm{d}s =
(\psi - \psi_0)^2,
\end{align*}
Assumption \ref{assumption:identification} holds.

Assumptions \ref{assumption:norm-diff} and \ref{assumption:hessian-pd} hold
with $D(s, u, \psi_0) = s$ because 
$F^{-1}(s, \psi) - F^{-1}(s, \psi_0) = s (\psi - \psi_0)$.

Since $\{\sqrt{T}(F_T(s) - F(s, \psi_0)): s \in \mathbb{R}\}$ is
P-Donsker, and $0 < \frac{\partial F^{-1}(s, \psi_0)}{\partial s} = \psi_0 <
\infty$ and the support is compact, $\sqrt{T}(F_T^{-1}(\cdot) - F^{-1}(\cdot, \psi_0)) \Rightarrow \psi_0\mathbb{G}(\cdot)$, where $\mathbb{G}$ is the standard Brownian bridge by Section 3.9.4.2 of \cite{Vaart1996}. This implies that Assumptions \ref{assumption:big-oh} (i) and \ref{assumption:CLT} hold as $D(s, \psi_0)$ is bounded.

\subsection{Two-sided Uniform Model}

\label{sec:two-sided-uniform}

We deal with real-valued random variable, and do not require projection. We can take $\widehat{Q}_T(\cdot, u, \psi) = Q(\cdot, \psi)$. We assume that $0 < \psi_0 < 1$.

\paragraph{The MSCD Estimator}

In this case, $Q(s, \psi_0) = F(s, \psi_0)$. Let $F_T(s) = I(Z_t \le
s)$. Because $\{Z_t\}_{t=1}^T$ is a random sample, 
$
\{\sqrt{T}(F_T(s) - F(s, \psi_0)): s \in \mathbb{R}\}
$
is P-Donsker, and 
\begin{align*}
T \int_{-\infty}^{\infty} ((F_T(s) - F(s, \psi_0)))^2 \mathrm{d}s = T
\int_{0}^{1} ((F_T(s) - F(s, \psi_0)))^2 \mathrm{d}s \le \left[\sup_{s \in 
\mathbb{R}} (\sqrt{T}(F_T(s) - F(s, \psi_0))) \right]^2 = O_p(1).
\end{align*}
Therefore, Assumption \ref{assumption:big-oh} (i) holds.

Example \ref{sec:two-sided-uniform} (ii)
 shows that $Q(s; \psi)$ is
norm-differentiable at $\psi = \psi_0$, and $D(s; \psi_0)$ is $L_2(%
\mathbb{R} \times \mathbb{S}^{d-1}, \mathrm{d}s\mathrm{d}\varsigma(u))$%
-measurable. Therefore, Assumption \ref{assumption:norm-diff} is satisfied. Assumption \ref{assumption:CLT} is satisfied
\begin{align*}
    \int_{0}^{1} (F_T(s) - F(s, \psi_0)) D(s, \psi_0) \mathrm{d}s = \frac{1}{\sqrt{T}}\sum_{t=1}^{T} \int_{0}^{1} (\mathds{1}(Z_t \le s) - F(s, \psi_0)) D(s, \psi_0) \mathrm{d}s \xrightarrow{d} N(0, \sigma^2)
\end{align*}
for some $\sigma^2 \ge 0$ by the central limit theorem as $\int D^2(s, \psi_0) ds$ is bounded.

\paragraph{The MSWD Estimator}

In MSWD estimator, $Q(s, \psi_0) = F^{-1}(s, \psi_0)$.

Since 
\begin{align*}
\int_0^1 (F^{-1}(s; \psi) - F^{-1}(s; \psi_0))^2 \mathrm{d}s = \frac{1}{3%
} (\psi - \psi_0)^2,
\end{align*}
Assumption \ref{assumption:identification} holds.

Since $F^{-1}(s; \psi)$ is linear in $\psi$, $Q(s, u,
\psi)$ is norm-differentiable at $\psi = \psi_0$ with 
\begin{align*}
D(s, \psi_0) = 
\begin{cases}
4 s & \text{ if } 0 \le s \le 1/4, \\ 
\frac{4}{3}(1 - s) & \text{ if } 1/4 \le s \le 1. 
\end{cases}%
\end{align*}
Therefore, Assumptions \ref{assumption:norm-diff} and \ref%
{assumption:hessian-pd} hold.
 
Following Proposition A.18 of \cite{Bobkov2019}, 
\begin{align*}
|F_T^{-1}(s) - F^{-1}(s)| & = |F^{-1}(F(F_T^{-1}(s))) - F^{-1}(s)| \\
& = \left|\int_{s}^{F(F_T^{-1}(s))} \frac{1}{f(F^{-1}(z))} \mathrm{d}z\right| \\
& \le \frac{1}{\min f(s)} |
F(F_T^{-1}(s)) - s| \\
& \le C |%
F(F_T^{-1}(s)) - s|. 
\end{align*}
where $C = \max\{4\psi_0, 4(1-\psi_0)/3\}$. The last inequality holds because $f(\cdot) = f(\cdot; \psi_0)$ under Assumption \ref{assumption:DGP}, and $1/f(\cdot; \psi_0) \le \max\{4\psi_0, 4(1-\psi_0)/3\}$.
This implies that
\begin{align*}
    T \int_0^1 (F_T^{-1}(s; \psi) - F^{-1}(s; \psi_0))^2 \mathrm{d}s \le C^2 \left( \sup_{s \in [0, 1]} |%
\sqrt{T}(F(F_T^{-1}(s), \psi_0) - s)|\right)^2.
\end{align*}
Equations (1.1.6) and (1.4.5) in \cite{Csorgo1983} imply that
\begin{align*}
    \sup_{s \in [0, 1]} |%
\sqrt{T}(F(F_T^{-1}(s), \psi_0) - s)| = \sup_{s \in [0, 1]} |%
\sqrt{T}(F_T(F^{-1}(s; \psi_0)) - s)| = \sup_{s \in \mathbb{R}} |%
\sqrt{T}(F_T(s) - F(s; \psi_0))| = O_p(1).
\end{align*}
The last equality holds because $\{\sqrt{T}(F_T(s) - F(s, \psi_0)): s \in \mathbb{R}\}$ is
P-Donsker. Therefore, Assumption \ref{assumption:big-oh} (i) holds.

Since $D(s, \psi_0)$ and $F^{-1}(s, \psi_0)$ are bounded by some absolute constants and $D(s, \psi_0)$ is continuous except $s = 1/4$, we can show Assumption \ref{assumption:CLT} by Theorem 1 (i) and Remark 2 in Chapter 19 of \cite{Shorack2009}.

\section{Verification of Assumptions for Two-sided Parameter-Dependent Support Model in Example 2.1}
\label{appendix:two-sided}

It follows from equation (\ref{eq:two-sided}) that when $y<g(x,\psi )$, 
\begin{equation*}
\frac{\partial F(y|X,\psi )}{\partial \psi }=-\frac{\partial g(x,\psi )}{%
\partial \psi }f_{L,\epsilon }(y-g(x,\psi )|x,\psi )+\int_{-\infty
}^{y-g(x,\psi )}\frac{\partial f_{L,\epsilon }(\epsilon |x,\psi )}{\partial
\psi }\mathrm{d}\epsilon ;
\end{equation*}%
when $y>g(x,\psi )$, 
\begin{equation*}
\frac{\partial F(y|X,\psi )}{\partial \psi }=-\frac{\partial g(X,\psi )}{%
\partial \psi }f_{U,\epsilon }(y-g(X,\psi )|X,\psi )+\int_{-\infty }^{0}%
\frac{\partial f_{L,\epsilon }(\epsilon |X,\psi )}{\partial \psi }\mathrm{d}%
\epsilon +\int_{0}^{y-g(x,\psi )}\frac{\partial f_{U,\epsilon }(\epsilon
|x,\psi )}{\partial \psi }\mathrm{d}\epsilon .
\end{equation*}

\begin{lemma}[two-sided model]
\label{lemma:Lipchitz-two-sided} Assume that $%
f_{L,\epsilon }(\epsilon |x.\psi )$ and $f_{U,\epsilon }(\epsilon |x,\psi )$
are continuous in $(\epsilon ,\psi )$ for each $x$, and
and $f_{L, \epsilon }(\epsilon |x,\psi )$ and $f_{R, \epsilon }(\epsilon |x,\psi )$ are differentiable with respect to $\psi$ for each $\epsilon$ and $\psi$; 
and $g(x, \psi)$ is continuously differentiable with respect to $\psi$ for each $x$, 
and there exist $\bar{f}%
_{L,\epsilon }(y|x)$ and $\bar{f}_{L,\epsilon }(y|x)$ such that 
\begin{enumerate}
    \item[(i)] $\left\Vert \frac{f_{L,\epsilon }(y|x,\psi )}{\partial \psi }\right\Vert \leq 
\bar{f}_{L,\epsilon }(y|x)$, $\left\Vert \frac{f_{U,\epsilon }(\epsilon
|x,\psi )}{\partial \psi }\right\Vert \leq \bar{f}_{U,\epsilon }(\epsilon
|x)$, \\
     \item[(ii)] $\int_{-\infty }^{0}\bar{f}_{L,\epsilon }(\epsilon |x)\mathrm{d}%
\epsilon <\infty$ and $\int_{0}^{\infty}\bar{f}_{U,\epsilon
}(\epsilon |x)\mathrm{d}\epsilon <\infty$ for each $x$.
\end{enumerate}
In addition, we assume that 
\begin{enumerate}
    \item[(iii)] $\sup_{x,\psi }\left\Vert \frac{\partial g(x,\psi )}{\partial \psi }%
\right\Vert <\infty$, $\quad \sup_{\epsilon ,x,\psi }|f_{L,\epsilon }(\epsilon
|x,\psi )|<\infty$, $\sup_{\epsilon ,x,\psi }|f_{U,\epsilon }(\epsilon
|x,\psi )|<\infty$, \\
     \item[(iv)] $\sup_{x,\psi } \int_{-\infty }^{0}\left\|\frac{%
\partial f_{L,\epsilon }(\epsilon |x,\psi )}{\partial \psi } \right\|\mathrm{d}%
\epsilon  <\infty$ ,$\quad \sup_{x,\psi
}\int_{0}^{\infty}\left\|\frac{\partial f_{U,\epsilon }(\epsilon
|x,\psi )}{\partial \psi } \right\| \mathrm{d}\epsilon  <\infty$.
\end{enumerate}
Then $F(y|x, \psi)$ is Lipschitz continuous with respect to 
$\psi $ uniformly in $y,x$. 
\end{lemma}
\begin{proof}
    The proof is identical to that of Lemma \ref{lemma:Lipchitz-one-sided}.
\end{proof}

Like for one-sided model, the first two assumptions in lemma \ref%
{lemma:Lipchitz-two-sided} ensure validity of interchanging the limit\
operation and integration. The third and assumptions holds under
Conditions C2 and C3 in \cite{Chernozhukov2004}. The fourth and sixth conditions is stronger than the conditions $\sup_{\psi } \mathbb{E} \int_{-\infty }^{0}\left\|\frac{%
\partial f_{L,\epsilon }(\epsilon |x,\psi )}{\partial \psi } \right\|\mathrm{d}%
\epsilon  <\infty$, $\sup_{\psi } \mathbb{E} \int_{0}^{\infty}\left\|\frac{\partial f_{U,\epsilon }(\epsilon
|x,\psi )}{\partial \psi } \right\| \mathrm{d}\epsilon  <\infty$ which are similar to Conditions C2 in \cite{Chernozhukov2004}. 
But we do not
require that 
\begin{equation*}
	\lim_{\epsilon \downarrow 0}f_{U,\epsilon }(\epsilon |x,\psi
)-\lim_{\epsilon \uparrow 0}f_{L,\epsilon }(\epsilon |x,\psi )>\eta >0
\end{equation*}
as in \cite{Chernozhukov2004}.

In the two-sided model, when $y < g(x, \psi)$, 
\begin{align*}
\frac{\partial^2 F(y|x, \psi)}{\partial \psi \partial \psi^{\prime}} & = - 
\frac{\partial^2 g(x, \psi)}{\partial \psi \partial \psi^{\prime}} f_{L,
\epsilon}(y - g(x, \psi)|x, \psi) \\
& \quad + \frac{\partial g(x, \psi)}{\partial \psi} \left[\frac{\partial
g(x, \psi)}{\partial \psi^{\prime}} \frac{\partial f_{L, \epsilon}(y - g(x,
\psi)|x, \psi)}{\partial \epsilon} + \frac{\partial f_{L, \epsilon}(y - g(x, \psi)|x, \psi) }{%
\partial \psi^{\prime}} \right] \\
& \quad - \frac{\partial g(x, \psi)}{\partial \psi} \frac{f_{L, \epsilon}(y
- g(x, \psi)|x, \psi) }{\partial \psi^{\prime}} + \int_{-\infty}^{y - g(x,
\psi)} \frac{\partial^2 f_{L, \epsilon}(\epsilon, x, \psi)}{\partial \psi
\partial \psi^{\prime}} \mathrm{d}\epsilon.
\end{align*}

When $y > g(x, \psi)$, 
\begin{align*}
\frac{\partial^2 F(y|x, \psi)}{\partial \psi \partial \psi^{\prime}} & = - 
\frac{\partial^2 g(x, \psi)}{\partial \psi \partial \psi^{\prime}} f_{U,
\epsilon}(y - g(x, \psi)|x, \psi) \\
& \quad + \frac{\partial g(x, \psi)}{\partial \psi} \left[\frac{\partial
g(x, \psi)}{\partial \psi^{\prime}} \frac{\partial f_{U, \epsilon}(y - g(x,
\psi)|x, \psi)}{\partial \epsilon} + \frac{\partial f_{U, \epsilon}(y - g(x, \psi)|x, \psi) }{%
\partial \psi^{\prime}} \right] \\
& \quad - \frac{\partial g(x, \psi)}{\partial \psi} \frac{f_{U, \epsilon}(y
- g(x, \psi)|x, \psi) }{\partial \psi^{\prime}} + \int_{0}^{y - g(x,
\psi)} \frac{\partial^2 f_{U, \epsilon}(\epsilon, x, \psi)}{\partial \psi
\partial \psi^{\prime}} \mathrm{d}\epsilon \\
& \quad + \int_{-\infty}^{0} \frac{\partial^2 f_{L,
\epsilon}(\epsilon, x, \psi)}{\partial \psi \partial \psi^{\prime}} \mathrm{d%
}\epsilon.
\end{align*}

\begin{lemma}[two-sided model]
For two-sided model, in addition to all conditions in Lemma \ref%
{lemma:Lipchitz-two-sided}, let us assume that $f_{L,\epsilon
}(\epsilon |x,\psi )$ and $ f_{U,\epsilon }(y|x,\psi
)$ are continuously differentiable with respect to $(\epsilon, \psi)$ except for $\epsilon = 0$ for each $x$, and $f_{L, \epsilon }(\epsilon |x,\psi )$ and $f_{U, \epsilon }(\epsilon |x,\psi )$ are second-order continuously differentiable with respect to $\psi$ for each $\epsilon$ and $\psi$; 
and $g(x, \psi)$ is continuously second-order differentiable with respect to $\psi$ for each $x$, 
Also, there exist $\tilde{f}%
_{L,\epsilon }(\epsilon |x)$ and $\tilde{f}_{U,\epsilon }(\epsilon |x)$ such
that 
\begin{enumerate}
    \item[(i)] $\left\Vert \frac{\partial f_{L,\epsilon }(\epsilon |x,\psi )}{\partial \psi
\partial \psi ^{\prime }}\right\Vert \leq \tilde{f}_{L,\epsilon }(\epsilon
|x)$, $\int_{-\infty }^{0}\tilde{f}_{L,\epsilon }(\epsilon |x)\mathrm{d}%
\epsilon <\infty$ for all $x$, 
    \item[(ii)] $\left\Vert \frac{\partial f_{U,\epsilon }(\epsilon |x,\psi )}{\partial \psi
\partial \psi ^{\prime }}\right\Vert \leq \tilde{f}_{U,\epsilon }(\epsilon
|x)$, $\int_{0}^{\infty} \tilde{f}_{U}(\epsilon |x)\mathrm{d}\epsilon
<\infty$  for all $x$,
    \item[(iii)] 
    $\sup_{x,\psi }\left\Vert \frac{\partial^{2}g(x,\psi )}{\partial \psi \partial \psi ^{\prime }}\right\Vert$, $\sup_{\epsilon \le 0 ,x_{i},\psi
}\left\Vert \frac{\partial f_{L,\epsilon }(\epsilon |x,\psi )}{\partial \psi}\right\Vert$, 
$\sup_{\epsilon \ge 0 ,x_{i},\psi }\left\Vert \frac{\partial f_{U,\epsilon }(\epsilon |x,\psi )}{\partial \psi }\right\Vert$, 
$\sup_{\epsilon < 0, x,\psi }\left\Vert \frac{\partial f_{L,\epsilon}(\epsilon |x,\psi )}{\partial \epsilon }\right\Vert$, \\
$\sup_{\epsilon > 0, 
,x,\psi }\left\Vert \frac{\partial f_{U,\epsilon }(\epsilon |x,\psi )}{\partial \epsilon }\right\Vert$,  
$\sup_{x,\psi }\int_{-\infty }^{0 }\left\Vert  \frac{\partial ^{2}f_{L,\epsilon }(\epsilon |x,\psi )}{\partial \psi \partial \psi^{\prime }}\right\Vert \mathrm{d}\epsilon$, 
$\sup_{x,\psi} 
\int_{0}^{\infty} \left\Vert \frac{\partial ^{2}f_{U,\epsilon }(\epsilon |x,\psi )}{\partial \psi \partial \psi^{\prime }}\right\Vert \mathrm{d}\epsilon$ are bounded above by finite constants. 
\end{enumerate}
Then, $\widehat{Q}_{T}(\cdot ;u,\psi
) $ is norm-differentiable at $\psi =\psi _{0}$ with $D(s;u,\psi )=\mathbb{E}%
\left[ \frac{\partial F(u_{1}^{-1}(s-u_{2}^{\top }X_{t})|X_{t},\psi )}{%
\partial \psi }\right] $.
\end{lemma}

\begin{proof}
    The proof is similar to that of Lemma \ref{lemma:norm-diff-one-sided-general}.
\end{proof}

\section{Proofs of the Results in Section \ref{section:asym-theory}}

For simplicity of notation, we denote
\begin{align*}
	\mathcal{M}_{T}(\psi ) & = \left( \int_{\mathbb{S}^{d-1}}
	\int_{\mathcal{S}}(Q_{T}(s;u)-\widehat{Q}_{T}(s;u,\psi))^{2}w(s)\mathrm{d}
	sd\varsigma (u)\right) ^{1/2}, \\
	\mathcal{M}(\psi ) & =\left(\int_{\mathbb{S}^{d-1}}
	\int_{\mathcal{S}}(Q(s;u)-Q(s;u,\psi))^{2}w(s)\mathrm{d}sd\varsigma (u)\right)
	^{1/2}, \\
	\widehat{\mathcal{SW}}_{T} & = \left(\int_{\mathbb{S}^{d-1}} \int_{\mathcal{S}}(Q_{T}(s; u)-Q(s;
	u))^{2}w(s)\mathrm{d}s \mathrm{d}\varsigma(u) \right) ^{1/2}, \\
	\widehat{\mathcal{SW}}_{T} (\psi) & = \left(\int_{\mathbb{S}^{d-1}} \int_{\mathcal{S}}(\widehat{Q}_{T}(s,u, \psi
	)-Q(s,u, \psi))^{2}w(s)\mathrm{d}s\mathrm{d}\varsigma(u)\right) ^{1/2}, \\
	\overline{B}_T & = \int_{\mathbb{S}
	^{d-1}}\int_{\mathcal{S}} \lVert \widehat{D}_{T}(s;u, \psi _{0})\rVert^2 w(s) 
	\mathrm{d}s \mathrm{d}\varsigma(u), \\
	\overline{B}_0 & = \int_{\mathbb{S}
	^{d-1}}\int_{\mathcal{S}} \lVert D(s;u, \psi _{0})\rVert^2 w(s) 
	\mathrm{d}s \mathrm{d}\varsigma(u).
\end{align*}

\subsection{Proof of Theorem \protect\ref{thm:consistency-general}
(consistency)}

We can show the consistency by following standard approaches in the extremum
estimator.
Since $\hat{\psi}_{T}$ satisfies $(\mathcal{M}_{T}(\hat{\psi}_{T}))^{2}\leq
\inf_{\psi \in \Psi }(\mathcal{M}_{T}(\psi ))^{2}+T^{-1}\epsilon _{T}$, we have
\begin{align*}
	\mathcal{M}_{T}(\hat{\psi}_{T}) 
	\leq \inf_{\psi \in \Psi }\mathcal{M}_{T}(\psi )+T^{-1/2}\epsilon _{T}^{1/2} =\inf_{\psi \in \Psi }\mathcal{M}_{T}(\psi )+o_{p}(1) \le \mathcal{M}(\psi_0 )+o_{p}(1).
\end{align*}

By the triangle inequality of sliced Wasserstein
distance, and Minkowski inequality (c.f., \cite{Nadjahi2020SDProperty}), we can show that $\sup_{\psi \in \Psi }|\mathcal{M}_{T}(\psi )-\mathcal{M}(\psi )|$ under Assumption \ref{assumption:convergence} as follows:
\begin{align*}
& \sup_{\psi \in \Psi }|\mathcal{M}_{T}(\psi )-\mathcal{M}(\psi )| \\
& =\sup_{\psi \in \Psi }\left\vert \left( \int_{\mathbb{S}
^{d-1}}\int_{\mathcal{S}}(Q_{T}(s;u)-\widehat{Q}_{T}(s;u,\psi))^{2}w(s)\mathrm{d}
sd\varsigma (u)\right) ^{1/2} - \left(\int_{\mathbb{S}^{d-1}}
\int_{\mathcal{S}}(Q(s;u)-Q(s;u,\psi))^{2}w(s)\mathrm{d}sd\varsigma (u)\right)
^{1/2} \right\vert \\
& \leq \sup_{\psi \in \Psi} \left(\int_{\mathbb{S}^{d-1}}
\int_{S}(Q_{T}(s; u)- Q(s; u) - \widehat{Q}_{T}(s; u, \psi)+Q(s; u, \psi
))^{2}w(s)\mathrm{d}s \mathrm{d} \varsigma(u) \right)^{1/2} \\
& \leq \widehat{SW}_{T} + \sup_{\psi
\in \Psi } \widehat{SW}_{T} (\psi) \xrightarrow{p}0.
\end{align*}
Assumption \ref{assumption:identification} implies that for
each $\epsilon >0$, there exists $\delta >0$ such that for any $\psi \notin
B(\psi _{0},\epsilon )\Rightarrow \mathcal{M}(\psi )-\mathcal{M}(\psi
_{0})\geq \delta >0.$

Using these results, we have 
\begin{align*}
\Pr (\hat{\psi}_{T}\notin B(\psi _{0},\epsilon ))& \leq \Pr (\mathcal{M}(%
\hat{\psi}_{T})-\mathcal{M}(\psi _{0})\geq \delta ) \\
& = \Pr(\mathcal{M}(\hat{\psi}_{T})-\mathcal{M}_{T}(\hat{\psi}_{T})+\mathcal{%
\ \ M}_{T}(\hat{\psi}_{T})-\mathcal{M}(\psi _{0})\geq \delta ) \\
& \leq \Pr (\mathcal{M}(\hat{\psi}_{T})-\mathcal{M}_{T}(\hat{\psi}_{T})+ 
\mathcal{M}_{T}(\psi _{0})+o_{p}(1)-\mathcal{M}(\psi _{0})\geq \delta ) \\
& \leq \Pr (2\sup_{\psi \in \Psi }|\mathcal{M}_{T}(\psi)-\mathcal{M}(\psi
)|+o_{p}(1)\geq \delta ) \\
& \rightarrow 0.
\end{align*}
Therefore, $\hat{\psi}_{T}\xrightarrow{p}\psi _{0}$.

\subsection{Proof of Theorem \protect\ref{thm:asym-normality-general}
(asymptotic distribution)}

Here, we will verify Assumptions 1 to 6 in \cite{Andrews1999}. When $\psi_0$
is in the interior of $\Psi$, Assumptions 5 and 6 in \cite{Andrews1999}
hold. Therefore, it is enough to prove Assumptions 1 to 4 in \cite{Andrews1999}.

First, Assumption
1 holds in \cite{Andrews1999} are satisfied under Assumptions \ref{assumption:convergence} and \ref%
{assumption:identification} by Theorem 4.1 in \cite{Andrews1999}.

Let us remind that 
\begin{equation*}
\widehat{R}_{T}(s;u,\psi ,\psi _{0}):=\widehat{Q}_{T}(s;u,\psi )- \widehat{Q}
_{T}(s;u,\psi _{0})-(\psi -\psi _{0})^{\top }\widehat{D}_{T}(s;u,\psi _{0})
\end{equation*}
with 
\begin{equation*}
\sup_{\psi \in \Psi ;\;\lVert \psi - \psi_{0} \rVert \leq \tau_{T}} \left| 
\frac{T \int_{\mathbb{S}^{d-1}} \int_{\mathcal{S}} \left( \widehat{R}_{T}(s;u,\psi,\psi
_{0})\right) ^{2}w(s)dsd\varsigma (u)}{(1 + \lVert \sqrt{T} (\psi-\psi
_{0})\rVert )^{2}} \right| = o_{p}(1)
\end{equation*}
for any $\tau_{T} \rightarrow 0$.

The objective function can be decomposed as follows:
\begin{align*}
	& \int_{\mathbb{S}^{d-1}}\int_{\mathcal{S}}(Q_{T}(s;u)- \widehat{Q}_{T}(s;u, \psi
	))^{2}w(s)\mathrm{d}s \mathrm{d}\varsigma (u) \\
	&\quad  = \int_{\mathbb{S}^{d-1}}\int_{\mathcal{S}} (Q_{T}(s;u)- \widehat{Q}
	_{T}(s;u,\psi_0))^{2}w(s)\mathrm{d}s \mathrm{d}\varsigma(u) - 2(\psi -
	\psi_0)^{\top} A_T /\sqrt{T} + (\psi - \psi_0)^{\top} B_T (\psi - \psi_0) +
	R_T,
\end{align*}
where 
\begin{align*}
	A_T & = \int_{\mathbb{S}^{d-1}}\int_{\mathcal{S}} (Q_{T}(s;u)- \widehat{Q}
	_{T}(s;u,\psi_0)) \widehat{D}_{T}(s;u,\psi _{0}) w(s)\mathrm{d}s \mathrm{d}
	\varsigma(u); \\
	B_T & = \int_{\mathbb{S}^{d-1}} \int_{\mathcal{S}} \widehat{D}_{T}(s;u,\psi _{0}) 
	\widehat{D}_{T}^{\top}(s ;u,\psi _{0}) w(s)\mathrm{d}sd\varsigma(u); \\
	R_T & = \int_{\mathbb{S}^{d-1}}\int_{\mathcal{S}} \widehat{R}_{T}^2(s;u,\psi ,\psi
	_{0}) w(s)\mathrm{d}s \mathrm{d}\varsigma(u) - 2 \int_{\mathbb{S}
		^{d-1}}\int_{\mathcal{S}} (Q_{T}(s;u)- \widehat{Q}_{T}(s;u,\psi_0)) \widehat{R}
	_{T}(s;u,\psi ,\psi _{0}) w(s)\mathrm{d}s \mathrm{d}\varsigma(u) \\
	& \quad + 2(\psi - \psi_0)^{\top} \int_{\mathbb{S}^{d-1}}\int_{\mathcal{S}} 
	\widehat{D}_{T}(s;u,\psi_0) \widehat{R}_{T}(s;u,\psi ,\psi _{0}) w(s)\mathrm{d}s \mathrm{d}\varsigma(u).
\end{align*}

Then, it is enough to show 
\begin{align*}
A_{T} \xrightarrow{d} N(0,V_{0}), \;
B_{T}  \xrightarrow{p} B_{0}, \text{ and } 
\sup_{\psi \in \Psi ;\;\lVert \psi - \psi_{0}\rVert \leq \tau_{T}}\frac{
T\lvert R_{T}\rvert }{(1+\lVert \sqrt{T}(\psi -\psi _{0})\rVert)^{2}}& %
\xrightarrow{p} 0
\end{align*}
for any $\tau _{T}\rightarrow 0$ as $T \rightarrow \infty $.

This is because the first two conditions satisfy Assumption 3 in \cite%
{Andrews1999}, and the third condition implies Assumption 2* in \cite%
{Andrews1999}. Assumptions 1, 2*, and 3 \cite{Andrews1999} imply Assumption 4 in \cite{Andrews1999} by Theorem 1 in \cite{Andrews1999}.

In Steps 1 to 3 below, we will analyze the behaviors of $A_T$, $B_T$, and $%
R_T$, respectively.

\textbf{Step 1: } We will analyze the behavior of $A_T$.

Note that $Q(\cdot, \cdot) = Q(\cdot, \cdot, \psi_0)$. 
\begin{align*}
	A_T & = \sqrt{T} \int_{\mathbb{S}^{d-1}}\int_{\mathcal{S}} (Q_{T}(s;u)- \widehat{Q}
	_{T}(s;u,\psi_0)) \widehat{D}_{T}(s;u,\psi _{0}) w(s)\mathrm{d}s \mathrm{d}
	\varsigma(u) \\
	& = \sqrt{T} \int_{\mathbb{S}^{d-1}}\int_{\mathcal{S}} (Q_{T}(s;u)-Q(s; u)) 
	\widehat{D}_{T}(s;u,\psi _{0}) w(s)\mathrm{d}s \mathrm{d}\varsigma(u) \\
	& \quad + \sqrt{T} \int_{\mathbb{S}^{d-1}}\int_{\mathcal{S}} (Q(s; u, \psi_0) - 
	\widehat{Q}_{T}(s;u,\psi_0)) \widehat{D}_{T}(s;u,\psi _{0}) w(s)\mathrm{d}s 
	\mathrm{d}\varsigma(u) 
\end{align*}
Under Assumption \ref{assumption:big-oh}, we have 
\begin{align*}
\sqrt{T} \int_{\mathbb{S}^{d-1}}\int_{\mathcal{S}} (Q_{T}(s;u)-Q(s; u)) (\widehat{D}
_{T}(s;u,\psi _{0}) - D(s;u,\psi _{0})) w(s)\mathrm{d}s \mathrm{d}
\varsigma(u) & = o_p(1) \\
\sqrt{T} \int_{\mathbb{S}^{d-1}}\int_{\mathcal{S}} (Q(s; u, \psi_0) - \widehat{Q}
_{T}(s;u,\psi_0)) (\widehat{D}_{T}(s;u,\psi _{0}) - D(s;u,\psi _{0})) w(s) 
\mathrm{d}s \mathrm{d}\varsigma(u) & = o_p(1).
\end{align*}
Therefore, we have 
\begin{align*}
A_T & = \sqrt{T} \int_{\mathbb{S}^{d-1}}\int_{\mathcal{S}} (Q_{T}(s;u)-Q(s; u))
D(s;u,\psi _{0}) w(s)\mathrm{d}s \mathrm{d}\varsigma(u) \\
& \quad + \sqrt{T} \int_{\mathbb{S}^{d-1}}\int_{\mathcal{S}} (Q(s; u, \psi_0) - 
\widehat{Q}_{T}(s;u,\psi_0)) D(s;u,\psi _{0}) w(s)\mathrm{d}s \mathrm{d}
\varsigma(u) + o_p(1),
\end{align*}
which is asymptotically normal under Assumption \ref{assumption:CLT}.

\textbf{Step 2: } We will analyze the behavior of $B_T$.

Note that 
\begin{align*}
B_T & = \int_{\mathbb{S}^{d-1}}\int_{\mathcal{S}} \widehat{D}_{T}(s;u,\psi _{0}) 
\widehat{D}_{T}^{\top}(s ;u,\psi _{0}) w(s)\mathrm{d}s \mathrm{d}\varsigma(u)
\end{align*}
Let us denote 
\begin{align*}
B_0 & = \int_{\mathbb{S}^{d-1}}\int_{\mathcal{S}} D(s;u,\psi _{0})D^{\top}(s
;u,\psi _{0}) w(s)\mathrm{d}s \mathrm{d}\varsigma(u).
\end{align*}
Then, we have 
\begin{align*}
B_T - B_0 & = \int_{\mathbb{S}^{d-1}}\int_{\mathcal{S}} \widehat{D}_{T}(s;u,\psi
_{0})\widehat{D}_{T}^{\top}(s ;u,\psi _{0}) w(s)\mathrm{d}sd\varsigma(u) -
\int_{\mathbb{S}^{d-1}}\int_{\mathcal{S}} D(s;u,\psi _{0}) D^{\top}(s ;u,\psi
_{0}) w(s)\mathrm{d}sd\varsigma(u) \\
& = \int_{\mathbb{S}^{d-1}}\int_{\mathcal{S}} (\widehat{D}_{T}(s;u,\psi _{0}) -
D(s;u,\psi _{0})) \widehat{D}_{T}^{\top}(s ;u,\psi _{0}) w(s)\mathrm{d}
sd\varsigma(u) \\
& \quad + \int_{\mathbb{S}^{d-1}}\int_{\mathcal{S}} D(s ;u,\psi _{0}) (\widehat{D}
_{T}(s;u,\psi _{0}) - D(s;u,\psi _{0}))^{\top} w(s)\mathrm{d}s \mathrm{d}%
\varsigma(u),
\end{align*}
which is $o_p(1)$ when Assumption \ref{assumption:big-oh} (iii) holds.

\textbf{Step 3: } We will analyze the behavior of $R_T$.

Here, we will show 
\begin{equation*}
\sup_{\psi \in \Psi ;\;\lVert \psi -\psi _{0}\rVert \leq \tau _{T}}\frac{
T\lvert R_{T}\rvert }{(1+\lVert \sqrt{T}(\psi -\psi _{0})\rVert)^{2}}
=o_{p}(1).
\end{equation*}

Let us remind that 
\begin{align*}
R_T 
& = 
\int_{\mathbb{S}^{d-1}}\int_{\mathcal{S}} \left(\widehat{R}_{T}(s;u,\psi
,\psi _{0})\right)^2 w(s)\mathrm{d}s \mathrm{d}\varsigma(u) \\
& \quad 
- 
2 \int_{\mathbb{S}^{d-1}} \int_{\mathcal{S}} (Q_{T}(s;u)- \widehat{Q}_{T}(s;u,\psi_0)) 
\widehat{R}_{T}(s;u,\psi ,\psi _{0}) w(s)\mathrm{d}s \mathrm{d}\varsigma(u)
\\
& \quad + 
2(\psi - \psi_0)^{\top} \int_{\mathbb{S}^{d-1}} \int_{\mathcal{S}} 
\widehat{D}_{T}(s;u,\psi_0) \widehat{R}_{T}(s;u,\psi ,\psi _{0}) 
w(s) \mathrm{d}s \mathrm{d}\varsigma(u),
\end{align*}
where 
\begin{equation*}
\sup_{\psi \in \Psi ;\;\lVert \psi - \psi_{0} \rVert \leq \tau_{T}} \left| 
\frac{T \int_{\mathbb{S}^{d-1}}\int \left( \widehat{R}_{T}(s;u,\psi,\psi
_{0})\right) ^{2}w(s)dsd\varsigma (u)}{(1 + \lVert \sqrt{T}(\psi-\psi _{0})
\rVert )^{2}} \right| = o_{p}(1)
\end{equation*}
for any $\tau_{T} \rightarrow 0$.

Note that 
\begin{align*}
	|R_T| 
	& \le 
	\int_{\mathbb{S}^{d-1}} \int_{\mathcal{S}} \left(\widehat{R}_{T} (s;u,\psi ,\psi _{0})\right)^2 w(s)\mathrm{d}s \mathrm{d}\varsigma(u) 
	+ 
	2 \mathcal{M}_T(\psi_0)
	\left( \int_{\mathbb{S}^{d-1}}\int_{\mathcal{S}} \left(\widehat{R}_{T}(s;u,\psi
	,\psi _{0})\right)^2 w(s)\mathrm{d}s \mathrm{d}\varsigma(u) \right)^{1/2} \\
	& \quad + 2 \lVert \psi - \psi_0 \rVert_{\Psi} 
	\left( \overline{B}_T \right)^{1/2} 
	\left( \int_{\mathbb{S}^{d-1}} \int_{\mathcal{S}} \left(\widehat{R}_{T}(s;u,\psi ,\psi _{0})\right)^2 w(s) 
	\mathrm{d}s \mathrm{d}\varsigma(u) \right)^{1/2}.
	\end{align*}
This implies that 
\begin{align*}
& \sup_{\psi \in \Psi ;\;\lVert \psi -\psi _{0}\rVert \leq \tau _{T}}\frac{
T\lvert R_{T}\rvert }{(1+\lVert \sqrt{T}(\psi -\psi _{0})\rVert)^{2}} \\
& \qquad \le \sup_{\psi \in \Psi ;\;\lVert \psi -\psi _{0}\rVert \leq \tau_{T}} 
\frac{T \int_{\mathbb{S}^{d-1}}\int_{\mathcal{S}} \left(\widehat{R}_{T}(s;u,\psi
,\psi _{0})\right)^2 w(s)\mathrm{d}s \mathrm{d}\varsigma(u) }{(1+\lVert 
\sqrt{T}(\psi -\psi_{0})\rVert )^{2}} \\
& \qquad \quad + 
2 \sup_{\psi \in \Psi ;\;\lVert \psi -\psi _{0}\rVert \leq
\tau_{T}} \frac{T \mathcal{M}_T(\psi_0) \left( \int_{\mathbb{S}^{d-1}}\int_{\mathcal{S}} \left(\widehat{R}
_{T}(s;u,\psi ,\psi _{0})\right)^2 w(s)\mathrm{d}s \mathrm{d}\varsigma(u)
\right)^{1/2} }{(1+\lVert \sqrt{T}(\psi -\psi _{0 })\rVert)^{2}} \\
& \qquad \quad + 2 \sup_{\psi \in \Psi ;\;\lVert \psi -\psi _{0}\rVert \leq
\tau_{T}} \frac{T \lVert \psi - \psi_0 \rVert \left( \overline{B}_T \right)^{1/2} \left( \int_{\mathbb{S}
^{d-1}}\int_{\mathcal{S}} \left(\widehat{R}_{T}(s;u,\psi ,\psi _{0})\right)^2 w(s) 
\mathrm{d}s \mathrm{d}\varsigma(u) \right)^{1/2}}{(1+\lVert \sqrt{n}(\psi
-\psi _{0 })\rVert)^{2}}.
\end{align*}

First, we have 
\begin{align*}
\sup_{\psi \in \Psi ;\;\lVert \psi -\psi _{\ast }\rVert \leq \tau_{T}}\frac{
T \int_{\mathbb{S}^{d-1}}\int_{\mathcal{S}} \left(\widehat{R} _{T}(s;u,\psi ,\psi
_{0})\right)^2 w(s)\mathrm{d}s \mathrm{d}\varsigma(u) }{(1 + \lVert \sqrt{T}
(\psi -\psi _{\ast})\rVert)^{2}} & = o_p(1)
\end{align*}
by Assumption \ref{assumption:norm-diff}.

Second, we have 
\begin{align*}
& \sup_{\psi \in \Psi ;\;\lVert \psi -\psi_{0} \rVert \leq \tau_{T}} 
\frac{T \mathcal{M}_T(\psi_0)
\left( \int_{\mathbb{S}^{d-1}}\int_{\mathcal{S}}\left(\widehat{R}_{T}(s;u,\psi
,\psi _{0})\right)^2 w(s)\mathrm{d}s \mathrm{d}\varsigma(u) \right)^{1/2} }{
(1+\lVert \sqrt{T}(\psi - \psi_{0})\rVert)^{2}} \\
& \qquad \le \sqrt{T} \mathcal{M}_T(\psi_0) \sup_{\psi \in \Psi ;\;\lVert \psi -\psi_{0} \rVert \leq
\tau_{T}} \frac{\sqrt{T} \left( \int_{\mathbb{S}^{d-1}}\int_{\mathcal{S}}\left( 
\widehat{R}_{T}(s;u,\psi ,\psi _{0})\right)^2 w(s)\mathrm{d}s \mathrm{d}
\varsigma(u) \right)^{1/2} }{(1+\lVert \sqrt{T}(\psi - \psi_{0})\rVert)^{2}}
\\
& \qquad \le \sqrt{T} \mathcal{M}_T(\psi_0) \sup_{\psi \in \Psi ;\;\lVert \psi -\psi_{0} \rVert \leq
\tau_{T}} \frac{\sqrt{T} \left( \int_{\mathbb{S}^{d-1}}\int_{\mathcal{S}}\left( 
\widehat{R}_{T}(s;u,\psi ,\psi _{0})\right)^2 w(s)\mathrm{d}s \mathrm{d}
\varsigma(u) \right)^{1/2} }{1+\lVert \sqrt{T}(\psi - \psi_{0})\rVert} \\
& \qquad = \sqrt{T} \mathcal{M}_T(\psi_0)
\left(\sup_{\psi \in \Psi ;\;\lVert \psi -\psi_{0} \rVert \leq \tau_{T}} 
\frac{ T \int_{\mathbb{S}^{d-1}}\int_{\mathcal{S}}\left( \widehat{R}_{T}(s;u,\psi
,\psi _{0})\right)^2 w(s)\mathrm{d}s \mathrm{d} \varsigma(u) }{(1 + \lVert 
\sqrt{T} (\psi - \psi_{0})\rVert)^2}\right)^{1/2} \\
& \qquad = O_p(1) o_p(1) = o_p(1).
\end{align*}
Here, the second last equality holds because Assumptions \ref%
{assumption:big-oh} implies 
\begin{align*}
\sqrt{T} \mathcal{M}_T(\psi_0) 
\le \sqrt{T} \widehat{\mathcal{SW}}_T + \sqrt{T} \widehat{\mathcal{SW}}_T(\psi_0) 
= O_p(1),
\end{align*}
and Assumption \ref{assumption:norm-diff} implies 
\begin{align*}
\sup_{\psi \in \Psi ;\;\lVert \psi -\psi_{0 }\rVert \leq \tau_{T}} \frac{ T
\int_{\mathbb{S}^{d-1}}\int_{\mathcal{S}}\left(\widehat{R}_{T}(s;u,\psi ,\psi
_{0})\right)^2 w(s)\mathrm{d}s \mathrm{d}\varsigma(u) }{(1 + \lVert \sqrt{T}
(\psi -\psi _{0})\rVert)^{2}} = o_p(1).
\end{align*}

Third, we have 
\begin{align*}
& \sup_{\psi \in \Psi ;\;\lVert \psi - \psi_{0} \rVert \leq \tau_{T}} 
\frac{ T \lVert \psi - \psi_0 \rVert 
\left( \overline{B}_T \right)^{1/2} 
\left( \int_{\mathbb{S}^{d-1}}\int_{\mathcal{S}}\left( 
\widehat{R}_{T}(s;u,\psi ,\psi _{0})\right)^2 w(s)\mathrm{d}s \mathrm{d}
\varsigma(u) \right)^{1/2}}{(1+\lVert \sqrt{T}(\psi -\psi _{0 })\rVert)^{2}}
\\
& \qquad = \left( \overline{B}_T \right)^{1/2}
\sup_{\psi \in \Psi ;\;\lVert \psi - \psi_{0} \rVert \leq \tau_{T}} \frac{T
\lVert \psi - \psi_0 \rVert \left( \int_{\mathbb{S}^{d-1}}\int_{\mathcal{S}}\left( 
\widehat{R}_{T}(s;u,\psi ,\psi _{0})\right)^2 w(s)\mathrm{d}s \mathrm{d}
\varsigma(u) \right)^{1/2}}{(1+\lVert \sqrt{T}(\psi -\psi _{0 })\rVert)^{2}}
\\
& \qquad = \left( \overline{B}_T \right)^{1/2}
\sup_{\psi \in \Psi ;\;\lVert \psi - \psi_{0} \rVert \leq \tau_{T}} \frac{ 
\sqrt{T} \lVert \psi - \psi_0 \rVert \left( T \int_{\mathbb{S}
^{d-1}}\int_{\mathcal{S}}\left( \widehat{R}_{n}(s;u,\psi ,\psi _{0})\right)^2 w(s) 
\mathrm{d}s \mathrm{d} \varsigma(u) \right)^{1/2}}{(1+\lVert \sqrt{T}(\psi
-\psi _{0 })\rVert)(1+\lVert \sqrt{n}(\psi -\psi _{0 })\rVert)} \\
& \qquad = O_p(1) o_p(1) = o_p(1).
\end{align*}

Therefore, we have 
\begin{equation*}
\sup_{\psi \in \Psi ;\;\lVert \psi - \psi_{0}\rVert \leq \tau _{T}}\frac{
T\lvert R_{T}\rvert }{(1+\lVert \sqrt{T}(\psi -\psi _{0})\rVert)^{2}}
=o_{p}(1).
\end{equation*}

Therefore, Assumptions 1 to 6 in \cite{Andrews1999} hold, and we can apply
the Theorem 3 in \cite{Andrews1999}. This implies 
\begin{align*}
\sqrt{T}(\hat{\psi}_T - \psi_0) 
 = B_T^{-1} A_T + o_p(1) 
\xrightarrow{d} B_0^{-1} N(0, \Omega_0),
\end{align*}
where $\Omega _{0}=(e_{1}^{\top }, - e_{1}^{\top})V_{0}%
\begin{pmatrix}
e_{1} \\ 
- e_{1}%
\end{pmatrix}%
$ in which $e_{1}=(1,\dotsc ,1)^{\top }$ is a $d_{\psi}$
by $1$ vector of ones.

\section{Proofs of the Results in Section \ref{section:SC-conditional}} \label{appendix:proofs-SC}

\subsection{Preliminary Lemma}

\begin{lemma}
\label{lemma:Lipchitz-k}	
	
Suppose $F(y|x, \psi)$ is Lipschitz with respect to $\psi$ in the sense that
there exists a function $M(y, x)$ for any $\psi, \psi^{\prime}\in \Psi$, 
\begin{align*}
|F(y|x, \psi) - F(y|x, \psi^{\prime})| \le M(y, x) \|\psi - \psi^{\prime}\|,
\end{align*}
and $\int_{u \in \mathbb{S}^{d-1}, u_1 \ne 0} \int_{-\infty}^{\infty} \int
M^2(u_1^{-1}(s - u_2^{\top} x); x) \mathrm{d}F_X(x) w(s) \mathrm{d}s \mathrm{%
d }\varsigma(u) < \infty$, where $F_X(\cdot)$ is the CDF of $X_t$. Then, $%
k(x, x^{\prime}, \psi)$ and $k_2(x, x^{\prime}, \psi, \psi_0) $ are
Lipschitz in $\psi$.
\end{lemma}

\begin{proof}
Note that 
\begin{align*}
\mathbb{E}[I(u^{\top} Z_t \le s)|X_t, \psi] = 
\begin{cases}
F(u_1^{-1}(s - u_2^{\top} X_t) | X_t, \psi) & \text{if } u_1 > 0 \\ 
I(u_2^{\top} X_t \le s) & \text{if } u_2 < 0 \\ 
1 - F(u_1^{-1}(s - u_2^{\top} X_t) | X_t, \psi) & \text{if } u_1 > 0%
\end{cases}%
\end{align*}
Then, 
\begin{align*}
& |k(X_t, X_j, \psi) - k(X_t, X_j, \psi^{\prime})| \\
&\qquad \le 2 \int_{u \in \mathbb{S}^{d-1}, u_1 \ne 0} \int_{-\infty}^{\infty} %
\Bigg| \left(F(u_1^{-1}(s - u_2^{\top} X_t) | X_t, \psi) - \mathbb{E}
[F(u_1^{-1}(s - u_2^{\top} X_t) | X_t, \psi)]\right) \\
& \qquad \quad \quad \quad \quad \quad \quad \quad \quad \quad - \left(F(u_1^{-1}(s
- u_2^{\top} X_t) | X_t, \psi') - \mathbb{E}[F(u_1^{-1}(s - u_2^{\top} X_t) |
X_t, \psi')]\right) \Bigg| w(s) \mathrm{d}s \mathrm{d}\varsigma(u) \\
&\qquad  \quad + 2 \int_{u \in \mathbb{S}^{d-1}, u_1 \ne 0} \int_{-\infty}^{\infty} %
\Bigg| \left(F(u_1^{-1}(s - u_2^{\top} X_j) | X_j, \psi) - \mathbb{E}
[F(u_1^{-1}(s - u_2^{\top} X_j) | X_j, \psi)]\right) \\
&\qquad \quad \quad \quad \quad \quad \quad \quad \quad \quad - \left(F(u_1^{-1}(s
- u_2^{\top} X_j) | X_j, \psi') - \mathbb{E}[F(u_1^{-1}(s - u_2^{\top} X_j) |
X_j, \psi')]\right) \Bigg| w(s) \mathrm{d}s \mathrm{d}\varsigma(u).
\end{align*}
Since we have 
\begin{align*}
|F(u_1^{-1}(s - u_2^{\top} x) |x, \psi) - F(u_1^{-1}(s - u_2^{\top} x) |x,
\psi^{\prime})| \le M(u_1^{-1}(s - u_2^{\top} x); x) \| \psi -
\psi^{\prime}\|,
\end{align*}
where $\int_{u \in \mathbb{S}^{d-1}, u_1 \ne 0} \int_{-\infty}^{\infty} \int
M^2(u_1^{-1}(s - u_2^{\top} x); x) \mathrm{d}F_X(x) w(s) \mathrm{d }s 
\mathrm{d}\varsigma(u) < \infty$, $k(x, x^{\prime}, \psi)$ is Lipchitz
continuous with respect to $\psi$. With a similar calculation, we can show
that $k_2(x, x^{\prime}, \psi, \psi_0)$ is Lipchitz continuous with respect
to $\psi$ when $F(y|x, \psi)$ is Lipchitz continuous as well.
\end{proof}

\subsection{Proof of Lemma \ref{lemma:Lipchitz-one-sided}}

     Because $f_{\epsilon}(u|x, \psi) = 0$ when $u < 0$, we can represent $F(Y; \psi)$ as follows.
     \begin{align*}
         F(y|x, \psi) = \int_{-\infty}^{y-g(x, \psi)} f_{\epsilon}(u |x, \psi) d u.
     \end{align*}
     Then, 
     \begin{align*}
         |F(y|x, \psi') - F(y|x, \psi)| 
         & = 
         \left|\int_{-\infty}^{y-g(x, \psi')} f_{\epsilon}(u |x, \psi') d u - \int_{-\infty}^{y-g(x, \psi)} f_{\epsilon}(u |x, \psi) d u \right| \\
         & \le
         \left|\int_{y-g(x, \psi)}^{y-g(x, \psi')} f_{\epsilon}(u |x, \psi') d u \right| + 
        \left|\int_{-\infty}^{y-g(x, \psi)} (f_{\epsilon}(u |x, \psi') - f_{\epsilon}(u |x, \psi)) d u \right| \\
        & \le \sup_{\epsilon, x, \psi} | f_{\epsilon}(\epsilon|x, \psi)| |g(x, \psi') - g(x, \psi)|
        +
        \left|\int_{-\infty}^{y-g(x, \psi)} (\psi' - \psi)^{\top} \frac{\partial}{\partial \psi} f_{\epsilon}(u | x, \tilde{\psi}) d u \right| \\
        & \le
        \sup_{\epsilon, x, \psi} | f_{\epsilon}(\epsilon|x, \psi)| \sup_{x, \psi} \left\| \frac{\partial}{\partial \psi} g(x, \psi)  \right\| \| \psi' - \psi\|  \\
        & \quad +
       \left( \sup_{\psi, x} \int_0^{\infty} \left\| \frac{\partial}{\partial \psi} f_{\epsilon}(u | x, \psi)\right\| d\epsilon \right)  \| \psi' - \psi \| \\
       & \le  C \| \psi' - \psi \|,
     \end{align*}
     where $\tilde{\psi}$ is between $\psi$ and $\psi'$, and $C$ is some absolute positive constant.

\subsection{Proof of Lemma \ref{lemma:sup-bound-hessian}}

For simplicity of notation, let 
\begin{align*}
G(u,s,\psi )& =\mathbb{E}[F(u_{1}^{-1}(s-u_{2}^{\top }X_{t})|X_{t},\psi )],
\\
G_{t}(u,s,\psi )& =F(u_{1}^{-1}(s-u_{2}^{\top }X_{t})|X_{t},\psi ), \\
D_{t}(u,s,\psi )& =\frac{\partial }{\partial \psi }F(u_{1}^{-1}(s-u_{2}^{%
\top }X_{t})|X_{t},\psi ), \\
D(u,s,\psi )& =\mathbb{E}\left[ \frac{\partial }{\partial \psi }%
F(u_{1}^{-1}(s-u_{2}^{\top }X_{t})|X_{t},\psi )\right] , \\
R(s,u,\psi )& =G(u,s,\psi )-G(u,s,\psi _{0})-(\psi -\psi _{0})^{\top
}D(u,s,\psi ), \\
R_{t}(s,u,\psi )& =G_{t}(u,s,\psi )-G_{t}(u,s,\psi _{0})-(\psi -\psi
_{0})^{\top }D_{t}(u,s,\psi ).
\end{align*}%
Here, we set $D_{t}(s,u,\psi )=0$ at kink point $s$. Note that $R(s,u,\psi
)=\mathbb{E}[R_{t}(s,u,\psi )]$.

Our main purpose is to show that for any $\tau_T \rightarrow 0$, 
\begin{align*}
\sup_{\|\psi - \psi_0\| \le \tau_T} \frac{T \int_{u \in \mathbb{S}^{d-1}}
\int_{-\infty}^{\infty} |R(s ,u, \psi)|^2 w(s) \mathrm{d}s \mathrm{d}%
\varsigma(u) }{(1 + \|\sqrt{T}(\psi - \psi_0)\|^2) } = o(1).
\end{align*}
It is enough to show
\begin{align}
\sup_{\|\psi - \psi_0\| \le \tau_T} \frac{T \mathbb{E}\left[\int_{u \in 
\mathbb{S}^{d-1}} \int_{-\infty}^{\infty} |R_t(s ,u, \psi)|^2 w(s) \mathrm{d}%
s \mathrm{d}\varsigma(u) \right] }{(1 + \|\sqrt{T}(\psi - \psi_0)\|^2) } =
o(1) \label{eq:cond-norm-diff-ver2}
\end{align}
for any $\tau_T \rightarrow 0$. This is because $R(s,u,\psi
)=\mathbb{E}[R_{t}(s,u,\psi )]$ implies that 
\begin{align*}
     \int_{u \in \mathbb{S}^{d-1}}
\int_{-\infty}^{\infty} |R(s ,u, \psi)|^2 w(s) \mathrm{d}s \mathrm{d}%
\varsigma(u) \le \mathbb{E}\left[\int_{u \in 
\mathbb{S}^{d-1}} \int_{-\infty}^{\infty} |R_t(s ,u, \psi)|^2 w(s) \mathrm{d}%
s \mathrm{d}\varsigma(u) \right].
\end{align*}
Therefore, we will analyze equation (\ref{eq:cond-norm-diff-ver2}).  
Assume that all assumptions in Lemmas \ref{lemma:Lipchitz-one-sided} or
\ref{lemma:Lipchitz-two-sided} are satisfied. Then, $g(x,\psi )$ is Lipshitz continuous at $\psi _{0}$
uniformly in $x:$ 
\begin{equation*}
|g(x,\psi )-g(x,\psi ^{\prime })|\leq C_{g}\Vert \psi -\psi ^{\prime }\Vert
\end{equation*}%
and, $|D_{t}(s,u,\psi )|$ is uniformly bounded
by an absolute constant, and for any $\psi ,\psi ^{\prime }\in \Psi $, 
\begin{equation*}
|F(y|x,\psi )-F(y|x,\psi ^{\prime })|\leq C\Vert \psi -\psi ^{\prime }\Vert .
\end{equation*}

When $u^{-1}(s-u_{2}^{\top }X_{t})\in A_{T}(X_t):=[g(X_t,\psi )\wedge g(X_t,\psi
_{0})-C_{g}\tau _{T},g(X_t,\psi )\vee g(X_t,\psi _{0})+C_{g}\tau _{T}]$, 
\begin{equation*}
|R_{t}(s,u,\psi )|\leq C\Vert \psi -\psi ^{\prime }\Vert .
\end{equation*}%
for some absolute constant $C$. This implies that when $\Vert \psi -\psi
_{0}\Vert \leq \tau _{T}$, 
\begin{align*}
& \int_{u\in \mathbb{S}^{d-1}}\int_{u^{-1}(s-u_{2}^{\top }X_{t})\in
A_{T}(X_t)}|R_{t}(s,u,\psi )|^{2}w(s)\mathrm{d}s\mathrm{d}\varsigma (u) \\
& \leq
C\Vert \psi -\psi _{0}\Vert ^{2}|u_{1}|\Bigg[(|g(X_{i},\psi )-g(X_{i},\psi
_{0})|+2C\tau _{T}\bigg] =O(\tau _{T}^{3}).
\end{align*}

When $\Vert \psi -\psi _{0}\Vert \leq \tau _{T}$ and $u^{-1}(s-u_{2}^{\top
}X_{t})\notin A_{T}(X_t)$, by Taylor theorem, we have 
\begin{align*}
& F(u^{-1}(s-u_{2}^{\top }X_{t})|x,\psi )-F(u^{-1}(s-u_{2}^{\top
}X_{t})|x,\psi _{0})-(\psi -\psi )^{\top }\left( \frac{\partial
F(u^{-1}(s-u_{2}^{\top }X_{t})|x,\psi _{0})}{\partial \psi }\right) \\
& \qquad =(\psi -\psi _{0})^{\top }\frac{\partial ^{2}F(u^{-1}(s-u_{2}^{\top
}X_{t})|X_{t},\tilde{\psi})}{\partial \psi \partial \psi ^{\prime }}(\psi
-\psi _{0})
\end{align*}%
for some $\tilde{\psi}$. When the weight function is integrable, 
\begin{equation*}
\int_{u\in \mathbb{S}^{d-1}}\int_{u^{-1}(s-u_{2}^{\top }X_{t})\notin
A_{T}(X_t)}|R_{t}(s,u,\psi )|^{2}w(s)\mathrm{d}s\mathrm{d}\varsigma (u)\leq C
\sup_{y \ne g(x, \psi),x,\psi }\left\Vert \frac{\partial ^{2}F(y|x,\psi)}{\partial \psi \partial \psi ^{\prime }}\right\Vert
^{2}\Vert \psi -\psi _{0}\Vert ^{4}.
\end{equation*}%
for some positive constant $C$.
Therefore, it is enough to get
\begin{equation*}
\sup_{y \ne g(x, \psi),x,\psi }\left\Vert \frac{\partial ^{2}F(y|x, \psi)}{\partial
\psi \partial \psi ^{\prime }}\right\Vert <\infty .
\end{equation*} 
to achieve equation (\ref{eq:cond-norm-diff-ver2}).

\begin{remark}
    For two-sided model, we replace Lemma \ref{lemma:Lipchitz-one-sided} by  Lemma \ref{lemma:Lipchitz-two-sided}. In fact, Lemma \ref{lemma:Lipchitz-one-sided} or Lemma \ref{lemma:Lipchitz-two-sided} imply uniform Lipchitz continuity of $F$ and $g$ with respect to $\psi$.
\end{remark}

\subsection{Proof of Lemma \ref{lemma:norm-diff-one-sided-general}}

First, conditions in Lemma \ref{lemma:norm-diff-one-sided-general} implies that $F(y|x, \psi)$ is continuously seond-order differentiable with respect to $\psi$, and In the one-sided model, when $y > g(x, \psi)$, 
\begin{align*}
\frac{\partial^2 F(y|x, \psi)}{\partial \psi \partial \psi^{\prime}} & = - 
\frac{\partial^2 g(x, \psi)}{\partial \psi \partial \psi^{\prime}}
f_{\epsilon}(y - g(x, \psi)|x, \psi) \\
& \quad + \frac{\partial g(x, \psi)}{\partial \psi} \left[\frac{\partial
g(x, \psi)}{\partial \psi^{\prime}} \frac{\partial f_{\epsilon}(y - g(x, \psi)|x,
\psi)}{\partial \epsilon} + \frac{\partial f_{\epsilon}(y - g(x, \psi)|x, \psi) }{\partial
\psi^{\prime}} \right] \\
& \quad - \frac{\partial g(x, \psi)}{\partial \psi} \frac{f_{\epsilon}(y -
g(x, \psi)|x, \psi) }{\partial \psi^{\prime}} + \int_{0}^{y - g(x, \psi)} 
\frac{\partial^2 f_{\epsilon}(\epsilon, x, \psi)}{\partial \psi \partial
\psi^{\prime}} \mathrm{d}\epsilon.
\end{align*}
When $y < g(x, \psi)$, $\frac{\partial^2 F(y|x, \psi)}{\partial \psi \partial \psi^{\prime}} = 0$. 

Second, Conditions (iii) and (iv) in Lemma \ref{lemma:Lipchitz-one-sided} and Condition (ii) in Lemma \ref{lemma:norm-diff-one-sided-general} imply that
$\sup_{y \ne g(x, \psi),x,\psi }\left\Vert \frac{\partial ^{2}F(y|x, \psi)}{\partial
\psi \partial \psi ^{\prime }}\right\Vert <\infty$. Therefore, we have the desired result by Lemma  \ref{lemma:sup-bound-hessian}.

\subsection{Proof of Lemma \protect\ref{lemma:small-oh-SC}}

We will investigate the properties of $\left\{\sup_{\psi} \frac{1}{T^2}
\sum_{t=1}^n \sum_{j=1}^T k(X_t, X_j; \psi) \right\}$ to verify Assumption %
\ref{assumption:convergence} (ii). Since $k(X_t, X_j; \psi)$ is symmetric
with respect to $X_t$ and $X_j$, we have 
\begin{align*}
\frac{1}{T^2} \sum_{t=1}^T \sum_{j=1}^T k(X_t, X_j; \psi) = \frac{1}{T^2}
\sum_{t=1}^T k(X_t, X_t; \psi) + \frac{2}{T^2} \sum_{1 \le t < j \le T}
k(X_t, X_j, \psi)
\end{align*}
When $w(s)$ is integrable, $k(X_t, X_t, \psi)$ is uniformly bounded by an
absolute positive constant, which implies $\frac{1}{T^2} \sum_{t=1}^T k(X_t,
X_t; \psi) = o(1)$. Therefore, it is enough to show 
\begin{align}
\sup_{\psi \in \Psi} \left| \frac{2}{T^2} \sum_{1 \le t < j \le T} k(X_t,
X_j, \psi) \right| = o_p(1).  \label{eq:u-process-1}
\end{align}

Since $w(s)$ is integrable, we can interchange the expectation and integral
when we evaluate $\mathbb{E}[ k(X_i, X_j; \psi)]$, and 
\begin{align*}
\mathbb{E}[ k(X_t, X_j; \psi)] = 
\begin{cases}
\mathbb{E}[ k(X_t, X_t; \psi)] & \text{ if } t = j, \\ 
0 & \text{ if } t \ne j.%
\end{cases}%
\end{align*}
Then, we can verify Equation (\ref{eq:u-process-1}) under the Lipschitz
continuity of $k(X_t, X_j, \psi)$ with respect to $\psi$.

\begin{lemma}
\label{lemma:convergence-SC-cond-k}
Suppose $\Psi$ is compact and $k(X_t, X_j, \psi)$ is Lipshitz in $\psi$ in
the sense that for any $\psi, \tilde{\psi} \in \Psi$, we have 
\begin{align*}
|k(X_t, X_j, \psi) - k(X_t, X_j, \psi_0)| \le M(X_t, X_j) \|\psi - \tilde{%
\psi}\|
\end{align*}
and $\int \int M(x, \tilde{x}) d P_x d P_{\tilde{x}} < \infty$ where $P_x =
P_{\tilde{x}}$. Then, 
\begin{align*}
\sup_{\psi} \left\|\frac{2}{T^2} \sum_{1 \le t < j \le n} k(X_t, X_j, \psi)
\right\| = \sup_{\psi} \left\|\frac{2}{T^2} \sum_{1 \le t < j \le T} k(X_t,
X_j, \psi) - \frac{2}{T^2} \sum_{1 \le t < j \le n} \mathbb{E}[k(X_t, X_j,
\psi)] \right\| \xrightarrow{p} 0.
\end{align*}
\end{lemma}

\begin{proof}
See Corollary 4.1 of \cite{Newey1991} or Lemma 4 in Appendix of \cite%
{Briol2019}.
\end{proof}
Under the Lipschitz continuity of $F(\cdot|x, \theta)$ in $\psi$, $k(\cdot, \dot, \psi)$ is also Lipchitz continuous in $\psi$ by Lemma \ref{lemma:Lipchitz-k}. This concludes the proof.

\subsection{Proof of Lemma \protect
\ref{lemma:proof-norm-diff-approx-Q-SC}}

Let us remind that 
\begin{equation*}
\widehat{R}_{T}(s;u,\psi ,\psi _{0}):=\widehat{Q}_{T}(s;u,\psi )- \widehat{Q}
_{T}(s;u,\psi _{0})-(\psi -\psi _{0})^{\top }\widehat{D} _{T}(s;u,\psi _{0}).
\end{equation*}
We would like to show 
\begin{align*}
\sup_{\psi \in \Psi ;\;\lVert \psi - \psi_{0} \rVert \leq \tau_{T}} \left| 
\frac{T \int_{\mathbb{S}^{d-1}}\int_{\mathcal{S}} \left( \widehat{R}%
_{T}(s;u,\psi,\psi _{0})\right) ^{2}w(s)\mathrm{d}s \mathrm{d}\varsigma (u)}{%
(1 + \lVert \sqrt{T} (\psi-\psi _{0})\rVert) ^{2}}\right| = o_{p}(1)
\end{align*}
for any $\tau _{T} \rightarrow 0$.

Here, we will show it with $\widehat{D}_T(s; u, \psi_0) = D(s; u, \psi_0)$
which are defined in Condition \ref{condition:norm-diff-true-Q}.

Note that 
\begin{align*}
& \widehat{Q}_{T}(s;u,\psi ) - \widehat{Q} _{T}(s;u,\psi _{0})-(\psi -\psi
_{0})^{\top }\widehat{D} _{T}(s;u,\psi _{0}) \\
& = \big[\widehat{Q}_{T}(s;u,\psi ) - Q(s;u,\psi )\big] - \big[\widehat{Q}%
_{T}(s;u,\psi_0) - Q(s;u,\psi _{0})\big] + \big[Q(s;u,\psi ) -
Q(s;u,\psi_{0}) -(\psi -\psi _{0})^{\top }\widehat{D} _{T}(s;u,\psi _{0})%
\big].
\end{align*}

Under Condition \ref{condition:norm-diff-true-Q}, it is enough to show 
\begin{align*}
\sup_{|\psi - \psi_0| \le \tau_T} \frac{T \left(\int_{u \in \mathbb{S}
^{d-1}} \int \left( \big[\widehat{Q}_{T}(s;u,\psi ) - Q(s;u,\psi )\big] - %
\big[\widehat{Q}_{T}(s;u,\psi_0) - Q(s;u,\psi _{0})\big] \right)^2 w(s) 
\mathrm{d}s \mathrm{d}\varsigma(u) \right)}{(1 + \|\sqrt{T}(\psi -
\psi_0)\|)^2} & = o_p(1),
\end{align*}
where 
\begin{align*}
\widehat{Q}_{T}(s; u, \psi)-Q(s;u, \psi) & = \frac{1}{T} \sum_{t=1}^T \left( 
\mathbb{E}[I(u^{\top} Z_t \le s)|X_t, \psi] - G(s; u, \psi)\right).
\end{align*}

Because $\frac{1}{1 + \|\sqrt{T}(\psi - \psi_0)\|^2} \le 1$, it is enough to
show 
\begin{align*}
\sup_{|\psi - \psi_0| \le \tau_T} T \left(\int_{u \in \mathbb{S}^{d-1}} \int
\left((\widehat{Q}_T(s; u, \psi) - Q(s; u, \psi)) - (\widehat{Q}_T(s; u,
\psi_0) - Q(s; u, \psi_0)) \right)^2 w(s) \mathrm{d}s \mathrm{d}\varsigma(u)
\right) = o_p(1).
\end{align*}

Let us denote 
\begin{align*}
V_{2, T}(\psi) & = \left(\int_{u \in \mathbb{S}^{d-1}} \int \left((\widehat{%
Q }_T(s; u, \psi) - Q(s; u, \psi)) - (\widehat{Q}_T(s; u, \psi_0) - Q(s; u,
\psi_0)) \right)^2 w(s) \mathrm{d}s \mathrm{d}\varsigma(u) \right).
\end{align*}
$V_{2, T}(\psi)$ is a V-statistic. That is, 
\begin{align*}
V_{2, T}(\psi) = \frac{1}{T^2} \sum_{t=1}^T \sum_{j=1}^{T} k_2(X_t, X_j,
\psi, \psi_0),
\end{align*}
where 
\begin{align*}
& k_2(X_t, X_j, \psi, \psi_0) \\
& \qquad = \int_{u \in \mathbb{S}^{d-1}} \int \Bigg\{ \bigg[\left(\mathbb{E}
[I(u^{\top} Z_t \le s)|X_t, \psi] - G(s; u, \psi)\right) - \left(\mathbb{E}
[I(u^{\top} Z_t \le s)|X_t, \psi_0] - G(s; u, \psi_0)\right)\bigg] \\
&\qquad  \quad \quad \times \bigg[\left(\mathbb{E}[I(u^{\top} Z_j \le s)|X_j, \psi]
- G(s; u, \psi)\right) - \left(\mathbb{E}[I(u^{\top} Z_j \le s)|X_j, \psi_0]
- G(s; u, \psi_0)\right)\bigg] \Bigg\} w(s) \mathrm{d}s \mathrm{d}
\varsigma(u)
\end{align*}

Note that $k_2(X_t, X_j, \psi, \psi_0)$ is symmetric kernel, and $V_{2,
T}(\psi_0) = 0$ since $k_2(X_t, X_j, \psi_0, \psi_0) = 0$.

Then, we need to handle $n V_{2, T}(\psi)$. It can be decomposed as follows. 
\begin{align*}
T V_{2, T}(\psi) & = \frac{1}{T} \sum_{t=1}^T \sum_{j=1}^T k_2(X_t, X_j;
\psi, \psi_0) \\
& = \frac{1}{T} \sum_{t=1}^T k_2(X_t, X_t; \psi, \psi_0) + \frac{2}{T}
\sum_{1 \le t < j \le T} k_2(X_t, X_j, \psi, \psi_0) \\
& = \frac{1}{T} \sum_{t=1}^T \Big(k_2(X_t, X_t; \psi, \psi_0) - \mathbb{E}%
[k_2(X_t, X_t; \psi, \psi_0)] \Big) + \mathbb{E}[k_2(X_t, X_t; \psi,
\psi_0)]  \\
& \quad + \frac{2}{T} \sum_{1 \le t < j \le T} k_2(X_t, X_j, \psi, \psi_0).
\end{align*}
We will show that each term in the above expressions is $o_p(1)$ as $\tau_T
\rightarrow 0$. Note that $k_2(\cdot, \cdot, \psi, \psi_0)$ is Lipchitz continuous in $\psi$ by Lemma \ref{lemma:Lipchitz-k}.

Step 1) Because $\Psi$ is compact and $k_2(X_t, X_t, \psi, \psi_0)$ is
Lipschitz continuous with respect to $\psi$ for each $X_t$, by ULLN, we have 
\begin{align*}
\sup_{ \psi \in \Psi} \left| \frac{1}{T} \sum_{i=1}^T \Big(k_2(X_t, X_t;
\psi) - \mathbb{E}[k_2(X_t, X_t; \psi, \psi_0)] \Big)\right| \xrightarrow{p}
0.
\end{align*}

Step 2) Because $k_2(X_i, X_i, \psi_0, \psi_0) = 0$, and $k_2(X_i, X_i,
\psi, \psi_0)$ is Lipschitz continuous for each $X_i$, we have 
\begin{align*}
\mathbb{E}[k_2(X_i, X_i, \psi, \psi_0)] \le \mathbb{E}[M_2(X_i, X_i)] \|\psi
- \psi_0\|
\end{align*}
and it implies that $\sup_{\psi \in \Psi ;\;\lVert \psi - \psi_{0} \rVert
\leq \tau_{n}} \mathbb{E}[k_2(X_i, X_i, \psi, \psi_0)] = o_p(1)$.

Step 3) For $t \ne j$, $\mathbb{E}[k_2(X_t, X_j, \psi, \psi_0)] = 0$ when $%
w(s)$ is integrable. Then, 
\begin{align*}
U_{2,T} := \frac{2}{T^2} \sum_{1 \le t < j \le T} k_2(X_t, X_j, \psi, \psi_0)
\end{align*}
is a degenerate U-statistic.

Suppose $\Psi$ is compact, and for any $\psi, \tilde{\psi} \in \Psi$, we
have 
\begin{align*}
| k_2(X_t, X_j, \psi, \psi_0) - k_2(X_i, X_j, \tilde{\psi}, \psi_0)| \le
M_2(X_i, X_j) \|\psi - \tilde{\psi}\|
\end{align*}
with $\iint M_2^2(x, y) d F_{X_t}(x) d F_{X_j}(y) < \infty$ and $\int
M_2^2(x, x) d F_{X_t}(x) < \infty$. Then, by the proof of Lemma 4 in the
appendix of \cite{Briol2019}, we have $k_2(X_i, X_i, \psi, \psi_0)$ is
Euclidean with envelope $F \le M(X_i, X_j) diam(\Psi)$ where $diam(\Psi) =
\sup_{\theta, \tilde{\theta} \in \Psi} \| \theta - \tilde{\theta}\|$ is the
diameter of $\Psi$ Then, 
\begin{align*}
\int \int |k_2^2(X_i, X_j)| d F_{X_i} d F_{X_j} \le \left(\int \int
M_2^2(X_i, X_j) d F_{X_i} d F_{X_j}\right) \| \psi - \psi_0 \|
\end{align*}
So, $\int \int |k_2^2(X_i, X_j)| d F_{X_i} d F_{X_j} \rightarrow 0$ as $%
|\psi - \psi_0| \rightarrow 0$.

By Corollary 8 in \cite{Sherman1994}, we have $\sup_{\psi \in \Psi ;\;\lVert
\psi - \psi_{0} \rVert \leq \tau_{T}} |T U_{2, T}| = o_p(1)$.

Therefore, 
\begin{align*}
\sup_{\psi \in \|\psi - \psi_0\|_2 \le \tau_T } | T V_{2, T}(\psi) | =
\sup_{\psi \in \|\psi - \psi_0\|_2 \le \tau_T } | T V_{2, T}(\psi) - T V_{2,
T}(\psi_0) | = o_p(1).
\end{align*}

\end{appendices}

\end{document}